\documentclass[runningheads]{llncs}
\usepackage{amsmath}%
\usepackage{pifont}%
\usepackage{amssymb}%
\usepackage{stmaryrd}%
\usepackage{array}%
\usepackage{epic}%
\usepackage[fixamsmath]{mathtools}
\usepackage{cancel}
\usepackage{url}
\usepackage{bigstrut}
\usepackage{wrapfig}
\usepackage{ulem}%

\DeclareMathAlphabet\mathbb{U}{msb}{m}{n}
\DeclareMathAlphabet\mathbfcal{OMS}{cmsy}{b}{n}

\DeclareMathAlphabet{\mathrm}    {OT1}{cmr}{m}{n}
\DeclareMathAlphabet{\mathrmbf}  {OT1}{cmr}{bx}{n}
\DeclareMathAlphabet{\mathrmit}  {OT1}{cmr}{m}{it}
\DeclareMathAlphabet{\mathrmbfit}{OT1}{cmr}{bx}{it}
\DeclareMathAlphabet{\mathsf}    {OT1}{cmss}{m}{n}
\DeclareMathAlphabet{\mathsfbf}  {OT1}{cmss}{bx}{n}
\DeclareMathAlphabet{\mathsfit}  {OT1}{cmss}{m}{sl}
\DeclareMathAlphabet{\mathtt}    {OT1}{cmtt}{m}{n}
\DeclareMathAlphabet{\mathttbf}  {OT1}{cmtt}{bx}{n}
\DeclareMathAlphabet{\mathttit}  {OT1}{cmtt}{m}{it}
\DeclareMathAlphabet{\mathpzc}   {OT1}{pzc}{m}{it}
\newcommand{\keywords}[1]{\par\addvspace\baselineskip\noindent\enspace\ignorespaces{\bfseries Keywords:\,}#1}
\newcommand{\comment}[1]{}
\newenvironment{aside}
{\begin{flushleft}\begin{minipage}{340pt}}
{\end{minipage}\end{flushleft}}

\makeatletter
\def\dual#1{\expandafter\dual@aux#1\@nil}
\def\dual@aux#1/#2\@nil{\begin{tabular}{@{}c@{}}#1\\#2\end{tabular}}
\makeatother
\begin{document}

\pagestyle{headings}
\title{{\ttfamily FOLE} Equivalence} 
\titlerunning{{\ttfamily FOLE} Equivalence}  
\author{Robert E. Kent}
\institute{Ontologos}
\maketitle

\begin{abstract}
The first-order logical environment {\ttfamily FOLE} 
provides a rigorous and principled approach to 
distributed interoperable first-order information systems.
{\ttfamily FOLE} has been developed in two forms: 
a classification form and an interpretation form.
%
Two papers represent {\ttfamily FOLE} in a {\it{classification form}}
corresponding to ideas of the Information Flow Framework 
(\cite{kent:sem:integ:iff:iswc2003},\cite{kent:sem:integ:iff:dsp04391},\cite{iff}):
the first paper \cite{kent:fole:era:found} provides a foundation 
that connects elements of the {\ttfamily ERA} data model \cite{chen:76}
with components of the first-order logical environment {\ttfamily FOLE}; 
the second paper \cite{kent:fole:era:supstruc} provides a superstructure 
that extends {\ttfamily FOLE} to the formalisms of first-order logic.
The formalisms in the classification form of {\ttfamily FOLE} 
provide an appropriate framework for developing the {\tt{relational calculus}}.
%
Two other papers 
represent {\ttfamily FOLE} in an {\it{interpretation form}}:
the first paper \cite{kent:fole:era:tbl}
develops the notion of the \texttt{FOLE} table
following the relational model \cite{codd:90};
the second paper \cite{kent:fole:era:db}
discusses the notion of a \texttt{FOLE} relational database.
All the operations of the {\sf{relational algebra}} 
have been rigorously developed 
\cite{kent:fole:rel:ops}
using the interpretation form of \texttt{FOLE}.
%
The present study demonstrates that
the classification form of {\ttfamily FOLE}
is informationally equivalent to 
the interpretation form of {\ttfamily FOLE}.
In general,
the {\ttfamily FOLE} representation uses a conceptual structures approach,
that is completely compatible with the theory of institutions, formal concept analysis and information flow.
\end{abstract}

\tableofcontents


\newpage
\section{Preface.}\label{sec:preface}

The architecture for 
the first-order logical environment
{\ttfamily FOLE}
is displayed in Fig.~\ref{fig:big:pix}.
%
\footnote{
{\ttfamily FOLE} is described in the following papers:
``The {\ttfamily ERA} of {\ttfamily FOLE}: Foundation''     \cite{kent:fole:era:found},
``The {\ttfamily ERA} of {\ttfamily FOLE}: Superstructure'' \cite{kent:fole:era:supstruc},
``The {\ttfamily FOLE} Table''                              \cite{kent:fole:era:tbl},
``The {\ttfamily FOLE} Database''                           \cite{kent:fole:era:db}, 
``{\ttfamily FOLE} Equivalence''                            [this paper], and
``Relational Operations in \texttt{FOLE}''                  \cite{kent:fole:rel:ops}
.}
%
The left side of Fig.~\ref{fig:big:pix} represents {\ttfamily FOLE} in classification form.
The right side of Fig.~\ref{fig:big:pix} represents {\ttfamily FOLE} in interpretation form.

\begin{figure}
\begin{center}
{{\begin{tabular}{c}
\setlength{\unitlength}{0.47pt}
\begin{picture}(160,160)(0,-40)
%
\put(-35,80){\makebox(0,0){\footnotesize{\bf{Sound Logic}}}}
\put(-30,0){\makebox(0,0){\footnotesize{\bf{Structure}}}}
\put(-30,-40){\makebox(0,0){
$\underset{\text{\footnotesize\it{classification form}}}{\underbrace{\text{\hspace{60pt}}}}$}}
\put(180,-40){\makebox(0,0){
$\underset{\text{\footnotesize\it{interpretation form}}}{\underbrace{\text{\hspace{60pt}}}}$}}
\put(190,80){\makebox(0,0){\footnotesize{\bf{Database}}}}
\put(190,0){\makebox(0,0){\footnotesize{\bf{Table}}}}
\put(-188,83){\makebox(0,0){\scriptsize{\tt{Relational}}}}
\put(-188,63){\makebox(0,0){\scriptsize{\tt{Calculus}}}}
\put(310,15){\makebox(0,0){\scriptsize{\sf{Relational}}}}
\put(310,-5){\makebox(0,0){\scriptsize{\sf{Operations}}}}
%
\qbezier(20,95)(80,115)(140,95)\put(140,95){\vector(3,-1){0}}
\put(84,80){\makebox(0,0){\scriptsize{$\equiv$}}}
\qbezier(20,65)(80,45)(140,65)\put(20,65){\vector(-3,1){0}}
%
\put(225,12){\oval(30,30)[tl]}
\put(210,7){\vector(0,-1){0}}
\put(225,27){\line(1,0){10}}
\put(235,7){\oval(40,40)[r]}
\put(225,-13){\line(1,0){10}}
\put(-105,65){\oval(30,30)[br]}
\put(-90,69){\vector(0,1){0}}
\put(-115,70){\oval(40,40)[l]}
\put(-115,90){\line(1,0){10}}
\put(-115,50){\line(1,0){10}}
%
%
%
\put(-29,15){\line(0,1){45}}
\put(-30,15){\line(0,1){45}}
\put(-31,15){\line(0,1){45}}
\put(189,15){\line(0,1){45}}
\put(190,15){\line(0,1){45}}
\put(191,15){\line(0,1){45}}
\end{picture}
\end{tabular}}}
\end{center}
\caption{{\footnotesize{{\ttfamily FOLE} Architecture}}}
\label{fig:big:pix}
\end{figure}
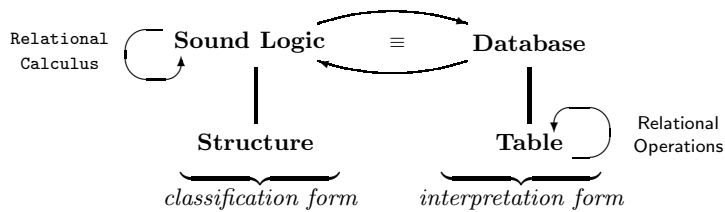
%



\paragraph{Classification form.}\label{par:class:form}

An ontology defines the primitives with which to model the knowledge resources for a community of discourse
(Gruber~\cite{gruber:eds2009}).
These primitives, which consist of classes, relationships and properties,
are represented by the \emph{entity-relationship-attribute} {\ttfamily ERA} model
(Chen~\cite{chen:76}). 
An ontology uses formal axioms to constrain the interpretation of these primitives. 
In short, an ontology specifies a logical theory.
\comment{
\begin{flushleft}
{\fbox{\fbox{\footnotesize{\begin{minipage}{330pt}
{\underline{\textsf{What are the elements of an ontology?}}}
An ontology is a model that clarifies and specifies a set of meanings in a formal language. 
Those meanings reflect the oncologist's understanding of the target subject matter, 
regarding the kinds of things that exist
and how those things are related to each other.
A taxonomy has a backbone with multiple link types each having a precise meaning.
%
\comment{
\newline
{\underline{\textsf{What are the elements of a database schema?}}}
A database schema defines the structure of a database in a formal language. 
There are three kinds:
conceptual, logical and physical. 
While we recognize that there are important differences, 
for convenience, 
we will conform to common usage and use the term ''database schema'' loosely to refer
to all of these.
}
\end{minipage}}}}}
\end{flushleft}
}

Two papers
provide a rigorous mathematical representation
for the {\ttfamily ERA} data model in particular,
and ontologies in general,
within the first-order logical environment {\ttfamily FOLE}.
%
These papers
represent the formalism and semantics of (many-sorted) first-order logic in a classification form
corresponding to ideas discussed in the Information Flow Framework 
(IFF~\cite{iff}).
The paper (Kent~\cite{kent:fole:era:found}) 
develops the notion of a \texttt{FOLE} \underline{\emph{structure}};
this
provides a {\sl foundation} 
that connects elements of the {\ttfamily ERA} data model with components of the first-order logical environment {\ttfamily FOLE}.
The paper (Kent \cite{kent:fole:era:supstruc}) 
develops the notion of a \texttt{FOLE} \underline{\emph{sound logic}};
this
provides a {\sl superstructure} 
that extends {\ttfamily FOLE} to the formalisms of first-order logic.
%

%


\newpage
\paragraph{Interpretation form.}\label{par:interpret:form}
%
The \emph{relational model} (Codd~\cite{codd:90}) 
is an approach for the information management of a ``community of discourse''
\footnote{Examples include:
an academic discipline;
a commercial enterprise;
the genetics research community,
library science;
the legal profession;
etc.}
using the semantics and formalism of (many-sorted) first-order predicate logic. 
The \emph{first-order logical environment} 
\texttt{FOLE}
(Kent~\cite{kent:iccs2013})
is a category-theoretic representation for this logic.
%
\footnote{Following the original discussion of {\ttfamily FOLE} (Kent~\cite{kent:iccs2013}), 
we use 
the term \emph{mathematical context} for the concept of a category,
the term \emph{passage} for the concept of a functor, and
the term \emph{bridge} for the concept of a natural transformation.
A context represents some ``species of mathematical structure''. 
A passage is a ``natural construction on structures of one species, 
yielding structures of another species'' 
(Goguen \cite{goguen:cm91}).}
%
Hence, the relational model can naturally be represented in \texttt{FOLE}.
A series of papers is being developed 
to provide a rigorous mathematical basis for {\ttfamily FOLE} 
by defining 
an architectural semantics for the relational data model.
%

Two papers provide a precise mathematical basis for \texttt{FOLE} interpretation.
Both of these papers expand on material found in the paper (Kent~\cite{kent:db:sem}).
The paper 
(Kent \cite{kent:fole:era:tbl})
develops the notion of a \texttt{FOLE} \underline{\emph{table}}
following the relational model (Codd~\cite{codd:90}).
The paper 
(Kent \cite{kent:fole:era:db})
develops the notion of a \texttt{FOLE} relational \underline{\emph{database}},
thus providing the foundation for
the 
semantics of first-order logical/relational database systems.
%


\paragraph{Equivalence.}\label{par:equiv}

%
%
The current paper 
``{\ttfamily FOLE} Equivalence''
defines an interpretation 
of {\ttfamily FOLE} in terms of the transformational passage,
first described in (Kent~\cite{kent:iccs2013}),
from the classification form of first-order logic
to an \underline{\emph{equivalent}} interpretation form,
thereby defining the formalism and semantics of first-order logical/relational database systems. 
%
Although 
the {\it{classification form}} 
follows the entity-relationship-attribute data model of Chen~\cite{chen:76},
the {\it{interpretation form}} incorporates the relational data model of Codd~\cite{codd:90}.
%
%
Relational operations are described in the paper 
(Kent \cite{kent:fole:rel:ops}).
The relational calculus will be discussed in a future paper.

%
\newpage
\section{Introduction.}\label{sec:intro}
\begin{flushleft}
{{\setlength{\extrarowheight}{1.6pt}
{{
{\begin{tabular}[t]{l@{\hspace{20pt}}l}
{{{\begin{minipage}{230pt}
This paper defines 
the equivalence between two forms of 
the first-order logical
environment {\ttfamily FOLE},
the classification form 
and
the interpretation form. 
Equivalence sits between 
the top of both forms,
pictured briefly on the right,
and more completely in Fig.\ref{fig:big:pix}
in the preface \S\ref{sec:preface}.
The classification form hierarchy 
(left hand side)
consists of 
``The {\ttfamily FOLE} Foundation''
at the bottom
and
``The {\ttfamily FOLE} Superstructure''
at the top.
The interpretation form hierarchy 
(right hand side)
consists of 
``The {\ttfamily FOLE} Table''
at the bottom
and
``The {\ttfamily FOLE} Database''
at the top.
%
\end{minipage}}}}
&
{{\begin{tabular}{c@{\hspace{5pt}}}
\setlength{\unitlength}{0.36pt}
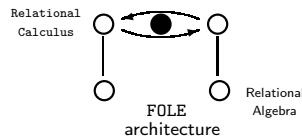
\begin{picture}(180,120)(-50,-20)
\put(60.5,80.5){\makebox(0,0){\tiny{$\equiv$}}}
\put(60.5,79.8){\makebox(0,0){\huge{$\bullet$}}}
%
\put(-60,93){\makebox(0,0){\tiny{\tt{Relational}}}}
\put(-60,73){\makebox(0,0){\tiny{\tt{Calculus}}}}
\put(180,10){\makebox(0,0){\tiny{\sf{Relational}}}}
\put(180,-10){\makebox(0,0){\tiny{\sf{Algebra}}}}
%
\qbezier(100,87)(60,102)(20,87)
\put(20,87){\vector(-4,-1){0}}
\qbezier(20,75)(60,60)(100,75)
\put(100,75){\vector(4,1){0}}
\put(0.3,80){\makebox(0,0){\huge{$\circ$}}}
\put(120.3,80){\makebox(0,0){\huge{$\circ$}}}
\put(1,10){\makebox(0,0){\huge{$\circ$}}}
\put(122,10){\makebox(0,0){\huge{$\circ$}}}
%
\put(0,68){\line(0,-1){40}}
\put(120,68){\line(0,-1){40}}
\put(46,-15){\scriptsize{\ttfamily FOLE}}
\put(22,-35){\scriptsize{\textsf{architecture}}}
\end{picture}
\end{tabular}}}
\end{tabular}}
}}}}
\end{flushleft}
%



\subsection{First Order Logical Environment}\label{subsec:FOLE}

%
This paper relates two forms of the first order logic environment {\ttfamily FOLE}: 
the classification form 
and
the interpretation form.
These two forms are shown to be ``informationally equivalent''
%
\footnote{See Prop.\,12 in \S\,A.1 of the paper
\cite{kent:fole:era:tbl}
``The {\ttfamily FOLE} Table''.}
%
to each other. 
Both forms have their own advantages:
the classification form allows formalism 
%
\footnote{See the paper \cite{kent:fole:era:supstruc}
``The {\ttfamily ERA} of {\ttfamily FOLE}: Superstructure''.}
%
to be easily defined,
whereas
the interpretation form allows relational operations 
%
\footnote{See the paper \cite{kent:fole:rel:ops}
``Relational Operations in {\ttfamily FOLE}''.}
to be easily defined.
The classification form of {\ttfamily FOLE} is realized in the notion of a (lax) sound logic,
whereas 
the interpretation form of {\ttfamily FOLE} is realized in the notion of a relational database.
To demonstrate the equivalence of 
the classification form
with
the interpretation form,
(lax) sound logics 
and relational databases
are shown
to be in a reflective relationship
(Fig.\,\ref{fig:equiv}).
%
\comment{See (Key) Prop.\,\ref{sat:tbl:interp}
in \S\,\ref{sub:sec:spec:sat}.}
%
Although a {\ttfamily FOLE} relational database can be taken as a whole,
a {\ttfamily FOLE} sound logic resolves into the two parts of
structure and specification,
plus the relationship of satisfaction.
%
\footnote{The satisfaction relation 
corresponds to the ``truth classification'' in Barwise and Seligman~\cite{barwise:seligman:97},
where the conceptual intent $\mathcal{M}^{\mathcal{S}}$
corresponds to the ``theory of $\mathcal{M}$''.}
%
%
%

In more detail
(Tbl.\,\ref{tbl:equiv:brief}),
the classification form of \texttt{FOLE} is represented by a (lax) sound logic
consisting of two components,
a (lax) structure and an abstract specification,
which are connected by a satisfaction relation.
The (lax) structure consists of 
a (lax) entity classification, 
an attribute classification (typed domain),
a schema, and 
one of the two equivalent descriptions:
either a tabular map
from predicates to
tables;
or
a bridge consisting of an indexed collection
of table tuple functions.
The abstract specification is the same as a the database schema.
Satisfaction is defined by
table interpretation using constraints.
%
The interpretation form of \texttt{FOLE} is represented by 
a relational database with constant type domain,
a tabular interpretation diagram,
whose projective components
consist of
a database schema,
a key diagram,
and
a tuple bridge.
The tabular interpretation diagram
is equivalent to the satisfaction relation between 
the (lax) structure 
and
the abstract specification. 
\footnote{See (Key) Prop.\,\ref{sat:tbl:interp}
in \S\,\ref{sub:sec:spec:sat}.}

\begin{figure}
\begin{center}
{{\begin{tabular}{c}
\setlength{\unitlength}{0.75pt}
\begin{picture}(180,40)(-5,0)
\put(0,1){\makebox(0,0){\normalsize{$\mathrmbf{Snd}$}}}
\put(80,2){\makebox(0,0){\normalsize{$\mathring{\mathrmbf{Snd}}$}}}
\put(160,0){\makebox(0,0){\normalsize{$\mathrmbf{Db}$}}}
\put(40,10){\makebox(0,0){\footnotesize{$\mathrmbfit{lax}$}}}
\put(120,12){\makebox(0,0){\footnotesize{$\mathrmbfit{im}$}}}
\put(120,0){\makebox(0,0){\tiny{$\dashv$}}}
\put(120,-12){\makebox(0,0){\footnotesize{$\mathrmbfit{inc}$}}}
%
\put(25,0){\vector(1,0){30}}
\put(100,-5){\vector(1,0){40}}
\put(140,5){\vector(-1,0){40}}
%
\end{picture}
\end{tabular}}}
\end{center}
\caption{{\ttfamily FOLE} Equivalence}
\label{fig:equiv}
\end{figure}

\comment{\begin{table}
\begin{center}
{\footnotesize{\setlength{\extrarowheight}{2pt}
{{\begin{tabular}{@{\hspace{5pt}}
r
@{\hspace{10pt}
$\xrightarrow[\text{relativization}]{\text{key-set}}$
\hspace{10pt}}
c
@{\hspace{10pt}
$\underset{\textstyle{\xrightarrow[\text{satisfaction}]{}}}{\xleftarrow{\text{components}}}$
\hspace{10pt}}
l
@{\hspace{5pt}}}
$\underset{(classification)}{\text{sound logic}}$
& 
$\underset{(interpretation)}{\text{sound logic}}$  
& 
database
\end{tabular}}}}}
\end{center}
\caption{Overview:  {\ttfamily FOLE} Equivalence}
\label{tbl:overview}
\end{table}}
%

\subsection{Overview.}\label{sub:sec:overview}


%
\begin{table}
\begin{center}
{\footnotesize{
\begin{tabular}{@{\hspace{5pt}}p{300pt}@{\hspace{5pt}}} 
\begin{description}
\item[$\mathring{\mathrmbf{Snd}}\xleftarrow{\mathrmbfit{im}}\mathrmbf{Db}$:] 
A database 
resolves into
a sound logic:
the specification is the database schema projection;
the (lax) structure is the constraint-free aspect of a database; and
tabular interpretation
defines satisfaction.
\item[$\mathring{\mathrmbf{Snd}}\xrightarrow{\mathrmbfit{inc}}\mathrmbf{Db}$:] 
A sound logic 
assembles 
a database:
the specification provides (is the same as) the database schema;
the (lax) structure provides a constraint-free interpretation
consisting of a collection of tables indexed by predicates; and
satisfaction defines the tabular interpretation.
%
\end{description}
\end{tabular}
}}
\end{center}
\caption{\texttt{FOLE} Equivalence in Brief} 
\label{tbl:equiv:brief}
\end{table}
Briefly,
\S\,\ref{sec:log:environ}
 defines some basic concepts of \texttt{FOLE}
(formula, interpretation and satisfaction);
\S\,\ref{sec:struc}
defines structures and converts them to a lax variety; 
\S\,\ref{sec:spec}
defines formal and abstract specifications,
and satisfaction of specifications by structures; 
\S\,\ref{sec:log}
defines (lax) sound logics
and converts these to 
relational databases; and
\S\,\ref{sec:db}
defines relational databases
and converts these to 
(lax) sound logics.
In overview,
Table \ref{tbl:figs:tbls} lists the figures and tables used in this paper.


\paragraph{Basic Concepts.}
%
\S\,\ref{sec:log:environ}
reviews some basic concepts 
(Tbl.\,\ref{tbl:fole:refs})
of the {\ttfamily FOLE} logical environment
from \S\,2 of the paper \cite{kent:fole:era:supstruc}
``Superstructure''.
This covers the concepts of formalism, tabular interpretation, and satisfaction.
\S\,\ref{sec:fmlism}
recapitulates the definition of formalism in 
\cite{kent:fole:era:supstruc}.
Formalism is defined in terms of the logical Boolean operations
and quantifiers
using syntactic flow
(Tbl.\,\ref{tbl:syn:flow}).  
\S\,\ref{sub:sec:interp:fml:constr}
extends to tables the definition of logical interpretation in 
\cite{kent:fole:era:supstruc}
using Boolean operations within fibers and quantifier flow between fibers
(Tbl.\,\ref{tbl:fml:int}).
It represents 
the syntactic flow operators of Tbl.\,\ref{tbl:syn:flow}
by their associated semantic flow operators
(Tbl.\,\ref{tbl:fml-sem:refl}).
\S\,\ref{sub:sec:sat} 
reviews the notion of satisfaction  in 
\cite{kent:fole:era:supstruc}, 
which links formalism and semantics via satisfaction. 
Satisfaction is defined for both sequents and constraints.
Satisfaction is reflective 
(Fig.\,\ref{fig:rel:tbl:refl})
between relations and tables.


\paragraph{Structures.}
\S\,\ref{sec:struc}
reviews and extends the basic concepts of {\ttfamily FOLE} structure
 and structure morphism
from \S\,4.3 of the paper \cite{kent:fole:era:found}``Foundation''.
In \S\,\ref{sub:sec:struc}, 
Def.\,\ref{def:struc}
defines a {\ttfamily FOLE} structure
with both a classification and interpretation form.
Fig.\,\ref{fig:fole:struc} illustrates this idea.
Note the global key sets 
used in the entity classification of this definition.
In \S\,\ref{sub:sec:struc:mor},
Def.\,\ref{def:struc:mor}
defines an extended notion of {\ttfamily FOLE} structure morphism
by adding internal bridges.
Fig.\,\ref{fig:fole:struc:mor} illustrates this idea.
Note the global key function 
in the entity infomorphism of this definition.
\S\,\ref{sub:sec:trans} defines a transition 
from strict to lax structures and structure morphisms.
This is accomplished by changing 
the entity classification and entity infomorphism to lax versions,
thus changing
the global key sets and key function
to local versions.
Note\;\ref{note:lax:ent:info}, 
Prop.\;\ref{prop:key} (Key) and Cor.\;\ref{cor:reduct:tup}
define a step-by-step process for this transition.
Fig.\,\ref{fig:transition}
illustrates the steps of this transition.
\S\,\ref{sub:sec:struc:lax}
and
\S\,\ref{sub:sec:struc:mor:lax}
introduce the idea of  
lax structures
and
lax structure morphisms.
We generalize to lax structures and lax structure morphisms 
by eliminating the global key sets and key functions.
Def.\,\ref{def:struc:lax} 
defines a lax {\ttfamily FOLE} structure
in terms of
either
a tabular interpretation function
or a tuple bridge.
Prop.\,\ref{prop:struc:2:|db|}
shows that a lax structure is the same as the constraint-free aspect of a database.
%
Def.\,\ref{def:lax:struc:mor}
defines a (lax) structure morphism.
Cor.\,\ref{cor:tbl:func}
shows that a (lax) structure morphism
defines a tabular interpretation bridge function.
Fig.\,\ref{fig:lax:struc:mor}
illustrates a (lax) structure morphism. 

\paragraph{Specifications.}
%
\S\,\ref{sec:spec}
reviews and extends the notion of a {\ttfamily FOLE} specification.
Formal {\ttfamily FOLE} specifications are covered 
in \S\,3.1 of the paper \cite{kent:fole:era:supstruc}
``Superstructure''.
Here we extend to abstract {\ttfamily FOLE} specifications.
\S\,\ref{sub:sec:spec}
defines 
a formal specification in Def.\,\ref{def:fml:spec}
and
an abstract specification in Def.\,\ref{def:abs:spec}.
Every abstract specification is closely linked to a companion formal specification.
\S\,\ref{sub:sec:spec:mor}
defines a abstract specification morphism in Def.\,\ref{def:abs:spec:mor}
and illustrates this in Fig.\,\ref{fig:spec:mor:list}.
Abstract specifications and morphisms 
are the same as 
database schemas and morphisms.
\S\,\ref{sub:sec:spec:sat}
defines specification satisfaction.
Def.\,\ref{satis:fml:spec}
defines satisfaction for formal specifications, whereas
Def.\,\ref{satis:abs:spec}
defines satisfaction for abstract specifications 
in terms of satisfaction for rheir formal companion.
Assuming satisfaction holds,
Def.\,\ref{def:abs:tbl:pass}
defines the abstract table passage
as the composition
of 
the object-identical companion passage,
the inclusion into the structure conceptual intent,
and
the 
relation interpretation passage
of Lem.\,\ref{lem:interp:pass}
in
\S\,\ref{sub:sec:sat}.
Prop.\,\ref{sat:tbl:interp} is \textbf{key}: 
it proves that 
satisfaction is equivalent to tabular interpretation,
This is central,
since it binds together the two parts of a sound logic,
its structure and its specification.
It
underpins the core idea that sound logics are equivalent to relational databases.
Fig.\,\ref{fig:spec:pass} 
illustrates and compares 
specification passages,
both in general and with satisfaction.
Tbl.\,\ref{tbl:spec:compare} 
illustrates and compares 
specifications with the sound logics
defined in \S\,\ref{sec:log}.

\paragraph{Sound Logics.}
%
\S\,\ref{sec:log}
reviews the notion of a {\ttfamily FOLE} sound logic.
\S\,\ref{sub:sub:sec:log}
discusses sound logics.
Def.\,\ref{def:snd:log}
defines (lax) sound logics
in terms of satisfaction between its two components: 
structure and abstract specification.
Prop.\,\ref{prop:snd:log:2:db}
proves that any (lax) sound logic
defines a relational database.
Tbl.\,\ref{tbl:snd:log}
lists and illustrates the various components of a (lax) sound logic.
\S\,\ref{sub:sub:sec:log:mor}
discusses sound logic morphisms.
Def.\,\ref{def:snd:log:mor}
defines the notion of a (lax) sound logic morphism.
Prop.\,\ref{prop:sat:preserve}
shows how (lax) sound logic morphisms
preserve and link
satisfaction
between their source and target logics.
The large figure Fig.\,\ref{fig:nat:comb}
expands in detail the naturality used in
this proposition.
Prop.\,\ref{prop:snd:log:mor:2:db:mor}
proves that a (lax) sound logic morphism
defines a database morphism.
Thm.\,\ref{thm:snd:log:2:db}
proves existence of a passage
from the context of (lax) \texttt{FOLE} sound logics 
  to the context of \texttt{FOLE} relational databases.
%
%
Tbl.\,\ref{tbl:snd:log:mor}
illustrates the components of a (lax) sound logic morphism.

\paragraph{Relational Databases.}
%
\S\,\ref{sec:db}
reviews the notion of a {\ttfamily FOLE} relational database.
\S\,\ref{sub:sec:db:obj} 
discusses {\ttfamily FOLE} relational databases.
%
Def.\,\ref{def:db}
defines a relational database
as 
an interpretation diagram
from predicates 
to tables for a fixed type domain.
Fig.\,\ref{fig:db} illustrates this definition.
Def.\,\ref{def:db:proj}
defines relational databases
using table projection passages.
%
Tbl.\,\ref{tbl:db}
lists the components of a \texttt{FOLE} relational database.
Prop.\,\ref{prop:db:to:struc}
shows that the constraint-free aspect of a
database
is the same as 
a (lax) structure.
%
Prop.\,\ref{prop:db:2:snd:log}
proves that a \texttt{FOLE} database defines a (lax) \texttt{FOLE} sound logic.
\S\,\ref{sub:sec:db:mor}
discusses \texttt{FOLE} relational database morphisms.
Def.\,\ref{def:db:mor}
defines a database morphism
to consist of 
a tabular passage 
between source/target predicate contexts, 
an infomorphism 
between source/target type domains, and 
a bridge connecting the
source/target tabular interpretations.
Fig.\,\ref{fig:db:mor} 
illustrates this definition.
%
Def.\,\ref{def:db:mor:proj}
defines relational database morphisms
using table projection passages.
%
Tbl.\,\ref{tbl:db:mor}
lists 
%
and
Fig.\,\ref{fig:db:mor:typ:dom}
illustrates the components of a \texttt{FOLE} relational database.
%
%
Prop.\,\ref{prop:db:mor:to:struc:mor}
shows that the constraint-free aspect of a
 \texttt{FOLE} database morphism
is the same as 
a (lax) \texttt{FOLE} structure morphism.
%
%
Prop.\,\ref{prop:db:mor:2:snd:log:mor}
proves that a \texttt{FOLE} database morphism defines a (lax) \texttt{FOLE} sound logic morphism.
%
%
Thm.\,\ref{thm:db:2:snd:log}
proves existence of a passage
from the context of \texttt{FOLE} relational databases
to the context of (lax) \texttt{FOLE} sound logics.
Thm.\,\ref{thm:db:iso:snd:log}
proves that the contexts of 
\texttt{FOLE} relational databases
and
\texttt{FOLE} (lax) sound logics are ``informationally equivelent''
by way of a reflection.

\paragraph{Appendix.}
\S\,\ref{sub:sec:adj:components}
lists the \texttt{FOLE} components used in this paper.
Tbl.\,\ref{tbl:adjoints}
in \S\,\ref{sub:sec:adj:components}
lists the equivalent and adjoint
versions of \texttt{FOLE} bridges.
Tbl.\,\ref{tbl:fole:morph}
in \S\,\ref{sub:sec:adj:components}
lists the \texttt{FOLE} Morphisms.
\S\,\ref{sub:sec:class:info}
reviews the concepts of classifications and infomorphisms.
Both are extended to lax versions.
%



\comment{
Sec.2 reviews the basic concepts of the FOLE logical environment
from \S.2 of the paper superstructure.
***
Sec.2 covers the basic concepts of formalism, tabular interpretation, and satisfaction.
***
Sec.2.1 recapitulates the definition of formalism in \S.2.1.1 of the paper superstructure. Here we define formulas in terms of the logical operations.
***
Sec.2.2 extends to tables the definition of logical interpretation in \S.2.2.1 of the paper superstructure using Boolean operations within fibers and logical quantifier flow between fibers.
***
Sec.2.3 reviews the notion of satisfaction  in \S.2.3 of the paper superstructure
links formalism and semantics via sataisfaction. Here we define satasfaction for sequents, constraints and (formal) specifications.
}  


%


%
\begin{table}
\begin{center}
{{\scriptsize{\setlength{\extrarowheight}{1.6pt}

{\begin{tabular}{l@{\hspace{10pt}}l}
\\
{\fbox{
\begin{tabular}[t]{|l@{\hspace{2pt}}l@{\hspace{2pt}:\hspace{8pt}}l|}
\hline
\S\ref{sec:preface}
&
Fig.\,\ref{fig:big:pix}
&
\texttt{FOLE} Architecture
\\\hline\hline
\S\ref{sec:intro} 
& Fig.\,\ref{fig:equiv} & \texttt{FOLE} Equivalence
\\\hline\hline
\S\ref{sec:log:environ} & Fig.\,\ref{fig:rel:tbl:refl} &  Table-Relation Reflection
\\\hline\hline
\S\ref{sec:struc} & Fig.\,\ref{fig:fole:struc} & \texttt{FOLE} Structure
\\\cline{2-2}
 & Fig.\,\ref{fig:fole:struc:mor} & Structure Morphism
\\\cline{2-2}
 & Fig.\,\ref{fig:transition} & Transition
\\\cline{2-2}
 & Fig.\,\ref{fig:lax:struc:mor} & Lax Structure Morphism
\\\hline\hline
\S\ref{sec:spec} & Fig.\,\ref{fig:spec:mor:list} & Specification Morphisms
\\\cline{2-2}
 & Fig.\,\ref{fig:spec:pass} & Specification Passages
\\\hline\hline
\S\ref{sec:log} & Fig.\,\ref{fig:nat:comb}
& Naturality Combined
\\\hline\hline
\S\ref{sec:db} & Fig.\,\ref{fig:db} & Database: Type Domain
\\\cline{2-2} & Fig.\,\ref{fig:db:mor} & Database Morphism
\\\cline{2-2} & Fig.\,\ref{fig:incl:bridge:tbl} & Inclusion Bridge: Tables
\\\cline{2-2} & Fig.\,\ref{fig:db:mor:typ:dom} & Database Morphism: Type Domain
\\\hline\hline
\S\ref{sec:append} & Fig.\,\ref{fig:info:mor} & Infomorphism
\\ & Fig.\,\ref{fig:lax:info:mor} & Lax Infomorphism
\\\hline
\end{tabular}}}
&
{\fbox{
\begin{tabular}[t]{|l@{\hspace{2pt}}l@{\hspace{2pt}:\hspace{8pt}}l|}
\hline
\S\,\ref{sec:intro} 
& Tbl.\,\ref{tbl:equiv:brief}     & \texttt{FOLE} Equivalence in Brief
\\\cline{2-2}
 & Tbl.\,\ref{tbl:figs:tbls} & Figures and Tables
\\\hline\hline
\S\,\ref{sec:log:environ}
   & Tbl.\,\ref{tbl:fole:refs}    & Basic Concepts 
\\\cline{2-2}
 & Tbl.\,\ref{tbl:syn:flow}     & Syntactic Flow
\\\cline{2-2}
 & Tbl.\,\ref{tbl:fml:int}      & Formula Interpretation 
\\\cline{2-2}
 & Tbl.\,\ref{tbl:fml-sem:refl} & Formal/Semantics Reflection 
\\\hline\hline
\S\,\ref{sec:spec}
   & Tbl.\,\ref{tbl:spec:compare} & Comparisons 
\\\hline
\S\,\ref{sec:log}
    & Tbl.\,\ref{tbl:snd:log}     & Sound Logic 
\\  & Tbl.\,\ref{tbl:snd:log:mor} & Sound Logic Morphism 
\\\hline
\S\,\ref{sec:db}
   & Tbl.\,\ref{tbl:db}     & Relational Database  
\\ & Tbl.\,\ref{tbl:db:mor} & Database Morphism 
\\\hline\hline
\S\,\ref{sec:append}
& Tbl.\,\ref{tbl:adjoints} & \texttt{FOLE} Adjoint Bridges
\\\newline
& Tbl.\,\ref{tbl:fole:morph} & \texttt{FOLE} Morphisms
\\\newline
& Tbl.\,\ref{tbl:lax:info} & Lax Infomorphisms
%
\\\hline
\end{tabular} 
}}
\end{tabular}}}}}
\end{center}
\caption{Figures and Tables}
\label{tbl:figs:tbls}
\end{table}
%


\newpage
\section{Basic Concepts}\label{sec:log:environ}

%

The basic concepts of \texttt{FOLE} are listed in
Tbl.\,\ref{tbl:fole:refs}.
\footnote{Satisfaction for (abstract) specifications
is define in
Def.\,\ref{satis:abs:spec} of \S\,\ref{sub:sec:spec:sat}.}
Except for the formula interpretation in Tbl.5,
we can refer all the formal material to 
the paper 
``The {\ttfamily ERA} of {\ttfamily FOLE}: Superstructure''
\cite{kent:fole:era:supstruc}.


%
\begin{table}
\begin{center}
{\footnotesize{
\setlength{\extrarowheight}{3.5pt}
\begin{tabular}{|
@{\hspace{5pt}}
l
@{\hspace{15pt}}
l
@{\hspace{5pt}}
|}
\hline
\S\,2.1 & Formalism (formulas, sequents, constraints)
\\
& Axioms (Tbl.\,3)
\\
\S\,2.2 & Semantics (formula interpretation)
\\
& Formal/Semantics Reflection (Tbl.\,6)
\\
\S\,2.3 & Satisfaction (sequents, constraints)
\\\hline
\end{tabular}}}
\end{center}
\caption{Basic Concepts}
\label{tbl:fole:refs}
\end{table}
%

\subsection{Formalism}\label{sec:fmlism}



\paragraph{Formulas.}

Let $\mathcal{S} = {\langle{R,\sigma,X}\rangle}$ be a fixed schema with 
a set of entity types $R$,
a set of sorts (attribute types) $X$
and a signature function
$R\xrightarrow{\sigma}\mathrmbf{List}(X)$.
The set of entity types $R$ is partitioned
$R = \bigcup_{{\langle{I,s}\rangle}\in\mathrmbf{List}(X)}R(I,s)$
into fibers,
where
$R(I,s){\;\subseteq\;}R$
is the fiber (subset) of all entity types with signature ${\langle{I,s}\rangle}$.
These are called ${\langle{I,s}\rangle}$-ary entity types. 
\footnote{This is a slight misnomer, 
since the signature of $r$ is $\sigma(R)={\langle{I,s}\rangle}$,
whereas the arity of $r$ is $\alpha(R)=I$.}
Here,
we follow the tuple, domain, and relation calculi from database theory,
using logical operations to extend the set of basic entity types $R$
to a set of defined entity types $\widehat{R}$ called formulas or queries.
%
%
%

%
Formulas, which are defined entity types corresponding to queries, 
are constructed by using logical connectives within a fiber and logical flow 
along signature morphisms between fibers (Tbl.~\ref{tbl:syn:flow}).
\footnote{An $\mathcal{S}$-signature morphism ${\langle{I',s'}\rangle}\xrightarrow{h}{\langle{I,s}\rangle}$
in $\mathrmbf{List}(X)$
is an arity function $I'\xrightarrow{\,h\,}I$ that preserves signature $s' = h{\,\cdot\,}s$.}
\footnote{The full version of \texttt{FOLE} 
(Kent~\cite{kent:iccs2013})
defines syntactic flow along term vectors.}
Logical connectives on formulas
express intuitive notions of natural language operations on the interpretation (extent, view) of formulas.
These connectives include:
conjunction, disjunction, negation, implication, etc.
For any signature ${\langle{I,s}\rangle}$,
let $\widehat{R}(I,s) \subseteq \widehat{R}$ 
denote the set of all formulas with this signature.
There are called ${\langle{I,s}\rangle}$-ary formulas.
The set of $\mathcal{S}$-formulas is partitioned as
$\widehat{R} = \bigcup_{{\langle{I,s}\rangle}\in\mathrmbf{List}(X)}\widehat{R}(I,s)$.
\begin{flushleft}
{{\footnotesize{\begin{minipage}{345pt}
\begin{itemize}
\item[\textsf{fiber:}]
Let ${\langle{I,s}\rangle}$ be any signature.
Any ${\langle{I,s}\rangle}$-ary entity type (relation symbol) is an 
${\langle{I,s}\rangle}$-ary formula;
that is, $R(I,s) \subseteq \widehat{R}(I,s)$.
For a pair of ${\langle{I,s}\rangle}$-ary formulas $\varphi$ and $\psi$, 
there are the following ${\langle{I,s}\rangle}$-ary formulas:
meet $(\varphi{\,\wedge\,}\psi)$, join $(\varphi{\,\vee\,}\psi)$,
implication $(\varphi{\,\rightarrowtriangle\,}\psi)$ and 
difference $(\varphi{\,\setminus\,}\psi)$. 
For ${\langle{I,s}\rangle}$-ary formula $\varphi$,
there is an ${\langle{I,s}\rangle}$-ary negation formula $\neg\varphi$.
There are top/bottom ${\langle{I,s}\rangle}$-ary formulas $\top_{{\langle{I,s}\rangle}}$ and $\bot_{{\langle{I,s}\rangle}}$.
\item[\textsf{flow:}]
Let ${\langle{I',s'}\rangle}\xrightarrow{h}{\langle{I,s}\rangle}$ be any signature morphism.
For ${\langle{I,s}\rangle}$-ary formula $\varphi$,
there are ${\langle{I',s'}\rangle}$-ary existentially/universally quantified formulas
${\scriptstyle\sum}_{h}(\varphi)$ and ${\scriptstyle\prod}_{h}(\varphi)$.
\footnote{For any index $i \in I$,
quantification for the complement inclusion signature function
${\langle{I\setminus\{i\},s'}\rangle}\xrightarrow{\text{inc}_{i}}{\langle{I,s}\rangle}$
gives the traditional syntactic quantifiers
$\forall_{i}\varphi,\exists_{i}{\varphi}$.}
For an ${\langle{I',s'}\rangle}$-ary formula $\varphi'$,
there is an ${\langle{I,s}\rangle}$-ary substitution formula
${h}^{\ast}(\varphi') = \varphi'(t)$.
\end{itemize}
\end{minipage}}}}
\end{flushleft}
\begin{table}[h]
\begin{center}
{{\footnotesize{\setlength{\extrarowheight}{2pt}{$\begin{array}{c}
{\langle{I',s'}\rangle}\xrightarrow{h}{\langle{I,s}\rangle}
\\
{\setlength{\unitlength}{0.6pt}\begin{picture}(120,60)(0,-25)
\put(0,0){\makebox(0,0){\footnotesize{$\widehat{R}(I',s')$}}}
\put(120,0){\makebox(0,0){\footnotesize{$\widehat{R}(I,s)$}}}
\put(60,20){\makebox(0,0){\scriptsize{${\scriptstyle\sum}_{{h}}$}}}
\put(62,2){\makebox(0,0){\scriptsize{${h}^{\ast}$}}}
\put(60,-22){\makebox(0,0){\scriptsize{${\scriptstyle\prod}_{{h}}$}}}
\put(85,12){\vector(-1,0){50}}
\qbezier(35,0)(43,0)(51,0)\qbezier(69,0)(77,0)(85,0)\put(85,0){\vector(1,0){0}}
\put(85,-12){\vector(-1,0){50}}
\end{picture}}
\end{array}$}}}}
\end{center}
\caption{Syntactic Flow}
\label{tbl:syn:flow}
\end{table}

\comment{
\begin{center}
{\fbox{\begin{tabular}{c}
\setlength{\unitlength}{0.78pt}
\begin{picture}(120,10)(0,-10)
\put(0,0){\makebox(0,0){\footnotesize{$\mathrmbfit{ext}_{\widehat{\mathcal{E}}}(I',s')$}}}
\put(120,0){\makebox(0,0){\footnotesize{$\mathrmbfit{ext}_{\widehat{\mathcal{E}}}(I,s)$}}}
\put(60,20){\makebox(0,0){\scriptsize{${\scriptstyle\sum}_{{h}}$}}}
\put(62,2){\makebox(0,0){\scriptsize{${h}^{\ast}$}}}
\put(60,-22){\makebox(0,0){\scriptsize{${\scriptstyle\prod}_{{h}}$}}}
\put(85,12){\vector(-1,0){50}}
\qbezier(35,0)(43,0)(51,0)\qbezier(69,0)(77,0)(85,0)\put(85,0){\vector(1,0){0}}
\put(85,-12){\vector(-1,0){50}}
\end{picture}
\end{tabular}}}
\end{center}
}

In general,
we regard formulas to be constructed entities or queries (defining views and interpretations; i.e., relations/tables),
not assertions.
Contrast this with the use of ``asserted formulas'' below.
For example, 
in a corporation data model the conjunction
$(\mathtt{Salaried}{\,\wedge\,}\mathtt{Married})$
is not an assertion, but a constructed entity type or query that defines the view
``salaried employees that are married''.
Formulas form a schema 
$\widehat{\mathcal{S}} = {\langle{\widehat{R},\widehat{\sigma},X}\rangle}$ 
that extends $\mathcal{S}$
with $\mathcal{S}$-formulas as entity types:
with the inductive definitions above,
the set of entity types 
is extended
to a set of logical formulas
$R\xhookrightarrow{\mathrmit{inc}_{\mathcal{S}}}\widehat{R}$,
and
the entity type signature function 
is extended
to a formula signature function $\widehat{R}\xrightarrow{\;\widehat{\sigma}\;}\mathrmbf{List}(X)$
with
$\sigma = \mathrmit{inc}_{\mathcal{S}}{\;\cdot\;}\widehat{\sigma}$.



\paragraph{Sequents.}

To make an assertion about things, 
we use a sequent.
Let $\mathcal{S} = {\langle{R,\sigma,X}\rangle}$ be a schema.
A (binary) $\mathcal{S}$-sequent 
\footnote{Since {\ttfamily FOLE} formulas are not just types,
but are constructed using, 
inter alia, 
conjunction and disjunction operations, 
we can restrict attention to \emph{binary} sequents.} 
is a pair of formulas
$(\varphi,\psi){\,\in\,}\widehat{R}{\,\times}\widehat{R}$
with the same signature 
$\widehat{\sigma}(\varphi) = {\langle{I,s}\rangle} = \widehat{\sigma}(\psi)$. 
To be explicit that we are making an assertion,
we use the turnstile notation $\varphi{\;\vdash\;}\psi$ for a sequent.
Then, 
we are claiming that a specialization-generalization relationship exists between the formulas $\varphi$ and $\psi$.
A asserted sequent $\varphi{\;\vdash\;}\psi$ expresses interpretation widening,
with the interpretation (view) of $\varphi$ required to be within the interpretation (view) of $\psi$.
%
%
%
An \emph{asserted} formula $\varphi\in\widehat{R}$ 
can be identified with the sequent $\top{\;\vdash\;}\varphi$ in $\widehat{R}{\,\times}\widehat{R}$,
which asserts the universal view ``all entities'' of signature $\sigma(\varphi)={\langle{I,s}\rangle}$.
%
Hence,
from an entailment viewpoint we can say that ``formulas are sequents''.
In the opposite direction,
there is an enfolding map
$\widehat{R}{\,\times}\widehat{R}\rightarrow\widehat{R}$
that maps $\mathcal{S}$-sequents to 
$\mathcal{S}$-formulas
$({\varphi}{\;\vdash\;}{\psi})\mapsto(\varphi{\,\rightarrowtriangle\,}\psi)$.
%
The axioms 
in Tbl.\,3 of the paper
\cite{kent:fole:era:supstruc}
make sequents into an order.
Let $\mathrmbf{Con}_{\mathcal{S}}(I,s)={\langle{\widehat{R}(I,s),\vdash}\rangle}$ denote the fiber preorder of $\mathcal{S}$-formulas.


\paragraph{Constraints.}

Sequents only connect formulas within a particular fiber: 
an $\mathcal{S}$-sequent $\varphi{\;\vdash\;}\psi$
is between two formulas with the same signature
$\widehat{\sigma}(\varphi) = {\langle{I,s}\rangle} = \widehat{\sigma}(\psi)$,
and hence between elements in the same fiber $\varphi,\psi{\;\in\;}\widehat{R}(I,s)$. 
We now define a useful notion that connects formulas across fibers.
An $\mathcal{S}$-constraint 
$\varphi'{\,\xrightarrow{h\,}\,}\varphi$
consists of a signature morphism
$\widehat{\sigma}(\varphi')={\langle{I',s'}\rangle}\xrightarrow{h}{\langle{I,s}\rangle}=\widehat{\sigma}(\varphi)$
and 
a binary sequent ${\scriptstyle\sum}_{h}(\varphi){\;\vdash\;}\varphi'$ in $\mathrmbf{Con}_{\mathcal{S}}(I',s')$,
or equivalently by the axioms  
in Tbl.\,3 of the paper
\cite{kent:fole:era:supstruc},
a binary sequent $\varphi{\;\vdash\;}{h}^{\ast}(\varphi')$ in $\mathrmbf{Con}_{\mathcal{S}}(I,s)$.
Hence,
a constraint requires 
that the interpretation of the $h^{\mathrm{th}}$-projection of $\varphi$ be within the interpretation of $\varphi'$,
or equivalently,
that the interpretation of $\varphi$ be within the interpretation of the $h^{\mathrm{th}}$-substitution of $\varphi'$. 
\footnote{In some sense,
this formula/constraint approach to formalism turns the tuple calculus upside down,
with atoms in the tuple calculus becoming constraints here.}

Given any schema $\mathcal{S}$,
an $\mathcal{S}$-constraint ${\varphi'}\xrightarrow{h}{\varphi}$
has source formula ${\varphi'}$ and target formula ${\varphi}$.
Constraints are closed under composition:
${\varphi''}\xrightarrow{h'}{\varphi'}\xrightarrow{h}{\varphi}
={\varphi''}\xrightarrow{h'{\,\cdot\,}h}{\varphi}$.
Let $\mathrmbf{Cons}(\mathcal{S})$ 
%
\footnote{$\mathrmbf{Cons}(\mathcal{S})$,
which stands for ``$\mathcal{S}$-constraints'',
was defined in the paper
\cite{kent:fole:era:supstruc}.}
%
denote the mathematical context, 
whose set of objects are $\mathcal{S}$-formulas 
and whose set of morphisms are $\mathcal{S}$-constraints.
This context is fibered over the projection passage
$\mathrmbf{Cons}(\mathcal{S})\rightarrow\mathrmbf{List}(X)
:\bigl(\varphi'{\,\xrightarrow{h\,}\,}\varphi\bigr)
\mapsto
\bigl({\langle{I',s'}\rangle}\xrightarrow{h}{\langle{I,s}\rangle}\bigr)$.
Formula formation is idempotent:
$\widehat{{\widehat{R}}} = \widehat{R}$
and
$\mathrmbf{Cons}(\widehat{\mathcal{S}})
= \mathrmbf{Cons}(\mathcal{S})$.
%
\footnote{Formulas form a schema 
$\widehat{\mathcal{S}} = {\langle{\widehat{R},\widehat{\sigma},X}\rangle}$ 
that extends $\mathcal{S}$
with $\mathcal{S}$-formulas as entity types (relation names):
by induction,
the set of entity types 
is extended
to a set of logical formulas
$R\xhookrightarrow{\mathrmit{inc}_{\mathcal{S}}}\widehat{R}$,
and
the entity type signature function 
is extended
to a formula signature function $\widehat{R}\xrightarrow{\;\widehat{\sigma}\;}\mathrmbf{List}(X)$
with
$\sigma = \mathrmit{inc}_{\mathcal{S}}{\;\cdot\;}\widehat{\sigma}$.}
%

Sequents are special cases of constraints:
a sequent $\varphi'{\;\vdash\;}\varphi$
asserts a constraint ${\varphi}\xrightarrow{1}{\varphi'}$
that uses an identity signature morphism.
\footnote{A constraint in the fiber $\mathrmbf{Con}_{\mathcal{S}}(I,s)$
uses an identity signature morphism ${\varphi}\xrightarrow{1}{\varphi'}$,
and hence is a sequent $\varphi'{\;\vdash\;}\varphi$.}
Since an asserted formula $\varphi$ 
can be identified with the sequent $\top{\;\vdash\;}\varphi$,
it can also be identified with the constraint ${\varphi}\xrightarrow{1}{\top}$.
Thus,
from an entailment viewpoint we can say that
``formulas are sequents are constraints''.
In the opposite direction,
there are enfolding maps
that map $\mathcal{S}$-constraints to $\mathcal{S}$-formulas:
either
$\bigl({\varphi'}\xrightarrow{h}{\varphi}\bigr)\mapsto\bigl({\scriptstyle\sum}_{h}(\varphi){\,\rightarrowtriangle\,}\varphi'\bigr)$ 
with signature ${\langle{I',s'}\rangle}$, 
or
$\bigl({\varphi'}\xrightarrow{h}{\varphi}\bigr)\mapsto\bigl(\varphi{\,\rightarrowtriangle\,}{h}^{\ast}(\varphi')\bigr)$ 
with signature ${\langle{I,s}\rangle}$.


%

\newpage
\subsection{Interpretation}\label{sub:sec:interp:fml:constr}




The logical semantics of a structure $\mathcal{M}$ resides in its core,
which is defined by its formula interpretation function 
\begin{equation}\label{eqn:fmla:interp}
\mathrmbfit{T}_{\mathcal{M}} : \widehat{R} \rightarrow \mathrmbf{Tbl}(\mathcal{A}).
\end{equation}
%
%
%
The formula interpretation function,
%
which extends the traditional interpretation function 
$T_{\mathcal{M}} : R\rightarrow\mathrmbf{Tbl}(\mathcal{A})$
of a structure $\mathcal{M}$
(Def.\,\ref{def:struc} in \S\,\ref{sec:struc}),
is defined in Tbl.~\ref{tbl:fml:int}
by induction within fibers and flow between fibers.
To respect the logical semantics,
the formal flow operators 
(${\scriptstyle\sum}_{h}$,${\scriptstyle\prod}_{h}$,${h}^{\ast}$)
for existential/universal quantification and substitution
reflect the semantic flow operators
(${\scriptstyle\sum}_{h}$,${\scriptstyle\prod}_{h}$,${h}^{\ast}$)
via interpretation
(Tbl.~\ref{tbl:fml-sem:refl}).
\comment{This tabular/relational reflection 
is a special case of Prop.~5.59 in the text {\itshape Topos Theory}~\cite{johnstone:77},
suggesting that all of the development in this paper could be done in an arbitrary topos,
or even in a more general setting.}
%
\footnote{The morphic aspect of Tbl.~\ref{tbl:fml-sem:refl} 
anticipates the definition of satisfaction
in \S.\ref{sub:sec:sat}:
for sequents within the fibers
$\mathrmbfit{T}_{\mathcal{A}}(I,s)$
and
for constraints using flow between fibers;
${\bigl\langle{{\scriptscriptstyle\sum}_{h}
{\;\dashv\;}{h}^{\ast}
{\;\dashv\;}{\scriptstyle\prod}_{{h}}
\bigr\rangle}}:
\mathrmbf{Tbl}_{\mathcal{A}}(I',s')
{\;\leftrightarrows\;}
\mathrmbf{Tbl}_{\mathcal{A}}(I,s)$.}
%

%
\begin{description}
\item[\textsf{fiber}:] 
For each signature ${\langle{I,s}\rangle}{\,\in\,}\mathrmbf{List}(X)$, 
the fiber function
$\widehat{R}(I,s)
\xrightarrow{\mathrmbfit{T}^{\mathcal{M}}_{{\langle{I\!,s}\rangle}}}
\mathrmbf{Tbl}_{\mathcal{A}}(I,s)$
is defined in the top part of Tbl.~\ref{tbl:fml:int}
by induction on formulas.
%
At the base step
(first line in the top of Tbl.~\ref{tbl:fml:int}),
it defines the formula interpretation of an entity type $r\in\widehat{R}$ 
as the traditional interpretation of that type:
the
$\mathcal{A}$-table
$T_{\mathcal{M}}(r) 
= {\langle{\sigma(r),\mathrmbfit{K}(r),\tau_{r}}\rangle} 
\in \mathrmbf{Tbl}(\mathcal{A})$
(Fig.\,\ref{fig:fole:struc}
of
\S\ref{sec:struc}).
\comment{
whose signature is the
$X$-signature $\sigma(r)
\in \mathrmbf{List}(X)$,
whose key set is the extent
$\mathrmbfit{K}(r) = \mathrmbfit{ext}_{\mathcal{E}}(r)$
in the entity classification
$\mathcal{E}={\langle{R,\mathrmbfit{K}}\rangle}$
, and whose 
tuple function $\mathrmbfit{K}(r)\xrightarrow{\tau_{r}}\mathrmbfit{tup}_{\mathcal{A}}(\sigma(r))$,
which is a restriction of the tuple function $K\xrightarrow{\tau}\mathrmbf{List}(Y)$.
}
%
At the induction step
(remaining lines in the top of Tbl.~\ref{tbl:fml:int}),
it represents the logical operations by their associated boolean operations: 
meet of interpretations for conjunction,
join of interpretations for disjunction,
etc.;
see the boolean operations defined in \S\,3.2 of the paper
``Relational Operations in {\ttfamily FOLE}''
\cite{kent:fole:rel:ops}.
\newline
\item[\textsf{flow}:] 
(bottom Tbl.~\ref{tbl:fml:int}), 
%
It represents the syntactic flow operators 
in 
Tbl.~\ref{tbl:syn:flow} 
by their associated semantic flow operators;
see the adjoint flow for fixed type domain $\mathcal{A}$ 
defined in \S\,3.3.1 of the paper
``Relational Operations in {\ttfamily FOLE}''
\cite{kent:fole:rel:ops}.
%
\footnote{The paper \cite{kent:fole:rel:ops} used only two of the three operators
in the fiber adjunction of tables
\begin{equation}\label{def:fbr:adj:sign:mor}
{{\begin{picture}(120,10)(0,-4)
\put(60,0){\makebox(0,0){\footnotesize{$
\mathrmbf{Tbl}_{\mathcal{A}}(I',s')
{\;\xleftarrow
[{\bigl\langle{
{\scriptscriptstyle\sum}_{h}{\;\dashv\;}{h}^{\ast}
\bigr\rangle}}]
{{\bigl\langle{\acute{\mathrmbfit{tbl}}_{\mathcal{A}}(h)
{\;\dashv\;}
\grave{\mathrmbfit{tbl}}_{\mathcal{A}}(h)}\bigr\rangle}}
\;}
\mathrmbf{Tbl}_{\mathcal{A}}(I,s)
$}}}
\end{picture}}}
\end{equation}
defined by composition/pullback:
the left adjoint existential operation (called projection) and 
the right adjoint substitution or inverse image operation (called inflation).}
%
The 
function 
$\widehat{R}\xrightarrow{\mathrmbfit{T}_{\mathcal{M}}}\mathrmbf{Tbl}(\mathcal{A})$
is the parallel combination of its fiber functions
{\footnotesize{
$
\biggl\{ 
\widehat{R}(I,s)
\xrightarrow{\mathrmbfit{T}^{\mathcal{M}}_{{\langle{I\!,s}\rangle}}}
\mathrmbf{Tbl}_{\mathcal{A}}(I,s)
\mid {\langle{I,s}\rangle}{\,\in\,}\mathrmbf{List}(X) 
\biggr\}
$}\normalsize}
defined above.
The function
$\mathrmbfit{T}_{\mathcal{M}}$ 
is defined in terms of these fiber functions and the flow operators.
%
\end{description}
\begin{table}
\begin{center}
{\footnotesize{\setlength{\extrarowheight}{2pt}
\begin{tabular}{|@{\hspace{5pt}}r@{\hspace{20pt}}l@{\hspace{5pt}}|}
\multicolumn{2}{l}{\rule[-6pt]{0pt}{1pt}\textsf{fiber:} signature ${\langle{I,s}\rangle}$ 
with extent (tuple set) 
$\mathrmbfit{tup}_{\mathcal{A}}(I,s)
{\,=\,}\prod_{i\in{I}}\,\mathcal{A}_{s_{\!i}}$}
\\ \hline
\textit{operator} 
&
\textit{definition}
\hfill
$\mathrmbfit{T}_{\mathcal{M}}(\varphi)
\in
\mathrmbf{Tbl}_{\mathcal{A}}(I,s)
= \bigl(\mathrmbf{Set}{\,\downarrow\,}\mathrmbfit{tup}_{\mathcal{A}}(I,s)\bigr)
$
\\
entity type & 
$\mathrmbfit{T}_{\mathcal{M}}(r)
=T_{\mathcal{M}}(r) 
= {\langle{\sigma(r),\mathrmbfit{K}(r),\tau_{r}}\rangle}$ 
\hfill
$r{\,\in\,}R(I,s)\subseteq\widehat{R}(I,s)$
\\
meet & 
$\mathrmbfit{T}_{\mathcal{M}}(\varphi{\,\wedge\,}\psi)
=\mathrmbfit{T}_{\mathcal{M}}(\varphi){\,\wedge\,}\mathrmbfit{T}_{\mathcal{M}}(\psi)$
\hfill
$\varphi,\psi{\,\in\,}\widehat{R}(I,s)$
\\
join &
$\mathrmbfit{T}_{\mathcal{M}}(\varphi{\,\vee\,}\psi)
=\mathrmbfit{T}_{\mathcal{M}}(\varphi){\,\vee\,}\mathrmbfit{T}_{\mathcal{M}}(\psi)$
\\
top & 
$\mathrmbfit{T}_{\mathcal{M}}({\scriptstyle\top_{{\langle{I,s}\rangle}}})
={\langle{\mathrmbfit{tup}_{\mathcal{A}}(I,s),1}\rangle}$
\hfill$\text{terminal table}$
\\
bottom & 
$\mathrmbfit{T}_{\mathcal{M}}({\scriptstyle\bot_{{\langle{I,s}\rangle}}})
={\langle{\emptyset,0}\rangle}$
\hfill$\text{initial table}$
\\
negation & 
$\mathrmbfit{T}_{\mathcal{M}}(\neg\varphi)=
\mathrmbfit{T}_{\mathcal{M}}(\top_{{\langle{I,s}\rangle}}{\,\setminus\,}\varphi) =
\mathrmbfit{T}_{\mathcal{M}}(\top_{{\langle{I,s}\rangle}})
{-}\mathrmbfit{T}_{\mathcal{M}}(\varphi)$
\\
implication & 
$\mathrmbfit{T}_{\mathcal{M}}(\varphi{\,\rightarrowtriangle\,}\psi)
=\mathrmbfit{T}_{\mathcal{M}}(\varphi){\,\rightarrowtriangle\,}\mathrmbfit{T}_{\mathcal{M}}(\psi)
=(\neg\mathrmbfit{T}_{\mathcal{M}}(\varphi)){\,\vee\,}\mathrmbfit{T}_{\mathcal{M}}(\psi)$
\\
difference & 
$\mathrmbfit{T}_{\mathcal{M}}(\varphi{\,\setminus\,}\psi)
= \mathrmbfit{T}_{\mathcal{M}}(\varphi){-}\mathrmbfit{T}_{\mathcal{M}}(\psi)
=\mathrmbfit{T}_{\mathcal{M}}(\varphi){\,\wedge\,}(\neg\mathrmbfit{T}_{\mathcal{M}}(\psi))$
\rule[-5pt]{0pt}{10pt}
\\ \hline
\multicolumn{2}{l}{\rule{0pt}{20pt}\textsf{flow:} signature morphism
${\langle{I',s'}\rangle}\xrightarrow{h}{\langle{I,s}\rangle}$}
\\
\multicolumn{2}{l}{\rule[-6pt]{0pt}{1pt}with tuple map
$\mathrmbfit{tup}_{\mathcal{A}}(I',s')\xleftarrow{\mathrmbfit{tup}_{\mathcal{A}}(h)}\mathrmbfit{tup}_{\mathcal{A}}(I,s)$}
\\ \hline
\textit{operator} & \multicolumn{1}{l|}{\textit{definition}}
\\
existential
& $\mathrmbfit{T}_{\mathcal{M}}({\scriptstyle\sum}_{h}(\varphi))
={\scriptstyle\sum}_{h}(\mathrmbfit{T}_{\mathcal{M}}(\varphi))$
\hfill
$\varphi{\,\in\,}\widehat{R}(I,s)$
\\
universal
& $\mathrmbfit{T}_{\mathcal{M}}({\scriptstyle\prod}_{h}(\varphi))
={\scriptstyle\prod}_{h}(\mathrmbfit{T}_{\mathcal{M}}(\varphi))$
\\
substitution
& 
$\mathrmbfit{T}_{\mathcal{M}}({h}^{\ast}(\varphi'))
={h}^{\ast}(\mathrmbfit{I}_{\mathcal{M}}(\varphi'))$
\rule[-5pt]{0pt}{10pt}
\hfill
$\varphi'{\,\in\,}\widehat{R}(I',s')$
\\ \hline
\end{tabular}}}
\end{center}
\caption{Formula Interpretation}
\label{tbl:fml:int}
\end{table}
\begin{table}
\begin{center}
{\scriptsize{\begin{tabular}{@{\hspace{30pt}}c@{\hspace{50pt}}c}
{{\begin{tabular}[b]{c}
\setlength{\unitlength}{0.79pt}
\begin{picture}(120,90)(0,-10)
\put(-2,80){\makebox(0,0){\footnotesize{$
{\bigl\langle{\widehat{R}(I',s'),\vdash'}\bigr\rangle}$}}}
\put(120,80){\makebox(0,0){\footnotesize{$
{\bigl\langle{\widehat{R}(I,s),\vdash}\bigr\rangle}$}}}
\put(0,0){\makebox(0,0){\footnotesize{$
\mathrmbf{Tbl}_{\mathcal{A}}(I',s')$}}}
\put(120,0){\makebox(0,0){\footnotesize{$
\mathrmbf{Tbl}_{\mathcal{A}}(I,s)$}}}
\put(60,100){\makebox(0,0){\scriptsize{${\scriptstyle\sum}_{{h}}$}}}
\put(62,82){\makebox(0,0){\scriptsize{${h}^{\ast}$}}}
\put(60,58){\makebox(0,0){\scriptsize{${\scriptstyle\prod}_{{h}}$}}}
\put(60,20){\makebox(0,0){\scriptsize{${\scriptstyle\sum}_{{h}}$}}}
\put(62,2){\makebox(0,0){\scriptsize{${h}^{\ast}$}}}
\put(60,-20){\makebox(0,0){\scriptsize{${\scriptstyle\prod}_{{h}}$}}}
\put(6,40){\makebox(0,0)[l]{\scriptsize{$\mathrmbfit{T}^{\mathcal{M}}_{{\langle{I'\!,s'}\rangle}}$}}}
\put(126,40){\makebox(0,0)[l]{\scriptsize{$\mathrmbfit{T}^{\mathcal{M}}_{{\langle{I\!,s}\rangle}}$}}}
\put(85,92){\vector(-1,0){50}}
\qbezier(35,80)(43,80)(51,80)\qbezier(69,80)(77,80)(85,80)\put(85,80){\vector(1,0){0}}
\put(85,68){\vector(-1,0){50}}
\put(85,12){\vector(-1,0){50}}
\put(85,-12){\vector(-1,0){50}}
\qbezier(35,0)(43,0)(51,0)\qbezier(69,0)(77,0)(85,0)\put(85,0){\vector(1,0){0}}
\put(0,65){\vector(0,-1){50}}
\put(120,65){\vector(0,-1){50}}
\end{picture}
\end{tabular}}}
&
{{\begin{tabular}[b]{l}
{\footnotesize\setlength{\extrarowheight}{2.5pt}$\begin{array}[b]{l}
{\langle{I',s'}\rangle}\xrightarrow{h}{\langle{I,s}\rangle} \\ 
\mathrmbfit{tup}_{\mathcal{A}}(I',s')\xleftarrow{\mathrmbfit{tup}(h)}\mathrmbfit{tup}_{\mathcal{A}}(I,s)
\end{array}$}
\\ \\
{\footnotesize\setlength{\extrarowheight}{2.5pt}$\begin{array}[b]{|@{\hspace{5pt}}l@{\hspace{2pt}}|}
\hline
{\scriptstyle\sum}_{{h}} \dashv {h}^{\ast} \dashv {\scriptstyle\prod}_{{h}} 
\\ \hline
{\scriptstyle\sum}_{{h}} \circ \mathrmbfit{T}^{\mathcal{M}}_{{\langle{I'\!,s'}\rangle}}
 = \mathrmbfit{T}^{\mathcal{M}}_{{\langle{I\!,s}\rangle}} \circ {\scriptstyle\sum}_{{h}}
\\
{h}^{\ast} \circ \mathrmbfit{T}^{\mathcal{M}}_{{\langle{I\!,s}\rangle}} 
= \mathrmbfit{T}^{\mathcal{M}}_{{\langle{I'\!,s'}\rangle}} \circ {h}^{\ast}	
\\
{\scriptstyle\prod}_{{h}} \circ \mathrmbfit{T}^{\mathcal{M}}_{{\langle{I'\!,s'}\rangle}}
 = \mathrmbfit{T}^{\mathcal{M}}_{{\langle{I\!,s}\rangle}} \circ {\scriptstyle\prod}_{{h}}
\\ \hline
\end{array}$}
\end{tabular}}}
\\ & \\
\end{tabular}}}
\end{center}
\caption{Formal/Semantics Reflection}
\label{tbl:fml-sem:refl}
\end{table}
%


\newpage
\subsection{Satisfaction.}\label{sub:sec:sat} 






Satisfaction is a fundamental classification
between formalism and semantics.
The atom of formalism used in satisfaction is the {\ttfamily FOLE} constraint, whereas 
the atom of semantics used in satisfaction is the {\ttfamily FOLE} structure.
Satisfaction is defined in terms of 
the table formula interpretation function
$\mathrmbfit{T}_{\mathcal{M}} : \widehat{R} \rightarrow \mathrmbf{Tbl}(\mathcal{A})$
of \S\,\ref{sub:sec:interp:fml:constr}
and the associated
the relation formula interpretation function
$\mathrmbfit{R}_{\mathcal{M}} : \widehat{R} \rightarrow \mathrmbf{Rel}(\mathcal{A})$.
%
\footnote{Denoted by 
$\mathrmbfit{I}_{\mathcal{M}} : \widehat{R} \rightarrow \mathrmbf{Rel}(\mathcal{A})$
in
\S\,2,2,1 of the paper
``The {\ttfamily ERA} of {\ttfamily FOLE}: Superstructure''
\cite{kent:fole:era:supstruc}}
%
\comment{Important definitions follow a logical order:
\emph{formula interpretation}
$\;\Rightarrow\;$ \emph{satisfaction}
$\;\Rightarrow\;$ \emph{structure interpretation}
$\;\Rightarrow\;$ \emph{sound logic interpretation}.}
Assume that we are given a structure schema 
$\mathcal{S} = {\langle{R,\sigma,X}\rangle}$
with a set of predicate symbols $R$
and a signature (header) map 
$R \xrightarrow{\sigma} \mathrmbf{List}(X)$.


\paragraph{Sequent Satisfaction.}

\begin{sloppypar}
A (lax) $\mathcal{S}$-structure 
$\mathcal{M} = {\langle{\mathcal{E},\sigma,\tau,\mathcal{A}}\rangle} \in \mathrmbf{Struc}(\mathcal{S})$
\emph{satisfies} a formal $\mathcal{S}$-sequent $\varphi{\;\vdash\;}\varphi'$,
for a pair of formulas
$(\varphi,\varphi'){\,\in\,}\widehat{R}{\,\times}\widehat{R}$
with the same signature 
$\widehat{\sigma}(\varphi) = {\langle{I,s}\rangle} = \widehat{\sigma}(\varphi')$,
when the interpretation widening of views asserted by the sequent actually holds in $\mathcal{M}$:
%
\begin{center}
{{\footnotesize{$
\underset{\text{in}\;\mathrmbf{Rel}_{\mathcal{A}}(I,s)}
{\mathrmbfit{R}_{\mathcal{M}}(\varphi){\;\subseteq\;}\mathrmbfit{R}_{\mathcal{M}}(\varphi')}
\;\;\;
\text{\underline{reflectively}}
\footnote{\label{reflection}See the formal/semantics reflection 
defined in \S\,A.1 of the paper ``The {\ttfamily FOLE} Table''
\cite{kent:fole:era:tbl}
and
in \S\,3.1 of the paper ``Relational Operations in {\ttfamily FOLE}''
\cite{kent:fole:rel:ops}.}
\;\;\;
\underset{\text{in}\;\mathrmbf{Tbl}_{\mathcal{A}}(I,s)}
{\mathrmbfit{T}_{\mathcal{M}}(\varphi){\;\xrightarrow{\;k\;}\;}
\mathrmbfit{T}_{\mathcal{M}}(\varphi')}
$.}}}
\end{center}
Satisfaction is symbolized either by
$\mathcal{M}{\;\models_{\mathcal{S}}\;}(\varphi{\;\vdash\;}\varphi')$
or by
$\varphi{\;\vdash_{\mathcal{M}}\;}\varphi'$.
For each signature ${\langle{I,s}\rangle}\in\mathrmbf{List}(X)$,
satisfaction defines the fiber order
${\langle{\widehat{R}(I,s),\leq_{\mathcal{M}}}\rangle}$,
%
\comment{The set of formulas $\widehat{R}$ is partitioned into fibers:
for any signature ${\langle{I,s}\rangle}\in\mathrmbf{List}(X)$,
let $\widehat{R}(I,s){\;\subseteq\;}\widehat{R}$ denote the fiber (subset) of all 
formulas with this signature.
These are called ${\langle{I,s}\rangle}$-ary formulas.}
%
where
$\varphi{\;\leq_{\mathcal{M}}\;}\varphi'$
when
$\varphi{\;\vdash_{\mathcal{M}}\;}\varphi'$
for any two formulas $\varphi,\varphi'\in{\widehat{R}}(I,s)$.
%
\end{sloppypar}
%

%
\comment{Satisfaction using formula 
in the traditional approach
is define 
in \S\,2.3 of \cite{kent:fole:era:supstruc}.}
%

%
\begin{figure}
\begin{center}
{{\begin{tabular}{c}
\setlength{\unitlength}{0.51pt}
\begin{picture}(180,230)(0,-40)
\put(0,180){\makebox(0,0){\footnotesize{$\mathrmbfit{K}(\varphi')$}}}
\put(180,180){\makebox(0,0){\footnotesize{$\mathrmbfit{K}(\varphi)$}}}
\put(10,90){\makebox(0,0){\footnotesize{${\wp}t'_{\varphi'}(\mathrmbfit{K}(\varphi'))$}}}
\put(175,90){\makebox(0,0){\footnotesize{${\wp}t_{\varphi}(\mathrmbfit{K}(\varphi))$}}}
\put(-10,0){\makebox(0,0){\footnotesize{$\mathrmbfit{tup}_{\mathcal{A}}(\sigma(\varphi'))$}}}
\put(190,0){\makebox(0,0){\footnotesize{$\mathrmbfit{tup}_{\mathcal{A}}(\sigma(\varphi))$}}}
\put(90,190){\makebox(0,0){\scriptsize{$k$}}}
\put(90,55){\makebox(0,0){\scriptsize{$\mathrmbfit{R}_{\mathcal{M}}(h)$}}}
\put(95,100){\makebox(0,0){\scriptsize{$r$}}}
\put(90,14){\makebox(0,0){\scriptsize{$\mathrmbfit{tup}_{\mathcal{A}}(h)$}}}
\put(-45,105){\makebox(0,0)[r]{\scriptsize{$t_{\varphi'}$}}}
\put(230,105){\makebox(0,0)[l]{\scriptsize{$t_{\varphi}$}}}
\put(-1,50){\makebox(0,0)[r]{\scriptsize{$i_{\varphi'}$}}}
\put(187,50){\makebox(0,0)[l]{\scriptsize{$i_{\varphi}$}}}
\put(-1,140){\makebox(0,0)[r]{\scriptsize{$e_{\varphi'}$}}}
\put(187,140){\makebox(0,0)[l]{\scriptsize{$e_{\varphi}$}}}
\put(150,180){\vector(-1,0){115}}
\put(120,90){\vector(-1,0){55}}
\put(140,45){\vector(-1,0){100}}
\put(125,0){\vector(-1,0){70}}
\put(180,165){\vector(0,-1){60}}
\put(0,165){\vector(0,-1){60}}
\put(0,70){\vector(0,-1){55}}\put(6,70){\oval(12,12)[t]}
\put(180,70){\vector(0,-1){55}}\put(186,70){\oval(12,12)[t]}
\qbezier(-15,165)(-75,90)(-15,15)\put(-15,15){\vector(2,-3){0}}
\qbezier(195,165)(255,90)(195,15)\put(195,15){\vector(-2,-3){0}}
\put(-67,45){\makebox(0,0)[r]{\footnotesize{$
\mathrmbfit{R}_{\mathcal{M}}(\varphi')\left\{\rule{0pt}{28pt}\right.$}}}
\put(247,45){\makebox(0,0)[l]{\footnotesize{$\left.\rule{0pt}{28pt}\right\}\mathrmbfit{R}_{\mathcal{M}}(\varphi)$}}}
\put(0,-40){\makebox(0,0){\normalsize{$\underset{\textstyle{\mathrmbfit{T}_{\mathcal{M}}(\varphi')}}{\underbrace{\rule{40pt}{0pt}}}$}}}
\put(180,-40){\makebox(0,0){\normalsize{$\underset{\textstyle{\mathrmbfit{T}_{\mathcal{M}}(\varphi)}}{\underbrace{\rule{40pt}{0pt}}}$}}}
\put(140,-47){\vector(-1,0){100}}
\put(90,-37){\makebox(0,0){\scriptsize{$\mathrmbfit{T}_{\mathcal{M}}(h)$}}}
\put(90,-58){\makebox(0,0){\scriptsize{${\langle{h,k}\rangle}$}}}
\end{picture}
\end{tabular}}}
\end{center}
\caption{Table-Relation Reflection}
\label{fig:rel:tbl:refl}
\end{figure}
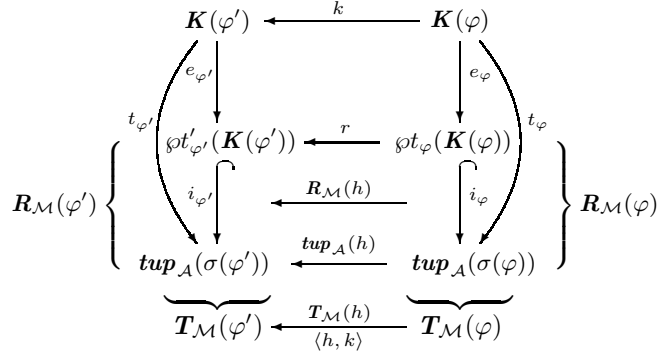
%

%
\newpage
\paragraph{Constraint Satisfaction.}

%


\begin{sloppypar}
A (lax) $\mathcal{S}$-structure $\mathcal{M}\in\mathrmbf{Struc}(\mathcal{S})$
\emph{satisfies} 
a formal $\mathcal{S}$-constraint $\varphi'{\;\xrightarrow{h}\;}\varphi$,
for a pair of formulas
$(\varphi',\varphi){\,\in\,}\widehat{R}{\,\times}\widehat{R}$
connected by a signature morphism
$\widehat{\sigma}(\varphi')={\langle{I',s'}\rangle}\xrightarrow{h}{\langle{I,s}\rangle}=\widehat{\sigma}(\varphi)$,
when 
%
%
the following equivalence holds
%
\comment{{${\exists}_{h}(\mathrmbfit{R}(\varphi))
={\wp}\mathrmbfit{tup}_{\mathcal{A}}(h)({\wp}\tau(\mathrmbfit{K}(\varphi)))$
and
${h}^{-1}(\mathrmbfit{R}(\varphi'))
=\mathrmbfit{tup}_{\mathcal{A}}(h)^{-1}({\wp}\tau(\mathrmbfit{K}(\varphi')))$}.}
%
\begin{equation}\label{adj:fbr:ord}
{{\begin{picture}(100,50)(0,0)
\put(50,20){\makebox(0,0){{\footnotesize{$
\underset{\textstyle{
\underset{{\text{in}\;\mathrmbf{Rel}(\mathcal{A})}}
{\mathrmbfit{R}_{\mathcal{M}}(\varphi')
{\;\xleftarrow[\;{\langle{h,r}\rangle}\;]{\;\mathrmbfit{R}_{\mathcal{M}}(h)\;}\;}
\mathrmbfit{R}_{\mathcal{M}}(\varphi)}
}}
{\underbrace{\underset{\text{in}\;\mathrmbf{Rel}_{\mathcal{A}}(I',s')}
{\mathrmbfit{R}_{\mathcal{M}}(\varphi')
{\;\supseteq\;}
{\exists}_{h}(\mathrmbfit{R}_{\mathcal{M}}(\varphi))}
{\;\;\;\;\;\;\;\;\rightleftarrows\;\;\;\;\;\;\;\;}
\underset{\text{in}\;\mathrmbf{Rel}_{\mathcal{A}}(\mathcal{S})}
{{h}^{-1}(\mathrmbfit{R}_{\mathcal{M}}(\varphi'))
{\;\supseteq\;}
\mathrmbfit{R}_{\mathcal{M}}(\varphi)}
}}
$;}}}}
\end{picture}}}
\end{equation}
\underline{reflectively}$^{\ref{reflection}}$
{(visualized in Fig.\,\ref{fig:rel:tbl:refl})}

%
\comment{See the formal/semantics reflection 
discussed in Sec.~\ref{sub:sec:interp:fml:constr},
illustrated in Tbl.~\ref{tbl:fml-sem:refl}, and
defined in \S\,3.1 of the paper
``Relational Operations in {\ttfamily FOLE}''
\cite{kent:fole:rel:ops}.}
%
%
\begin{equation}\label{adj:fbr:cxt}
{{\begin{picture}(100,50)(0,0)
\put(50,25){\makebox(0,0){{\footnotesize{$
\underset{\textstyle{
\underset{{\text{in}\;\mathrmbf{Tbl}(\mathcal{A})}}
{\mathrmbfit{T}_{\mathcal{M}}(\varphi')
{\;\xleftarrow[{\langle{h,k}\rangle}]{\mathrmbfit{T}_{\mathcal{M}}(h)\;}\;}
\mathrmbfit{T}_{\mathcal{M}}(\varphi)}
}}
{\underbrace{
\underset{{\text{in}\;\mathrmbf{Tbl}_{\mathcal{A}}(I',s')}}
{\mathrmbfit{T}_{\mathcal{M}}(\varphi')\xleftarrow{\;k\;}{\scriptstyle\sum}_{h}(\mathrmbfit{T}_{\mathcal{M}}(\varphi))}
{\;\;\;\;\;\;\;\;\rightleftarrows\;\;\;\;\;\;\;\;}
\underset{{\text{in}\;\mathrmbf{Tbl}_{\mathcal{A}}(\mathcal{S})}}
{{h}^{\ast}(\mathrmbfit{T}_{\mathcal{M}}(\varphi'))\xleftarrow{\;k'\,}\mathrmbfit{T}_{\mathcal{M}}(\varphi)}}}
$.}}}}
\end{picture}}}
\end{equation}
Satisfaction is symbolized by
$\mathcal{M}{\;\models_{\mathcal{S}}\;}(\varphi'{\;\xrightarrow{h}\;}\varphi)$.
\end{sloppypar}


%
\begin{lemma}\label{lem:nat:cxt}
A (lax) $\mathcal{S}$-structure $\mathcal{M}$
determines a mathematical context
$\mathcal{M}^{\mathcal{S}}{\;\sqsubseteq\;}\mathrmbf{Cons}(\mathcal{S})$,
called the \emph{conceptual intent} of $\mathcal{M}$,
whose objects are $\mathcal{S}$-formulas and
whose morphisms are $\mathcal{S}$-constraints satisfied by $\mathcal{M}$.
The ${\langle{I,s}\rangle}^{\text{th}}$ fiber is the order
${\mathcal{M}^{\mathcal{S}}(I,s)}
={\langle{\widehat{R}(I,s),\geq_{\mathcal{M}}}\rangle}$.
\footnote{The satisfaction relation 
corresponds to the ``truth classification'' in Barwise and Seligman~\cite{barwise:seligman:97},
where the conceptual intent $\mathcal{M}^{\mathcal{S}}$
corresponds to the ``theory of $\mathcal{M}$''.}
\end{lemma}
\begin{proof}
\begin{sloppypar}
$\mathcal{M}^{\mathcal{S}}$ is closed under constraint identities and constraint composition.
\newline
{\bfseries 1:} $\mathcal{M}{\;\models_{\mathcal{S}}\;}(\varphi{\;\xrightarrow{1}\;}\varphi)$,
and
{\bfseries 2:}
if $\mathcal{M}{\;\models_{\mathcal{S}}\;}(\varphi''{\;\xrightarrow{h'}\;}\varphi')$
and $\mathcal{M}{\;\models_{\mathcal{S}}\;}(\varphi'{\;\xrightarrow{h}\;}\varphi)$,
then $\mathcal{M}{\;\models_{\mathcal{S}}\;}(\varphi''{\;\xrightarrow{h'{\,\cdot\,}h}\;}\varphi)$,
since $\varphi''{\;\geq_{\mathcal{M}}\;}{\scriptstyle\sum}_{h'}(\varphi')$
and $\varphi'{\;\geq_{\mathcal{M}}\;}{\scriptstyle\sum}_{h}(\varphi)$
implies
$\varphi''{\;\geq_{\mathcal{M}}\;}
{\scriptstyle\sum}_{h'}({\scriptstyle\sum}_{h}(\varphi))={\scriptstyle\sum}_{h'{\,\cdot\,}h}(\varphi)$.
\hfill\rule{5pt}{5pt}
\end{sloppypar}
\end{proof}
%


%
\begin{lemma}\label{lem:interp:pass}
For any structure $\mathcal{M}$,
there is 
a relation interpretation passage and
a table interpretation graph morphism,
%
\footnote{Relational interpretation is closed under composition;
tabular interpretation is closed under composition only up to key equivalence.}
%
\begin{equation}
{{\begin{picture}(120,30)(0,-4)
\put(20,20){\makebox(0,0)[l]{\footnotesize{$
{\mathcal{M}^{\mathcal{S}}}
^{\text{op}}
\!\xrightarrow{\;\mathrmbfit{R}_{\mathcal{M}}}
\mathrmbf{Rel}(\mathcal{A})
\xhookrightarrow{\mathrmbfit{inc}_{\mathcal{A}}} 
\mathrmbf{Tbl}(\mathcal{A})
$ and}}}
\put(14,0){\makebox(0,0)[l]{\footnotesize{$
{|\mathcal{M}^{\mathcal{S}}|}
^{\text{op}}
\!\xrightarrow{\;\mathrmbfit{T}_{\mathcal{M}}}
|\mathrmbf{Tbl}(\mathcal{A})|
$,}}}
\end{picture}}}
\end{equation}
which extend
the formula interpretation function
$\mathrmbfit{T}_{\mathcal{M}} : \widehat{R} \rightarrow \mathrmbf{Tbl}(\mathcal{A})$
of \S\,\ref{sub:sec:interp:fml:constr},
and hence 
the structure interpretation function 
$T_{\mathcal{M}} : R\rightarrow\mathrmbf{Tbl}(\mathcal{A})$
of \S\,\ref{sec:struc}.
\end{lemma}
\begin{proof}
\mbox{}
\begin{description}
\item[objects:] 
An $\mathcal{S}$-formula
$\varphi \in \mathcal{M}^{\mathcal{S}}(I,s)
={\langle{\widehat{R}(I,s),\leq_{\mathcal{M}}}\rangle}$
with $\mathcal{S}$-signature ${\langle{I,s}\rangle}$
is mapped to an $\mathcal{A}$-table
$\mathrmbfit{T}_{\mathcal{M}}(\varphi) =
{\langle{I,s,K,t}\rangle}$
defined by induction in the 
top part of
Tbl.\;\ref{tbl:fml:int}
in
in \S\,\ref{sub:sec:interp:fml:constr}.
Reflectively,
the formula
$\varphi$
is mapped to an $\mathcal{A}$-relation
$\mathrmbfit{R}_{\mathcal{M}}(\varphi) =
{\langle{I,s,{\wp}t(K)}\rangle}$.
\item[morphisms:] 
An $\mathcal{S}$-constraint $\varphi'{\;\xrightarrow{h}\;}\varphi$
satisfied by
the $\mathcal{S}$-structure $\mathcal{M}\in\mathrmbf{Struc}(\mathcal{S})$,
$\mathcal{M}{\;\models_{\mathcal{S}}\;}(\varphi'{\;\xrightarrow{h}\;}\varphi)$,
is mapped by constraint satisfaction 
to the
$\mathrmbf{Tbl}(\mathcal{A})$-morphism 
$\mathrmbfit{T}_{\mathcal{M}}(\varphi')
=\mathcal{T}'={\langle{I',s',\mathrmbfit{K}(\varphi'),t_{\varphi'}}\rangle}
\xleftarrow{\langle{h,k}\rangle}
{\langle{I,s,\mathrmbfit{K}(\varphi),t_{\varphi}}\rangle}=
\mathcal{T}=\mathrmbfit{T}_{\mathcal{M}}(\varphi)$
with indexing $X$-sorted signature morphism 
${\langle{I',s'}\rangle}
\xrightarrow{\;h\;}
{\langle{I,s}\rangle}$ 
and 
a key function $\mathrmbfit{K}(\varphi')\xleftarrow{\,k\,}\mathrmbfit{K}(\varphi)$
%
\footnote{Defined by choice.
However,
the choice for a composite signature morphism 
is not identical to the composition of the component choices;
only equivalent.}
%
satisfying either of the adjoint fiber morphisms
in Disp.\ref{adj:fbr:cxt}
leading to the naturality condition 
{{$k{\;\cdot\;}t_{\varphi'}=
t_{\varphi}{\;\cdot\;}\mathrmbfit{tup}_{\mathcal{A}}(h)$}}
and
visualized in Fig.\,\ref{fig:rel:tbl:refl}.
Reflectively,
the constraint
$\varphi'{\;\xrightarrow{h}\;}\varphi$
is mapped to an $\mathcal{A}$-relation morphism
$\mathrmbfit{R}_{\mathcal{M}}(\varphi')
=\mathcal{R}'={\langle{I',s',
{\wp}t'_{\varphi'}(\mathrmbfit{K}(\varphi'))}\rangle}
\xleftarrow{\langle{h,r}\rangle}
{\langle{I,s,{\wp}t_{\varphi}(\mathrmbfit{K}(\varphi))}\rangle}=
\mathcal{R}=\mathrmbfit{R}_{\mathcal{M}}(\varphi)$
satisfying either of the adjoint fiber orderings
in Disp.\ref{adj:fbr:ord}
leading to the naturality condition 
{{$r{\;\cdot\;}i_{\varphi'}=
i_{\varphi}{\;\cdot\;}\mathrmbfit{tup}_{\mathcal{A}}(h)$}}
and visualized in Fig.\,\ref{fig:rel:tbl:refl}.
%
\mbox{}\hfill\rule{5pt}{5pt}
\end{description}
\comment{
The interpretation of two composable constraints is mapped, 
by \underline{choice} of key functions, 
to the composition of the individual table morphisms.}
\end{proof}
%



%
\newpage
\section{{\ttfamily FOLE} Structures.}\label{sec:struc}

%
We define two forms of structures and structure morphisms:
a classification (strict) form and an interpretation (lax) form.
%

\subsection{Structures.}\label{sub:sec:struc}

\comment{\fbox{{Add previous struc-cf-db equiv:}}}
\comment{\fbox{{non-lax structure morphisms with surjective predicate passage $\mathrmbf{R}_{2}\xrightarrow{\;\mathrmbfit{R}\;}\mathrmbf{R}_{1}$.}}}


Assume 
that we are given a schema 
$\mathcal{S} = {\langle{R,\sigma,X}\rangle}$
with a set of predicate symbols $R$
and a signature (header) map 
$R \xrightarrow{\sigma} \mathrmbf{List}(X)$.
\begin{definition}\label{def:struc}
An $\mathcal{S}$-structure 
$\mathcal{M} = 
{\langle{\mathcal{E},\sigma,\tau,\mathcal{A}}\rangle}$
(LHS Fig.\,\ref{fig:fole:struc})
consists of 
an entity classification
$\mathcal{E} = {\langle{R,K,\models_{\mathcal{E}}}\rangle}$
of predicate symbols $R$ and keys $K$,
an attribute classification (typed domain)
$\mathcal{A} = {\langle{X,Y,\models_{\mathcal{A}}}\rangle}$
of sorts $X$ and data values $Y$, and 
a list designation ${\langle{\sigma,\tau}\rangle} : \mathcal{E} \rightrightarrows \mathrmbf{List}(\mathcal{A})$
with signature map $R \xrightarrow{\sigma} \mathrmbf{List}(X)$,
tuple map $K \xrightarrow{\tau} \mathrmbf{List}(Y)$
and \underline{defining condition}
$k{\;\models_{\mathcal{E}}\;}r$
implies
$\tau(k){\;\models_{\mathrmbf{List}(\mathcal{A})}\;}\sigma(r)$.
\footnote{In more detail,
if entity $k{\,\in\,}K$ is of type $r{\,\in\,}R$,
then the description tuple $\tau(k)={\langle{J,t}\rangle}$ 
is the same ``size'' ($J=I$) as
the signature $\sigma(r)={\langle{I,s}\rangle}$ 
and for each index $n \in I$ the data value $t_{n}$ is of sort $s_{n}$.}
\end{definition}
Projections define
the type hypergraph (schema)
$\mathcal{S} = {\langle{R,\sigma,X}\rangle}$,
and
an instance hypergraph (universe)
$\mathcal{U} = {\langle{K,\tau,Y}\rangle}$.
%
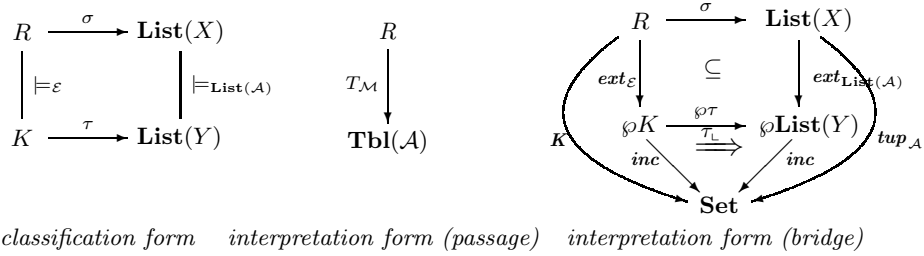
\begin{figure}
\begin{center}
\begin{tabular}{@{\hspace{0pt}}c@{\hspace{10pt}}c@{\hspace{10pt}}c}
{{\begin{tabular}{c}
\begin{tabular}{c}
\setlength{\unitlength}{0.5pt}
\begin{picture}(140,120)(-10,-15)
\put(0,80){\makebox(0,0){\footnotesize{$R$}}}
\put(0,0){\makebox(0,0){\footnotesize{$K$}}}
\put(87,80){\makebox(0,0)[l]{\footnotesize{$\mathrmbf{List}(X)$}}}
\put(87,0){\makebox(0,0)[l]{\footnotesize{$\mathrmbf{List}(Y)$}}}
\put(8,40){\makebox(0,0)[l]{\scriptsize{$\models_{\mathcal{E}}$}}}
\put(128,40){\makebox(0,0)[l]{\scriptsize{$\models_{\mathrmbf{List}(\mathcal{A})}$}}}
\put(50,90){\makebox(0,0){\scriptsize{$\sigma$}}}
\put(50,10){\makebox(0,0){\scriptsize{$\tau$}}}
\put(20,80){\vector(1,0){60}}
\put(20,0){\vector(1,0){60}}
\put(0,65){\line(0,-1){50}}
\put(120,65){\line(0,-1){50}}
%
\end{picture}
\end{tabular}
\\\\
\end{tabular}}}
&
{{\begin{tabular}{c}
\setlength{\unitlength}{0.5pt}
\begin{picture}(80,120)(0,15)
\put(40,122){\makebox(0,0){\footnotesize{$\mathrmit{R}$}}}
\put(40,38){\makebox(0,0){\footnotesize{$\mathrmbf{Tbl}(\mathcal{A})$}}}
\put(34,82){\makebox(0,0)[r]{\scriptsize{$T_{\mathcal{M}}$}}}
\put(40,108){\vector(0,-1){56}}
%
%
\end{picture}
\end{tabular}}}
&
{{\begin{tabular}{c}
\setlength{\unitlength}{0.5pt}
\begin{picture}(180,120)(-30,-35)
\put(0,80){\makebox(0,0){\footnotesize{$R$}}}
\put(128,80){\makebox(0,0){\footnotesize{$\mathrmbf{List}(X)$}}}
\put(0,0){\makebox(0,0){\footnotesize{${\wp}K$}}}
\put(128,0){\makebox(0,0){\footnotesize{${\wp}\mathrmbf{List}(Y)$}}}
\put(60,-60){\makebox(0,0){\footnotesize{$\mathrmbf{Set}$}}}
\put(-1,35){\makebox(0,0)[r]{\scriptsize{$
{{\mathrmbfit{ext}_{\mathcal{E}}}}
$}}}
\put(130,35){\makebox(0,0)[l]{\scriptsize{$
{{\mathrmbfit{ext}_{\mathrmbf{List}(\mathcal{A})}}}
$}}}
\put(-54,-10){\makebox(0,0)[r]{\scriptsize{$\mathrmbfit{K}$}}}
\put(179,-10){\makebox(0,0)[l]{\scriptsize{$\mathrmbfit{tup}_{\mathcal{A}}$}}}
\put(50,90){\makebox(0,0){\scriptsize{$\sigma$}}}
\put(50,10){\makebox(0,0){\scriptsize{${\wp}\tau$}}}
\put(55,43){\makebox(0,0){{$\subseteq$}}}
\put(60,-19){\makebox(0,0){\large{$\xRightarrow{\tau_{{\scriptscriptstyle{\llcorner}}}\;}$}}}
\put(15,-25){\makebox(0,0)[r]{\scriptsize{$\mathrmbfit{inc}$}}}
\put(110,-25){\makebox(0,0)[l]{\scriptsize{$\mathrmbfit{inc}$}}}
\put(20,80){\vector(1,0){60}}
\put(20,0){\vector(1,0){60}}
\put(5,-10){\vector(1,-1){40}}
\put(115,-10){\vector(-1,-1){40}}
\put(0,65){\vector(0,-1){50}}
\put(120,65){\vector(0,-1){50}}
\qbezier(-17,67)(-120,-10)(35,-57)\put(35,-57){\vector(4,-1){0}}
\qbezier(137,67)(240,-10)(85,-57)\put(85,-57){\vector(-4,-1){0}}
\end{picture}
\end{tabular}}}
\\&&\\
{\footnotesize{\itshape{classification form}}}
&
{\footnotesize{\itshape{interpretation form (passage)}}}
&
{\footnotesize{\itshape{interpretation form (bridge)}}}
\end{tabular}
\end{center}
\caption{{\ttfamily FOLE} Structure}
\label{fig:fole:struc}
\end{figure}
\mbox{}\newline
By dropping the global key set $K$,
the entity classification 
$\mathcal{E} = {\langle{R,K,\models_{\mathcal{E}}}\rangle}$
can 
be regarded 
(RHS Fig.\,\ref{fig:fole:struc})
as a set-valued function
$\mathrmbfit{K} = 
R \xrightarrow[\mathrmbfit{K}]{\mathrmbfit{ext}_{\mathcal{E}}} {\wp}K\subseteq\mathrmbf{Set}$;
that is,
as an indexed collection 
$\mathcal{E}={\langle{R,\mathrmbfit{K}}\rangle}$
of subsets of keys
$\bigl\{ \mathrmbfit{ext}_{\mathcal{E}}(r)
= \mathrmbfit{K}(r) \subseteq \mathrmbf{Set} \mid r \in R \bigr\}$.
The defining condition
of the list designation 
demonstrates that a structure
can 
be regarded 
(MID Fig.\,\ref{fig:fole:struc})
as a table-valued function
$R \xrightarrow {T_{\mathcal{M}}} \mathrmbf{Tbl}(\mathcal{A})$;
that is,
as an $R$-indexed collection of 
$\mathcal{A}$-tables
$T_{\mathcal{M}}(r)={\langle{\sigma(r),\mathrmbfit{K}(r),\tau_{r}}\rangle}$,
whose tuple maps
$\bigl\{
\mathrmbfit{K}(r)\xrightarrow{\tau_{r}}\mathrmbfit{tup}_{\mathcal{A}}(\sigma(r))
\mid r \in R \bigr\}$
are restrictions 
of the tuple map $K \xrightarrow{\tau} \mathrmbf{List}(Y)$.
Here, 
an entity type 
(predicate symbol)
$r \in R$ 
is \underline{interpreted} 
as 
the
$\mathcal{A}$-table
$T_{\mathcal{M}}(r) = {\langle{\sigma(r),\mathrmbfit{K}(r),\tau_{r}}\rangle} \in \mathrmbf{Tbl}(\mathcal{A})$,
consisting of 
the $X$-signature $\sigma(r)
\in \mathrmbf{List}(X)$,
the key set $\mathrmbfit{K}(r) = \mathrmbfit{ext}_{\mathcal{E}}(r) \subseteq K$, and 
the 
tuple function $\mathrmbfit{K}(r)\xrightarrow{\tau_{r}}\mathrmbfit{tup}_{\mathcal{A}}(\sigma(r))$,
which is a restriction of the tuple function $K\xrightarrow{\tau}\mathrmbf{List}(Y)$.
The list classification 
$\mathrmbf{List}(\mathcal{A}) = 
{\langle{\mathrmbf{List}(X),\mathrmbf{List}(Y),\models_{\mathrmbf{List}(\mathcal{A})}}\rangle}$
can alternatively be regarded 
(RHS Fig.\,\ref{fig:fole:struc})
(Fig.\,4
of the paper
\cite{kent:fole:era:found})
as a set-valued function
$\mathrmbf{List}(X) \xrightarrow
[\mathrmbfit{tup}_{\mathcal{A}}]
{\mathrmbfit{ext}_{\mathrmbf{List}(\mathcal{A})}} {\wp}\mathrmbf{List}(Y)$;
that is,
as an indexed collection 
$\mathrmbf{List}(\mathcal{A}) =
{\langle{\mathrmbf{List}(X),\mathrmbfit{tup}_{\mathcal{A}},\mathrmbf{List}(Y)}\rangle}$
of subsets of tuples
$\bigl\{ 
\mathrmbfit{ext}_{\mathrmbf{List}(\mathcal{A})}(I,s) = 
\mathrmbfit{tup}_{\mathcal{A}}(I,s)
\subseteq \mathrmbf{List}(Y)
\mid (I,s) \in \mathrmbf{List}(X) \bigr\}$.
%
The bridge
$K\xRightarrow{\tau_{{\scriptscriptstyle{\llcorner}}}\;}{\wp}\tau{\;\circ\;}\mathrmbfit{inc}$
provides ``restrictions'' of the bridge
$K\xRightarrow{\;\tau\;}\mathrmbf{List}(Y)$
to subsets of $K$.
%

%
\newpage
\subsection{Structure Morphisms.}\label{sub:sec:struc:mor}




Assume that we are given a schema morphism
(Tbl.\,\ref{tbl:fole:morph}
in \S\,\ref{sub:sec:adj:components})
%
\comment{A (strict) schema morphism has an identity bridge 
${\sigma_{2}}{\;\cdot\;}{\scriptstyle\sum}_{f}{\;\xRightarrow[=]{\,\grave{\varphi}\;\,}\;}{r}{\;\cdot\;}{\sigma_{1}}$.}
%
\begin{equation}\label{struc:mor:assume}
\mathcal{S}_{2}={\langle{R_{2},\sigma_{2},X_{2}}\rangle} 
\xRightarrow{\;{\langle{r,\grave{\varphi},f}\rangle}\;}
{\langle{R_{1},\sigma_{1},X_{1}}\rangle}=\mathcal{S}_{1}
\end{equation}
%
%
consisting of 
a function on relation symbols
$R_{2}\xrightarrow{\,\mathrmit{r}\;}R_{1}$,
a sort function
$X_{2}\xrightarrow{f}\mathcal{X}_{1}$,
and
a schema bridge
$\sigma_{2}{\;\cdot\;}{\scriptstyle\sum}_{f}
{\;\xRightarrow{\,\grave{\varphi}\;\,}\;}r{\;\cdot\;}\sigma_{1}$,
whose $r_{2}^{\text{th}}$-component is the morphism
{{$
{{\scriptstyle\sum}_{f}(\sigma_{2}(r_{2}))}
\xrightarrow[h]{\grave{\varphi}_{r_{2}}}
{\sigma_{1}(r(r_{2}))}$}}
in $\mathrmbf{List}(X_{1})$.
\begin{definition}\label{def:struc:mor}
%
A structure morphism
$\mathcal{M}_{2}\xrightleftharpoons{{\langle{r,k,\grave{\varphi},\beta,f,g}\rangle}}\mathcal{M}_{1}$
from source structure 
$\mathcal{M}_{2}
={\langle{\mathcal{E}_{2},{\langle{\sigma_{2},\tau_{2}}\rangle},\mathcal{A}_{2}}\rangle}$
to target structure
$\mathcal{M}_{1}
={\langle{\mathcal{E}_{1},{\langle{\sigma_{1},\tau_{1}}\rangle},\mathcal{A}_{1}}\rangle}$
along 
(top Fig.\;\ref{fig:fole:struc:mor})
a schema morphism (Disp.\ref{struc:mor:assume})
is a list designation morphism
consisting of 
(back Fig.\;\ref{fig:fole:struc:mor})
an entity infomorphism 
$\mathcal{E}_{2}\xrightleftharpoons{{\langle{r,k}\rangle}}\mathcal{E}_{1}$,
and
(front Fig.\;\ref{fig:fole:struc:mor})
an attribute infomorphism 
$\mathcal{A}_{2}\xrightleftharpoons{{\langle{f,g}\rangle}}\mathcal{A}_{1}$.
These are bridged by
the above schema morphism 
and 
(bottom Fig.\;\ref{fig:fole:struc:mor})
a universe morphism
$\mathcal{U}_{2}\xLeftarrow{\;{\langle{k,\beta,g}\rangle}\;}\mathcal{U}_{1}$
with tuple bridge
${k{\;\cdot\;}\tau_{2}{\;\xRightarrow{\,\beta\,\;}\;}\tau_{1}{\;\cdot\;}{\scriptstyle\sum}_{g}}$,
whose $k_{1}^{\text{th}}$-component is the morphism
{{$
{\tau_{2}(k(k_{1}))}
\xrightarrow[h]{\beta_{k_{1}}}
{{\scriptstyle\sum}_{g}(\tau_{1}(k_{1}))}$}}
in $\mathrmbf{List}(Y_{2})$,
%
with equal arity
\begin{center}
${\footnotesize{
{\underset{{(I_{2}\xrightarrow[h]{\grave{\varphi}_{r_{2}}}I_{1})\;=\;(J_{2}\xrightarrow[h]{\beta_{k_{1}}}J_{1})}}
{\underbrace{
\grave{\varphi}_{r_{2}}{\;\circ\;}\mathrmbfit{arity}_{X_{1}}=\beta_{k_{1}}{\;\circ\;}\mathrmbfit{arity}_{Y_{2}}}}}
}}$
\;\underline{when}\;
${\footnotesize{
{\underset{{
k_{1}{\;\models_{{\langle{r,k}\rangle}}\;}r_{2}}}{\underbrace{k(k_{1}){\;\models_{\mathcal{E}_{2}}\;}r_{2}
\text{ \underline{iff} }
k_{1}{\;\models_{\mathcal{E}_{1}}\;}r(r_{2})}}}
}}$.
\end{center}
%
\end{definition}
%
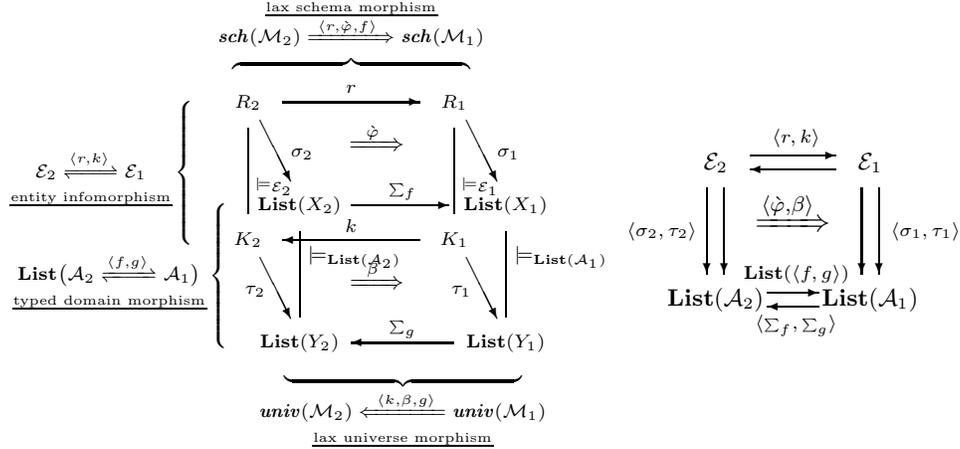
\begin{figure}
\begin{center}
{{\begin{tabular}{c@{\hspace{15pt}}c}
{{\begin{tabular}{c}
\setlength{\unitlength}{0.65pt}
\begin{picture}(320,250)(-125,-55)
\put(-30,60){\begin{picture}(0,0)(0,0)
\put(0,80){\makebox(0,0){\scriptsize{$R_{2}$}}}
\put(120,80){\makebox(0,0){\scriptsize{$R_{1}$}}}
\put(0,0){\makebox(0,0){\scriptsize{$K_{2}$}}}
\put(120,0){\makebox(0,0){\scriptsize{$K_{1}$}}}
\put(60,88){\makebox(0,0){\scriptsize{$r$}}}
\put(60,8){\makebox(0,0){\scriptsize{$k$}}}
\put(4,31){\makebox(0,0)[l]{\scriptsize{$\scriptscriptstyle\models_{\mathcal{E}_{2}}$}}}
\put(124,31){\makebox(0,0)[l]{\scriptsize{$\scriptscriptstyle\models_{\mathcal{E}_{1}}$}}}
\put(20,80){\vector(1,0){80}}
\put(100,0){\vector(-1,0){80}}
\put(0,65){\line(0,-1){50}}
\put(120,65){\line(0,-1){50}}
\end{picture}}
\put(0,0){\begin{picture}(0,0)(0,0)
\put(0,80){\makebox(0,0){\scriptsize{$\mathrmbf{List}(X_{2})$}}}
\put(120,80){\makebox(0,0){\scriptsize{$\mathrmbf{List}(X_{1})$}}}
\put(0,0){\makebox(0,0){\scriptsize{$\mathrmbf{List}(Y_{2})$}}}
\put(120,0){\makebox(0,0){\scriptsize{$\mathrmbf{List}(Y_{1})$}}}
\put(60,88){\makebox(0,0){\scriptsize{${\scriptstyle\sum}_{f}$}}}
\put(60,8){\makebox(0,0){\scriptsize{${\scriptstyle\sum}_{g}$}}}
\put(45,115){\makebox(0,0){\footnotesize{$\xRightarrow{\;\;\grave{\varphi}\;\;}$}}}
%
%
\put(5,50){\makebox(0,0)[l]{\scriptsize{$\models_{\mathrmbf{List}(\mathcal{A}_{2})}$}}}
\put(125,50){\makebox(0,0)[l]{\scriptsize{$\models_{\mathrmbf{List}(\mathcal{A}_{1})}$}}}
\put(30,80){\vector(1,0){57}}
\put(90,0){\vector(-1,0){60}}
\put(0,65){\line(0,-1){50}}
\put(120,65){\line(0,-1){50}}
\end{picture}}
\put(-5,110){\makebox(0,0)[l]{\scriptsize{$\sigma_{2}$}}}
\put(115,110){\makebox(0,0)[l]{\scriptsize{$\sigma_{1}$}}}
\put(-20,30){\makebox(0,0)[r]{\scriptsize{$\tau_{2}$}}}
\put(100,30){\makebox(0,0)[r]{\scriptsize{$\tau_{1}$}}}
\put(-23,130){\vector(1,-2){18}}
\put(97,130){\vector(1,-2){18}}
\put(-23,50){\vector(1,-2){18}}
\put(97,50){\vector(1,-2){18}}
\put(45,35){\makebox(0,0){\footnotesize{$\xRightarrow{\;\;\beta\;\;}$}}}
\put(30,180){\makebox(0,0){\scriptsize{$\overset{\textstyle{\stackrel{\text{\emph{lax schema morphism}}}{
\mathrmbfit{sch}(\mathcal{M}_{2})\xRightarrow{\;{\langle{r,\grave{\varphi},f}\rangle}\;}\mathrmbfit{sch}(\mathcal{M}_{1})}}}{\overbrace{\rule{90pt}{0pt}}}$}}}
\put(-60,100){\makebox(0,0)[r]{\scriptsize{$
\underset{\scriptstyle{\text{\emph{entity infomorphism}}}}
{\mathcal{E}_{2}\xrightleftharpoons{{\langle{r,k}\rangle}}\mathcal{E}_{1}}
\left\{\rule{0pt}{30pt}\right.
$}}}
\put(-40,40){\makebox(0,0)[r]{\scriptsize{$
\underset{\scriptstyle{\text{\emph{typed domain morphism}}}}
{\mathrmbf{List}\bigl(\mathcal{A}_{2}\xrightleftharpoons{{\langle{f,g}\rangle}}\mathcal{A}_{1}\bigr)}
\left\{\rule{0pt}{30pt}\right.
$}}}
\put(60,-40){\makebox(0,0){\scriptsize{$\underset{\textstyle{\stackrel{\textstyle{
\mathrmbfit{univ}(\mathcal{M}_{2})\xLeftarrow{\;{\langle{k,\beta,g}\rangle}\;}\mathrmbfit{univ}(\mathcal{M}_{1})
}}{\scriptstyle{\text{\emph{lax universe morphism}}}}}}{\underbrace{\rule{90pt}{0pt}}}$}}}
\end{picture}
\end{tabular}}}
&
{{\begin{tabular}{c}
\setlength{\unitlength}{0.65pt}
\begin{picture}(140,120)(-15,-15)
\put(0,80){\makebox(0,0){\footnotesize{$\mathcal{E}_{2}$}}}
\put(0,0){\makebox(0,0){\footnotesize{$\mathrmbf{List}(\mathcal{A}_{2})$}}}
\put(90,80){\makebox(0,0){\footnotesize{$\mathcal{E}_{1}$}}}
\put(90,0){\makebox(0,0){\footnotesize{$\mathrmbf{List}(\mathcal{A}_{1})$}}}
\put(-11,40){\makebox(0,0)[r]{\scriptsize{${\langle{\sigma_{2},\tau_{2}}\rangle}$}}}
\put(102,40){\makebox(0,0)[l]{\scriptsize{${\langle{\sigma_{1},\tau_{1}}\rangle}$}}}
\put(47,94){\makebox(0,0){\scriptsize{${\langle{r,k}\rangle}$}}}
\put(47,15){\makebox(0,0){\scriptsize{${\mathrmbf{List}({\langle{f,g}\rangle})}$}}}
\put(47,-15){\makebox(0,0){\scriptsize{${\langle{{\scriptscriptstyle\sum}_{f},{\scriptscriptstyle\sum}_{g}}\rangle}$}}}
%
\put(45,45){\makebox(0,0){\large{$\xRightarrow{{\langle{\grave{\varphi},\beta}\rangle}}$}}}
%
\put(20,84){\vector(1,0){50}}
\put(70,76){\vector(-1,0){50}}
\put(30,4){\vector(1,0){30}}
\put(60,-4){\vector(-1,0){30}}
\put(-5,65){\vector(0,-1){50}}
\put(5,65){\vector(0,-1){50}}
\put(85,65){\vector(0,-1){50}}
\put(95,65){\vector(0,-1){50}}
\end{picture}
\end{tabular}}}
\end{tabular}}}
\end{center}
\caption{Structure Morphism}
\label{fig:fole:struc:mor}
\end{figure}
Let 
$\mathrmbf{Struc}$
denote the mathematical context of structures and their morphisms.
%

%
\newpage
\subsection{Transition.}\label{sub:sec:trans}


%
In this section (\S\,\ref{sec:struc}),
we define two forms of structures and structure morphisms:
a classification (strict) form 
%
and an interpretation (lax) form.
%
Fig.\,\ref{fig:fole:struc:mor}
illustrates the classification (strict) form; 
whereas
Fig.\,\ref{fig:lax:struc:mor}
illustrates the interpretation (lax) form.
Here we define a transition from the strict form to the lax form.
%


%
\begin{note}\label{note:lax:ent:info}
By forgetting the full key sets and the full key function,
the entity infomorphism becomes a \textbf{(lax)} entity infomorphism
$\overset{\mathrmbfit{K}_{2}}{\mathrmbfit{ext}_{\mathcal{E}_{2}}}
\xLeftarrow{\;\,\kappa\;}
{r}{\;\circ\;}
\overset{\mathrmbfit{K}_{1}}{\mathrmbfit{ext}_{\mathcal{E}_{1}}}$
consisting of
an $R_{2}$-indexed collection of key functions 
\newline\mbox{}\hfill
{\footnotesize{$\bigl\{
\overset{{\mathrmbfit{K}_{2}}}
{\mathrmbfit{ext}_{\mathcal{E}_{2}}}(r_{2})
\xleftarrow[k]{\kappa_{r_{2}}}
\overset{{\mathrmbfit{K}_{1}}}
{\mathrmbfit{ext}_{\mathcal{E}_{1}}}(r(r_{2}))
\mid r_{2} \in R_{2}
\bigr\}$.}}
\hfill\mbox{}\newline
Map
$k_{1}{\,\in\,}\overset{{\mathrmbfit{K}_{1}}}{\mathrmbfit{ext}_{\mathcal{E}_{1}}}(r(r_{2}))$
with
$k_{1}{\;\models_{\mathcal{E}_{1}}\;}r(r_{2})$
to 
$\kappa_{r_{2}}(k_{1}) = k(k_{1}){\,\in\,}\overset{{\mathrmbfit{K}_{2}}}{\mathrmbfit{ext}_{\mathcal{E}_{2}}}(r_{2})$
with
$k(k_{1}){\;\models_{\mathcal{E}_{2}}\;}r_{2}$.
%
%
The collection of $\mathrmbf{List}(X_{1})$-signature morphisms
{\footnotesize{$\bigl\{
{\scriptstyle\sum}_{f}(\overset{\mathrmbfit{S}_{2}}{\sigma_{2}}(r_{2}))
\xrightarrow[h]{\;\grave{\varphi}_{r_{2}}\;}\overset{\mathrmbfit{S}_{1}}{\sigma_{1}}(r(r_{2}))
\mid r_{2} \in R_{2} \bigr\}$}}
defines the $\mathrmbf{List}(Y_{1})$-tuple functions
\newline\mbox{}\hfill
{\footnotesize{$\bigl\{
\mathrmbfit{tup}_{\mathcal{A}_{1}}({\scriptstyle\sum}_{f}(\sigma_{2}(r_{2})))
\xleftarrow
[\mathrmbfit{tup}_{\mathcal{A}_{1}}(h)]
{\;\mathrmbfit{tup}_{\mathcal{A}_{1}}(\grave{\varphi}_{r_{2}})\;}
\mathrmbfit{tup}_{\mathcal{A}_{1}}(\sigma_{1}(r(r_{2})))
\mid r_{2} \in R_{2}
\bigr\}$}}.
\hfill\mbox{}\newline
%
%
\begin{figure}
\begin{center}
{{\footnotesize{\begin{tabular}{c}
\setlength{\unitlength}{0.55pt}
\begin{picture}(220,220)(-40,-80)
\put(-60,140){\makebox(0,0){\scriptsize{$K_{2}$}}}
\put(-60,-80){\makebox(0,0){\scriptsize{$\mathrmbf{List}(Y_{2})$}}}
\put(180,140){\makebox(0,0){\scriptsize{$K_{1}$}}}
\put(180,-80){\makebox(0,0){\scriptsize{$\mathrmbf{List}(Y_{1})$}}}
\put(0,80){\makebox(0,0){\scriptsize{$\mathrmbfit{K}_{2}(r_{2})$}}}
\put(125,80){\makebox(0,0){\scriptsize{$\mathrmbfit{K}_{1}(r(r_{2}))$}}}
\put(-27,0){\makebox(0,0){\scriptsize{$\mathrmbfit{tup}_{\mathcal{A}_{2}}(\sigma_{2}(r_{2}))$}}}
\put(150,0){\makebox(0,0){\scriptsize{$\mathrmbfit{tup}_{\mathcal{A}_{1}}(\sigma_{1}(r(r_{2})))$}}}
\put(60,-40){\makebox(0,0){\scriptsize{$
\mathrmbfit{tup}_{\mathcal{A}_{1}}({\scriptstyle\sum}_{f}(\sigma_{2}(r_{2})))$}}}
\put(63,90){\makebox(0,0){\scriptsize{$\kappa_{r_{2}}$}}}
\put(63,12){\makebox(0,0){\tiny{$\mathrmbfit{tup}(\grave{\varphi}_{r_{2}},f,g)$}}}
\put(0,40){\makebox(0,0)[r]{\scriptsize{$\tau_{2,r_{2}}$}}}
\put(125,40){\makebox(0,0)[l]{\scriptsize{$\tau_{1,r(r_{2})}$}}}
\put(105,-22){\makebox(0,0)[l]
{\tiny{$\mathrmbfit{tup}_{\mathcal{A}_{1}}(\grave{\varphi}_{r_{2}})
=h{\,\cdot\,}{(\mbox{-})}$}}} 
\put(15,-22){\makebox(0,0)[r]{\tiny{$
{(\mbox{-})}{\,\cdot\,}g=
\grave{\tau}_{{\langle{f,g}\rangle}}(\sigma_{2}(r_{2}))$}}}
\put(60,150){\makebox(0,0){\scriptsize{$k$}}}
\put(60,-95){\makebox(0,0){\scriptsize{${\scriptstyle\sum}_{g}$}}}
\put(-65,40){\makebox(0,0)[r]{\scriptsize{$\tau_{2}$}}}
\put(185,40){\makebox(0,0)[l]{\scriptsize{$\tau_{1}$}}}
\put(60,120){\makebox(0,0){\footnotesize{$=$}}}
\put(-33,55){\makebox(0,0){\footnotesize{$=$}}}
\put(153,55){\makebox(0,0){\footnotesize{$=$}}}
\put(20,-65){\makebox(0,0){\footnotesize{$=$}}}
\put(131,-54){\makebox(0,0){$\xRightarrow{\tau_{\grave{\varphi}_{r_{2}}}}$}}
\put(60,35){\makebox(0,0){\footnotesize{$\xRightarrow{\;\;\beta\;\;}$}}}
\dottedline{3}(-60,29)(180,29)
\put(60,57.5){\makebox(0,0){\footnotesize{$\xRightarrow{\,\beta_{r_{2}}\,}$}}}
\dottedline{3}(0,50)(120,50)
%
\put(0,65){\vector(0,-1){50}}
\put(120,65){\vector(0,-1){50}}
\put(80,80){\vector(-1,0){50}}
\put(90,0){\vector(-1,0){60}}
\put(45,-30){\vector(-3,2){30}}
\put(105,-10){\vector(-3,-2){30}}
\put(-60,125){\vector(0,-1){180}}
\put(155,140){\vector(-1,0){190}}
\put(140,-80){\vector(-1,0){160}}
\put(180,125){\vector(0,-1){180}}
\thinlines
\dottedline{3}(-10,95)(-40,125)\put(-40,125){\vector(-1,1){0}}
\dottedline{3}(140,95)(170,125)\put(170,125){\vector(1,1){0}}
\dottedline{3}(-16,-15)(-39,-61)\put(-39,-61){\vector(-1,-2){0}}
\dottedline{3}(136,-15)(159,-61)\put(159,-61){\vector(1,-2){0}}
\dottedline{3}(75,-50)(145,-73)\put(145,-73){\vector(3,-1){0}}
\end{picture}
\end{tabular}}}}
\end{center}
\caption{Transition}
\label{fig:transition}
\end{figure}
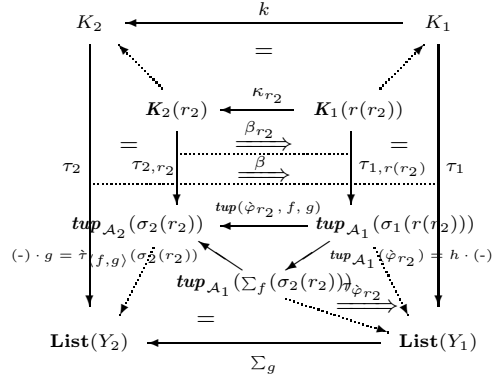
%
For any source predicate $r_{2} \in R_{2}$,
there is a bridge
$\kappa_{r_{2}}{\,\cdot\,}\bar{\tau}_{2,r_{2}}
\xRightarrow{\;\beta_{r_{2}}\;}
\bar{\tau}_{1,r(r_{2})}{\,\cdot\,}{\scriptstyle\sum}_{g}$
that is a restriction of the bridge
$k{\;\cdot\;}\tau_{2}{\;\xRightarrow{\,\beta\,\;}\;}\tau_{1}{\;\cdot\;}{\scriptstyle\sum}_{g}$
to the subset
$\mathrmbfit{K}_{1}(r(r_{2})){\;\subseteq\;}K_{1}$,
where 
$\bar{\tau}_{2,r_{2}} = \tau_{2,r_{2}}{\;\cdot\;}\mathrmbfit{inc}$ and
$\bar{\tau}_{1,r(r_{2})} = \tau_{1,r(r_{2})}{\;\cdot\;}\mathrmbfit{inc}$.
For each key
$k_{1}{\in}\mathrmbfit{K}_{1}(r(r_{2}))$,
the $k_{1}^{\text{th}}$-component of the bridge $\beta_{r_{2}}$ is the $\mathrmbf{List}(Y_{2})$-morphism
\newline\mbox{}\hfill
{{$\underset{\tau_{2}(k(k_{1}))}
{\bar{\tau}_{2,r_{2}}(\kappa_{r_{2}}(k_{1}))}
\xrightarrow[h]{\beta_{k_{1}}}
\underset{{\scriptstyle\sum}_{g}(\tau_{1}(k_{1}))}
{{\scriptstyle\sum}_{g}(\bar{\tau}_{1,r(r_{2})}(k_{1}))}$}}.
\footnote{Note that we have used the equal arity condition for a structure morphism
(Def\;\ref{def:struc:mor}):
$\grave{\varphi}_{r_{2}}{\;\circ\;}\mathrmbfit{arity}_{X_{1}}=\beta_{k_{1}}{\;\circ\;}\mathrmbfit{arity}_{Y_{2}}$
when
${\footnotesize{k(k_{1}){\;\models_{\mathcal{E}_{2}}\;}r_{2}
\text{ \underline{iff} }
k_{1}{\;\models_{\mathcal{E}_{1}}\;}r(r_{2})}}$.}
\hfill\mbox{}\newline
\end{note}
\begin{proposition}[Key]\label{prop:key}
Any structure morphism 
$\mathcal{M}_{2}\xrightleftharpoons{{\langle{r,k,\grave{\varphi},\beta,f,g}\rangle}}\mathcal{M}_{1}$
(Def\;\ref{def:struc:mor}) satisfies the following condition:
for any source predicate $r_{2} \in R_{2}$,
\newline\mbox{}\hfill
{\footnotesize{$
\kappa_{r_{2}}{\;\cdot\;}\tau_{2,r_{2}}
= 
\tau_{1,r(r_{2})}
{\;\cdot\;}
{\mathrmbfit{tup}_{\mathcal{A}_{1}}(
{\grave{\varphi}_{r_{2}}})}
{\;\cdot\;}
{\grave{\tau}_{{\langle{f,g}\rangle}}(
{\sigma_{2}}(r_{2}))} 
$.}}
%
\footnote{\label{tup:bridge}
The tuple bridge
$\mathrmbfit{tup}_{\mathcal{A}_{2}}
\xLeftarrow{\;\grave{\tau}_{{\langle{f,g}\rangle}}\;}
({\scriptstyle\sum}_{f})^{\mathrm{op}}\!{\circ\;}\mathrmbfit{tup}_{\mathcal{A}_{1}}$
is defined and used in the paper
``The {\ttfamily FOLE} Table''\cite{kent:fole:era:tbl}.}
%
\rule[6pt]{0pt}{10pt}
\hfill\mbox{}\newline
\end{proposition}
\begin{proof}
\begin{sloppypar}
For suppose that
$k_{1}{\,\in\,}\mathrmbfit{ext}_{\mathcal{E}_{1}}(r(r_{2}))$
with
$k_{2}=\kappa_{r_{2}}(k_{1})=k(k_{1}){\,\in\,}\mathrmbfit{ext}_{\mathcal{E}_{2}}(r_{2})$.
Define
$(I_{1},t_{1})=\tau_{1,r(r_{2})}(k_{1})=\tau_{1}(k_{1}){\,\in\,}\mathrmbfit{tup}_{\mathcal{A}_{1}}(\sigma_{1}(r(r_{2})))$
and
$(I_{2},t_{2})
=\tau_{2,r_{2}}(\kappa_{r_{2}}(k_{1}))
=\tau_{2}(k_{2}){\,\in\,}\mathrmbfit{tup}_{\mathcal{A}_{2}}(\sigma_{2}(r_{2}))$.
We have the $\mathrmbf{List}(Y_{2})$ morphism
$(I_{2},t_{2})=\tau_{2}(k(k_{1}))
\xrightarrow[h]{\beta_{k_{1}}}
{\scriptstyle\sum}_{g}(\tau_{1}(k_{1}))={\scriptstyle\sum}_{g}(I_{1},t_{1})$.
Thus,
$t_{2} = h{\,\cdot\,}t_{1}{\,\cdot\,}g$.
This means that
$(I_{2},t_{2})
= \mathrmbfit{tup}(h,f,g)(I_{1},t_{1})
= \grave{\tau}_{{\langle{f,g}\rangle}}(\sigma_{2}(r_{2}))(\mathrmbfit{tup}_{\mathcal{A}_{1}}(h)(I_{1},t_{1}))
= \grave{\tau}_{{\langle{f,g}\rangle}}(\sigma_{2}(r_{2}))(
\mathrmbfit{tup}_{\mathcal{A}_{1}}(\grave{\varphi}_{r_{2}})(I_{1},t_{1}))$.
Hence,
$\tau_{2,r_{2}}(\kappa_{r_{2}}(k_{1}))
= \grave{\tau}_{{\langle{f,g}\rangle}}(\sigma_{2}(r_{2}))(\mathrmbfit{tup}_{\mathcal{A}_{1}}(\grave{\varphi}_{r_{2}})(
\tau_{1,r(r_{2})}(k_{1})))
= \mathrmbfit{tup}(\grave{\varphi}_{r_{2}},f,g)(\tau_{1,r(r_{2})}(k_{1}))$,
for all
$k_{1}{\,\in\,}\mathrmbfit{ext}_{\mathcal{E}_{1}}(r(r_{2}))$.
%
\rule{5pt}{5pt}
\end{sloppypar}
\end{proof}
\begin{corollary}\label{cor:reduct:tup}
For $r_{2} \in R_{2}$,
each bridge 
$\kappa_{r_{2}}{\,\cdot\,}\bar{\tau}_{2,r_{2}}
\xRightarrow
{\;\beta_{r_{2}}\;}\bar{\tau}_{1,r(r_{2})}{\,\cdot\,}{\scriptstyle\sum}_{g}$
reduces to the composite
$\beta_{r_{2}} =
\tau_{1,r(r_{2})}{\;\cdot\;}\tau_{\grave{\varphi}_{r_{2}}}{\;\circ\;}{\scriptstyle\sum}_{g}$
for the bridge
$\mathrmbfit{tup}_{\mathcal{A}_{1}}(\grave{\varphi}_{r_{2}})
{\,\circ\,}\mathrmbfit{inc}
\xRightarrow{\tau_{\grave{\varphi}_{r_{2}}}}
\mathrmbfit{inc}$
associated with the function
$\mathrmbfit{tup}_{\mathcal{A}_{1}}({\scriptstyle\sum}_{f}(\sigma_{2}(r_{2})))
\xleftarrow[h{\,\cdot\,}{(\mbox{-})}]
{\mathrmbfit{tup}_{\mathcal{A}_{1}}(\grave{\varphi}_{r_{2}})}
\mathrmbfit{tup}_{\mathcal{A}_{1}}(\sigma_{1}(r(r_{2})))$.
\footnote{For tuple
$(I_{1},t_{1}){\,\in\,}\mathrmbfit{tup}_{\mathcal{A}_{1}}(\sigma_{1}(r(r_{2})))$,
the $(I_{1},t_{1})^{\text{th}}$ component 
of $\tau_{\grave{\varphi}_{r_{2}}}$
is the $\mathrmbf{List}(Y_{1})$-morphism 
$\mathrmbfit{tup}_{\mathcal{A}_{1}}(\grave{\varphi}_{r_{2}})(I_{1},t_{1})
=(I_{1},h{\,\cdot\,}t_{1})
\xrightarrow[h]{\tau_{\grave{\varphi}_{r_{2}}}(I_{1},t_{1})}(I_{1},t_{1})$.}
\end{corollary}
\begin{proof}
Straightforward.
\comment{
Let
$k_{1}{\in}\mathrmbfit{K}_{1}(r(r_{2}))$
be a key with associated tuple
\newline
$\tau_{1,r(r_{2})}(k_{1})
=(I_{1},t_{1})
{\,\in\,}\mathrmbfit{tup}_{\mathcal{A}_{1}}(\sigma_{1}(r(r_{2})))
{\;\subseteq\;}
\mathrmbf{List}(Y_{2})$.
\newline
Since
$\grave{\tau}_{{\langle{f,g}\rangle}}(\sigma_{2}(r_{2}))(
\mathrmbfit{tup}_{\mathcal{A}_{1}}(\grave{\varphi}_{r_{2}})(
\tau_{1,r(r_{2})}(k_{1})))
={\scriptstyle\sum}_{g}(\mathrmbfit{tup}_{\mathcal{A}_{1}}(\grave{\varphi}_{r_{2}})(\tau_{1}(k_{1})))$,
\newline
the $k_{1}^{\text{th}}$-component of the bridge $\beta_{r_{2}}$ is the $\mathrmbf{List}(Y_{2})$-morphism
\newline
${\scriptstyle\sum}_{g}(\mathrmbfit{tup}_{\mathcal{A}_{1}}(\grave{\varphi}_{r_{2}})(\tau_{1}(k_{1})))
\xrightarrow{\;h\;}
{{\scriptstyle\sum}_{g}(\tau_{1}(k_{1}))}$.
This is the ${\scriptstyle\sum}_{g}$-image of the $\mathrmbf{List}(Y_{1})$-morphism
$\mathrmbfit{tup}_{\mathcal{A}_{1}}(\grave{\varphi}_{r_{2}})(\tau_{1}(k_{1}))\xrightarrow{\;h\;}\tau_{1}(k_{1})$.
}
\mbox{}\hfill\rule{5pt}{5pt}
\end{proof}
%

%
\newpage
\subsection{Lax Structures.}\label{sub:sec:struc:lax}



%
To show equivalence between sound logics and databases,
we need to use lax structures.
%
\footnote{A structure becomes lax when we forget the global set of keys $K$
and use only a lax entity classification
$\mathcal{E}={\langle{R,\mathrmbfit{K}}\rangle}$.
We can retrieve a ``crisp'' entity classification by defining the disjoint union of key subsets
$\coprod_{r \in R} \mathrmbfit{K}(r)$.
But,
using the extent form of a structure (laxness),
we have lost a certain coordination of tuple functions
$\{
\mathrmbfit{ext}_{\mathcal{E}}(r)\xrightarrow{\tau_{r}}\mathrmbfit{tup}_{\mathcal{A}}(\sigma(r))
\mid r \in R \bigr\}$ here,
which may or may not be important.}
Again, 
assume that we are given a schema 
$\mathcal{S} = {\langle{R,\sigma,X}\rangle}$
with a set of predicate symbols $R$
and a signature (header) map 
$R \xrightarrow{\sigma} \mathrmbf{List}(X)$.
%
%
%
\begin{definition}\label{def:struc:lax}
A (lax) $\mathcal{S}$-structure $\mathcal{M} = {\langle{\mathcal{E},\sigma,\tau,\mathcal{A}}\rangle}$
consists of 
a lax entity classification
$\mathcal{E}={\langle{R,\mathrmbfit{K}}\rangle}$,
an attribute classification (typed domain)
$\mathcal{A} = {\langle{X,Y,\models_{\mathcal{A}}}\rangle}$,
the schema 
$\mathcal{S} = {\langle{R,\sigma,X}\rangle}$,
and one of the two equivalent descriptions:
%
\begin{itemize}
\item 
either
(MID Fig.\,\ref{fig:fole:struc})
a function
{{$R \xrightarrow{T_{\mathcal{M}}} \mathrmbf{Tbl}(\mathcal{A})$}}
consisting of an $R$-indexed collection of 
$\mathcal{A}$-tables
$T_{\mathcal{M}}(r)={\langle{\sigma(r),\mathrmbfit{K}(r),\tau_{r}}\rangle}$; 
\newline
\item 
or (RHS Fig.\,\ref{fig:fole:struc})
a tuple bridge
{{$\mathrmbfit{K}
\xRightarrow{\;\tau\;}\sigma{\;\circ\;}\mathrmbfit{tup}_{\mathcal{A}}$}}
consisting of 
an indexed collection of tuple functions 
$\{
\mathrmbfit{K}(r)\xrightarrow{\tau_{r}}\mathrmbfit{tup}_{\mathcal{A}}(\sigma(r))
\mid r \in R \}$.
\end{itemize}
\end{definition}
%
Hence,
we can think of a lax structure
$\mathcal{M} = {\langle{\mathcal{E},\sigma,\tau,\mathcal{A}}\rangle}$
as extending 
the {\ttfamily FOLE} schema
$R \xrightarrow{\,\sigma} \mathrmbf{List}(X)$
to the tabular interpretation function
$R \xrightarrow{\;T_{\mathcal{M}}\;} \mathrmbf{Tbl}(\mathcal{A})$.
%
%
%
\begin{proposition}\label{prop:struc:2:|db|}
A (lax) $\mathcal{S}$-structure 
$\mathcal{M} = {\langle{\mathcal{E},\sigma,\tau,\mathcal{A}}\rangle}$
defines,
is the same as,
the constraint-free aspect of a database
$\mathcal{R} = {\langle{R,\sigma,\mathcal{A},\mathrmbfit{K},\tau}\rangle}$
consisting of
a set of predicate types $R$,
a key function
$R\xrightarrow{\;\mathrmbfit{K}\;}\mathrmbf{Set}$, 
a schema $\mathcal{S} = {\langle{R,\sigma,X}\rangle}$ 
with
a signature function
$R\xrightarrow{\;\sigma\;}\mathrmbf{List}(X)$,
and
a (constraint-free)
tuple bridge
$\mathrmbfit{K}\xRightarrow{\;\tau\;}{S}\,{\,\circ\,}\mathrmbfit{tup}_{\mathcal{A}}$
consisting of
an $R$-indexed collection of tuple maps
$\bigl\{
\mathrmbfit{K}(r)\xrightarrow{\;\tau_{r}\;}\mathrmbfit{tup}_{\mathcal{A}}(\sigma(r))
\mid r \in R \bigr\}$.
\end{proposition}
\begin{proof}
See the previous discussion.
\end{proof}
\subsection{Lax Structure Morphisms.}\label{sub:sec:struc:mor:lax}



In order to make 
sound logic morphisms equivalent to 
database morphisms,
we need to eliminate the global key function
$K_{2}\xleftarrow{\;k\,}K_{1}$.
To do this we define a lax version of Def.\;\ref{def:struc:mor}.
Note\;\ref{note:lax:ent:info} eliminated the need for the global key function in the entity infomorphism>
Prop.\;\ref{prop:key} and Cor.\;\ref{cor:reduct:tup} 
eliminated the need for the universe morphism
with its global key function.
However,
we do need to include the condition of Prop.\;\ref{prop:key}. 

%
\newpage

\begin{definition}\label{def:lax:struc:mor}
For any two (lax) structures 
$\mathcal{M}_{2} = {\langle{\mathcal{E}_{2},\sigma_{2},\tau_{2},\mathcal{A}_{2}}\rangle}$
and
$\mathcal{M}_{1} = {\langle{\mathcal{E}_{1},\sigma_{1},\tau_{1},\mathcal{A}_{1}}\rangle}$,
a (lax) structure morphism
(Tbl.\,\ref{tbl:fole:morph}
in \S\,\ref{sub:sec:adj:components})
(Fig.~\ref{fig:lax:struc:mor})
\begin{center}
{\footnotesize{
$\mathcal{M}_{2} = {\langle{\mathcal{E}_{2},\sigma_{2},\tau_{2},\mathcal{A}_{2}}\rangle}
\xrightleftharpoons{{\langle{r,\kappa,\grave{\varphi},f,g}\rangle}}
{\langle{\mathcal{E}_{1},\sigma_{1},\tau_{1},\mathcal{A}_{1}}\rangle} = \mathcal{M}_{1}$
}}
\end{center}
%
consists of
(RHS Fig.\;\ref{fig:lax:struc:mor})
a schemed domain morphism 
\begin{center}
$\mathcal{D}_{2} = {\langle{R_{2},\sigma_{2},\mathcal{A}_{2}}\rangle} 
\xrightarrow{\;{\langle{\mathrmit{r},\grave{\varphi},f,g}\rangle}\;}
{\langle{R_{1},\sigma_{1},\mathcal{A}_{1}}\rangle} = \mathcal{D}_{1}$
\end{center}
consisting of
a
typed domain morphism
$\mathcal{A}_{2}\xrightleftharpoons{{\langle{f,g}\rangle}}\mathcal{A}_{1}$
and
a schema morphism 
$\mathcal{S}_{2} = {\langle{R_{2},\sigma_{2},X_{2}}\rangle} 
\xRightarrow{\;{\langle{\mathrmit{r},\grave{\varphi},f}\rangle}\;}
{\langle{R_{1},\sigma_{1},X_{1}}\rangle} = \mathcal{S}_{1}$
(Disp.\ref{struc:mor:assume} in \S\,\ref{sub:sec:struc:mor})
with common type function
$X_{2}\xrightarrow{\;\mathrmit{f}\;\,}X_{1}$,
%
\comment{
This has a signature bridge
${\sigma_{2}{\;\cdot\;}{\scriptstyle\sum}_{f}{\;\xRightarrow{\,\grave{\varphi}\;\,}\;}r{\;\cdot\;}\sigma_{1}}$
consisting of 
a collection of morphisms
$\bigl\{
{{\scriptstyle\sum}_{f}(\sigma_{2}(r_{2}))}\xrightarrow[h]{\grave{\varphi}_{r_{2}}}{\sigma_{1}(r(r_{2}))}
\mid r_{2} \in R_{2}
\bigr\}$
in $\mathrmbf{List}(X_{1})$.
}
%
and
(LHS Fig.\;\ref{fig:lax:struc:mor})
a 
lax entity infomorphism
\begin{center}
$\mathcal{E}_{2}
 = {\langle{R_{2},\mathrmbfit{K}_{2}}\rangle} 
\xleftharpoondown{{\langle{r,\kappa}\rangle}}
{\langle{R_{1},\mathrmbfit{K}_{1}}\rangle} = 
\mathcal{E}_{1}$
\end{center}
with 
a predicate function $R_{2}\xrightarrow{\;r\,}R_{1}$
and
a key bridge $\mathrmbfit{K}_{2}\xLeftarrow{\;\,\kappa\;}r{\,\circ\,}\mathrmbfit{K}_{1}$
consisting of 
the key functions
$\bigl\{{\mathrmbfit{K}_{2}}(r_{2})\xleftarrow{\kappa_{r_{2}}}{\mathrmbfit{K}_{1}}(r(r_{2}))
\mid r_{2} \in R_{2}
\bigr\}$,
%
which satisfy the condition
%
\begin{equation}\label{lax:struc:mor:cond}
{{
\Big\{
\kappa_{r_{2}}{\;\cdot\;}\tau_{2,r_{2}}
= 
\tau_{1,r(r_{2})}
{\;\cdot\;}{\mathrmbfit{tup}_{\mathcal{A}_{1}}({\grave{\varphi}_{r_{2}}})}
{\;\cdot\;}{\grave{\tau}_{{\langle{f,g}\rangle}}({\sigma_{2}}(r_{2}))} 
\mid r_{2} \in R_{2}
\Bigr\}
.}}
%
\footnote{
This is the constraint-free aspect of the database morphism condition
\newline\mbox{}\hfill
{\footnotesize$
\kappa{\;\bullet\;}\tau_{2}
= 
(\mathrmbfit{R}^{\mathrm{op}}\!{\circ\;}\tau_{1})
{\;\bullet\;}(\grave{\varphi}^{\mathrm{op}}\!{\circ\;}\mathrmbfit{tup}_{\mathcal{A}_{1}})
{\;\bullet\;}(\mathrmbfit{S}_{2}^{\mathrm{op}}{\;\circ\;}\grave{\tau}_{{\langle{f,g}\rangle}})
$.\normalsize}
\hfill\mbox{}\newline
See Def. \ref{def:db:mor:proj} in \S\,\ref{sub:sec:db:mor}
}
%
\end{equation}
See the (Key) Prop.\;\ref{prop:key}.
\end{definition}
\comment{
\begin{note}\label{note:lax:struc:mor:info}
Disp.\,\ref{lax:struc:mor:cond}
defines a table morphism
\newline\mbox{}\hfill
{\footnotesize{$
{\scriptstyle{
\overset{\textstyle{T_{2}(r_{2})}}
{\overbrace{\langle{\sigma_{2}(r_{2}),\mathcal{A}_{2},\mathrmbfit{K}_{2}(r_{2}),\tau_{2,r_{2}}}\rangle}}
\;\xleftarrow[{\langle{\grave{\varphi}_{r_{2}},f,g,\kappa_{r_{2}}}\rangle}]{\xi_{r_{2}}}\;
\overset{\textstyle{T_{1}(r(r_{2}))}}
{\overbrace{\langle{\sigma_{1}(r(r_{2})),\mathcal{A}_{1},\mathrmbfit{K}_{1}(r(r_{2})),\tau_{1,r(r_{2})}}\rangle}}
.}}
$}}
\end{note}
}
%
%
\begin{corollary}\label{cor:tbl:func}
%
\comment{When defining the database of a structure $\mathcal{M}$,
the entity classification
$\mathcal{E} = {\langle{R,K,\models_{\mathcal{E}}}\rangle}$ 
is required to be \emph{pseudo-partitioned}.
(see \S\;\ref{sec:intro})}
%
A (lax) structure morphism
$\mathcal{M}_{2}\xrightleftharpoons{{\langle{r,\kappa,\grave{\varphi},f,g}\rangle}}\mathcal{M}_{1}$
defines a tabular interpretation bridge function
$\mathrmit{R}_{2}\xrightarrow{\;\xi\;}\mathrmbf{mor}(\mathrmbf{Tbl})$,
which maps a predicate symbol
$r_{2}\in{\mathrmit{R}_{2}}$
with image
$r(r_{2})\in{\mathrmit{R}_{1}}$
to the table morphism
\begin{equation}\label{disp:lax:struc:mor:interp}
{\scriptstyle{
\overset{\textstyle{\mathrmbfit{T}_{2}(r_{2})}}
{\overbrace{\langle{\sigma_{2}(r_{2}),\mathcal{A}_{2},\mathrmbfit{K}_{2}(r_{2}),\tau_{2,r_{2}}}\rangle}}
\;\xleftarrow[{\langle{\grave{\varphi}_{r_{2}},f,g,\kappa_{r_{2}}}\rangle}]{\xi_{r_{2}}}\;
\overset{\textstyle{\mathrmbfit{T}_{1}(r(r_{2}))}}
{\overbrace{\langle{\sigma_{1}(r(r_{2})),\mathcal{A}_{1},\mathrmbfit{K}_{1}(r(r_{2})),\tau_{1,r(r_{2})}}\rangle}}
.}}
\end{equation}
visualized by the diagram
%
\begin{center}
{{\begin{tabular}{c}
\setlength{\unitlength}{0.66pt}
\begin{picture}(320,120)(-20,-30)
\put(-5,80){\makebox(0,0){\footnotesize{$\mathrmbfit{K}_{2}(r_{2})$}}}
\put(285,80){\makebox(0,0){\footnotesize{$\mathrmbfit{K}_{1}(r(r_{2}))$}}}
\put(-20,0){\makebox(0,0){\footnotesize{$\mathrmbfit{tup}_{\mathcal{A}_{2}}(\sigma_{2}(r_{2}))$}}}
\put(130,0){\makebox(0,0){\footnotesize{$
\mathrmbfit{tup}_{\mathcal{A}_{1}}({\scriptstyle\sum}_{f}(\sigma_{2}(r_{2})))$}}}
\put(280,0){\makebox(0,0){\footnotesize{$\mathrmbfit{tup}_{\mathcal{A}_{1}}(\sigma_{1}(r(r_{2})))$}}}
\put(140,90){\makebox(0,0){\scriptsize{$\kappa_{r_{2}}$}}}
\put(55,10){\makebox(0,0){\scriptsize{$\grave{\tau}_{{\langle{f,g}\rangle}}(\sigma_{2}(r_{2}))$}}}
\put(185,10){\makebox(0,0)[l]{\scriptsize{$\mathrmbfit{tup}_{\mathcal{A}_{1}}(\grave{\varphi}_{r_{2}})$}}} 
\put(180,-23){\makebox(0,0)[r]{\scriptsize{$\mathrmbfit{tup}(\grave{\varphi}_{r_{2}},f,g)$}}}
\put(-5,40){\makebox(0,0)[r]{\scriptsize{$\tau_{r_{2}}$}}}
\put(285,40){\makebox(0,0)[l]{\scriptsize{$\tau_{r_{1}}'$}}}
\put(245,80){\vector(-1,0){210}}
\put(70,0){\vector(-1,0){40}}
\put(230,0){\vector(-1,0){40}}
\put(0,65){\vector(0,-1){50}}
\put(280,65){\vector(0,-1){50}}
\put(10,-20){\oval(20,20)[bl]}
\put(10,-30){\line(1,0){260}}
\put(0,-14){\vector(0,1){0}}
\put(270,-20){\oval(20,20)[br]}
\put(-70,40){\makebox(0,0)[r]{\scriptsize{$
\stackrel{\textstyle{\mathrmbfit{source}}}{\mathrmbfit{table}}\left\{\rule{0pt}{33pt}\right.$}}}
\put(330,40){\makebox(0,0)[l]{\scriptsize{$
\left.\rule{0pt}{33pt}\right\}\stackrel{\textstyle{\mathrmbfit{target}}}{\mathrmbfit{table}}$}}}
\end{picture}
\end{tabular}}}
\end{center}
%
where
\begin{itemize}
\item 
{\footnotesize{$T_{2}(r_{2})=
{\langle{\sigma_{2}(r_{2}),\mathcal{A}_{2},\mathrmbfit{K}_{2}(r_{2}),\tau_{2,r_{2}}}\rangle}$}}
is a table defined by the
tabular interpretation function 
$R_{2} \xrightarrow{\;T_{2}\;} \mathrmbf{Tbl}(\mathcal{A}_{2})$
of structure $\mathcal{M}_{2}$,
and
\item 
{\footnotesize{$T_{1}(r(r_{2}))=
{\langle{\sigma_{1}(r(r_{2})),\mathcal{A}_{1},\mathrmbfit{K}_{1}(r(r_{2})),\tau_{1,r(r_{2})}}\rangle}$}}
is the table defined by the
tabular interpretation function 
$R_{1} \xrightarrow{\;T_{1}\;} \mathrmbf{Tbl}(\mathcal{A}_{1})$
of structure$\mathcal{M}_{1}$.
\end{itemize}
\end{corollary}
\begin{proof}
See Disp.\,\ref{lax:struc:mor:cond} above.
\mbox{}\hfill\rule{5pt}{5pt}
\end{proof}
Let 
$\mathring{\mathrmbf{Struc}}$
denote the 
mathematical 
context of (lax) structures and their morphisms.
Since
any structure is a lax structure and
any structure morphism is a lax structure morphism,
there is an passage 
$\mathrmbf{Struc}\xrightarrow{\;\mathrmbfit{lax}\;}\mathring{\mathrmbf{Struc}}$.


%

%
\begin{figure}
\begin{center}
%
{{\begin{tabular}{c}
\setlength{\unitlength}{0.66pt}
\begin{picture}(240,140)(0,-40)
\put(0,80){\makebox(0,0){\scriptsize{$r{\,\circ\;}\mathrmbfit{K}_{1}$}}}
\put(0,0){\makebox(0,0){\scriptsize{$\mathrmbfit{K}_{2}$}}}
\put(125,80){\makebox(0,0){\scriptsize{$
r{\,\circ\;}\sigma_{1}{\,\circ\;}\mathrmbfit{tup}_{\mathcal{A}_{1}}$}}}
\put(125,0){\makebox(0,0){\scriptsize{$
\sigma_{2}{\,\circ\;}\mathrmbfit{tup}_{\mathcal{A}_{2}}$}}}
\put(125,45){\makebox(0,0){\scriptsize{$({\grave{\varphi}}{\;\circ\;}\mathrmbfit{tup}_{\mathcal{A}_{1}})
{\;\bullet\;}(\sigma_{2}{\,\circ\;}\grave{\tau}_{{\langle{f,g}\rangle}})$}}}
\put(50,90){\makebox(0,0){\scriptsize{$r{\;\circ\;}\tau_{1}$}}}
\put(50,-10){\makebox(0,0){\scriptsize{$\tau_{2}$}}}
\put(-10,40){\makebox(0,0)[r]{\scriptsize{$\kappa$}}}
\put(50,80){\makebox(0,0){\large{$\xRightarrow{\;\;\;\;\;\;\;\;\;\;\;\;\;}$}}}
\put(50,0){\makebox(0,0){\large{$\xRightarrow{\;\;\;\;\;\;\;\;\;\;\;\;\;}$}}}
\put(0,40){\makebox(0,0){\large{$\bigg\Downarrow$}}}
\put(120,40){\makebox(0,0){\large{$\bigg\Downarrow$}}}
\put(224,80){\makebox(0,0){\scriptsize{$
{\scriptstyle\sum}_{f}{\,\circ\;}\mathrmbfit{tup}_{\mathcal{A}_{1}}$}}}
\put(220,0){\makebox(0,0){\scriptsize{$\mathrmbfit{tup}_{\mathcal{A}_{2}}$}}}
\put(230,45){\makebox(0,0)[l]{\scriptsize{$\grave{\tau}_{{\langle{f,g}\rangle}}$}}}
\put(220,40){\makebox(0,0){\large{$\bigg\Downarrow$}}}
\put(0,-20){\makebox(0,0){\scriptsize{\textit{lax entity}}}}
\put(0,-30){\makebox(0,0){\scriptsize{\textit{infomorphism}}}}
\put(220,-20){\makebox(0,0){\scriptsize{\textit{tuple attribute}}}}
\put(220,-30){\makebox(0,0){\scriptsize{\textit{infomorphism}}}}
\end{picture}
\end{tabular}}}
\end{center}
\caption{Lax Structure Morphism}
\label{fig:lax:struc:mor}
\end{figure}
\newpage
\section{{\ttfamily FOLE} Specifications.}\label{sec:spec}

%
A {\ttfamily FOLE} specification
$\mathcal{R} = {\langle{\mathrmbf{R},\mathrmbfit{S}}\rangle}$
consists of
a context $\mathrmbf{R}$ of predicates linked by constraints and  
a diagram 
$\mathrmbfit{S} : \mathrmbf{R} 
\rightarrow \mathrmbf{Set}$
of lists.
A specification morphism 
is a diagram morphism
${\langle{\mathrmbfit{R},\zeta}\rangle} : 
{\langle{\mathrmbf{R}_{2},\mathrmbfit{S}_{2}}\rangle} \Rightarrow 
{\langle{\mathrmbf{R}_{1},\mathrmbfit{S}_{1}}\rangle}$
consisting of a shape-changing passage 
$\mathrmbf{R}_{2} \xrightarrow{\:\mathrmbfit{R}\:} \mathrmbf{R}_{1}$ 
and
a bridge
$\zeta : \mathrmbfit{S}_{2} \Rightarrow \mathrmbfit{R} \circ \mathrmbfit{S}_{1}$. 
Composition is component-wise.
The mathematical context of 
{\ttfamily FOLE} specifications 
is denoted by 
$\mathrmbf{SPEC}
= \mathrmbf{List}^{\scriptscriptstyle{\Downarrow}}
= \bigl({(\mbox{-})}{\,\Downarrow\,}\mathrmbf{List}\bigr)$.
The 
subcontext of 
{\ttfamily FOLE} specifications
(with fixed sort sets
and fixed sort functions) 
is denoted by $\mathrmbf{Spec}\subseteq\mathrmbf{SPEC}$.

%
%
%

%
\subsection{Specifications.}\label{sub:sec:spec}
%
Assume 
that we are given a schema 
$\mathcal{S} = {\langle{R,\sigma,X}\rangle}$
with a set of predicate symbols $R$
and a signature (header) map 
$R \xrightarrow{\sigma} \mathrmbf{List}(X)$.
%


\begin{definition}\label{def:fml:spec}
A (formal) $\mathcal{S}$-specification is a subgraph 
$\mathrmbf{R}{\;\sqsubseteq\;}\mathrmbf{Cons}(\mathcal{S})$,
whose nodes are $\mathcal{S}$-formulas and whose edges are $\mathcal{S}$-constraints.
%
\footnote{We can place axiomatic restrictions on specifications
in various manners.
A  {\ttfamily FOLE} specification requires entailment to be a preorder,
satisfying reflexivity and transitivity.
It could also require satisfaction of sufficient axioms 
(Tbl.\,3
of the paper \cite{kent:fole:era:supstruc})
to described the various logical operations 
(connectives, quantifiers, etc.) 
used to build formulas in first-order logic.}
%
\end{definition}
Let $\mathrmbf{Spec}(\mathcal{S})={\wp}\mathrmbf{Cons}(\mathcal{S})$ 
denote the set of all $\mathcal{S}$-specifications.
%
\footnote{For any graph $\mathcal{G}$,
${\wp}\mathcal{G} = {\langle{{\wp}\mathcal{G},\sqsubseteq}\rangle}$ 
denotes the power preorder of all subgraphs of $\mathcal{G}$.}


%
\begin{definition}\label{def:abs:spec}
An (abstract) $\mathcal{S}$-specification
$\mathcal{T} = {\langle{\mathrmbf{R},\mathrmbfit{S},X}\rangle}$
%
\footnote{An abstract specification is also known as a database schema
(Def.\,\ref{def:db} in \S\,\ref{sub:sec:db:obj}).}
%
consists of:
a context $\mathrmbf{R}$, 
whose objects $r \in R$ are predicate symbols
and whose arrows $r'\xrightarrow{\,p\,}r$
are called \underline{abstract} $\mathcal{S}$-constraints, and
a passage
$\mathrmbf{R}\xrightarrow{\;\mathrmbfit{S}\;}\mathrmbf{List}(X)$,
which extends the signature map
${R}\xrightarrow{\;\sigma\;}\mathrmbf{List}(X)$
by mapping an 
abstract constraint $r'\xrightarrow{\,p\,}r$
to an $X$-signature morphism
$\mathrmbfit{S}(r')=\sigma(r')={\langle{I',s'}\rangle}
\xrightarrow[h]{\mathrmbfit{S}(p)=\sigma(p)}
{\langle{I,s}\rangle}=\sigma(r)=\mathrmbfit{S}(r)$.
%
\end{definition}
%
%
\comment{
%
An abstract $\mathcal{S}$-specification
$\mathcal{T} = {\langle{\mathrmbf{R},\mathrmbfit{S},X}\rangle}$
has a companion 
abstract $\mathcal{S}$-specification
$\mathcal{T}^{\ast} = {\langle{\mathrmbf{R}^{\ast},\mathrmbfit{S}^{\ast},X}\rangle}$
called its \emph{closure}.
%
$\mathrmbf{R}^{\ast}$
is the context generated by the graph $\mathrmbf{R}$,
whose objects are $\mathrmbf{R}$-nodes (predicate symbols), and 
whose morphisms are paths 
$r_{0}\xrightarrow{\,p_{0}\,}r_{1}
{\;...\;}
r_{n-1}\xrightarrow{\,p_{n-1}\,}r_{n}$
of $\mathrmbf{R}$-edges (paths of abstract constraints).
%
$\mathrmbf{R}^{\ast}\xrightarrow{\;\mathrmbfit{S}^{\ast}\;}\mathrmbf{List}(X)$
is a passage that extends the graph morphism
$\mathrmbf{R}\xrightarrow{\;\mathrmbfit{S}\;}|\mathrmbf{List}(X)|$
by composition of signature morphisms.
\mbox{}\newline
}
%
An abstract $\mathcal{S}$-specification
$\mathcal{T} = {\langle{\mathrmbf{R},\mathrmbfit{S},X}\rangle}$
has a companion formal $\mathcal{S}$-specification
$\widehat{\mathcal{T}} = {\langle{\widehat{\mathrmbf{R}},\widehat{\mathrmbfit{S}},X}\rangle}$,
whose graph
$\widehat{\mathrmbf{R}}
\sqsubseteq\mathrmbf{Cons}(\mathcal{S})$
is the set of all 
formal constraints
$\big\{ 
r'\xrightarrow[h]{\,\sigma(p)\,}r
\mid
r'\xrightarrow{\,p\,}r
\in \mathrmbf{R}
\big\}$.
%
\footnote{These are formal constraints,
since any predicate symbol is a formula.}
%
Since $\widehat{\mathrmbf{R}}$ is closed under composition and contains all identities,
$\widehat{\mathcal{T}}$ is an 
abstract $\mathcal{S}$-specification
with signature passage
$\widehat{\mathrmbf{R}}\xrightarrow{\widehat{\mathrmbfit{S}}}\mathrmbf{List}(X)$.
%
%

...

%
%

\comment{
Assume that we are given a schema morphism
\[\mbox{\footnotesize$
\mathcal{S}_{2}={\langle{R_{2},\sigma_{2},X_{2}}\rangle} 
\xrightarrow{\;{\langle{r,\grave{\varphi},f}\rangle}\;}
{\langle{R_{1},\sigma_{1},X_{1}}\rangle}=\mathcal{S}_{1}
$,\normalsize}\]
consisting of 
a function on relation symbols
$R_{2}\xrightarrow{\,\mathrmit{r}\;}R_{1}$,
a sort function
$X_{2}\xrightarrow{f}\mathcal{X}_{1}$,
and
a schema bridge
$\sigma_{2}{\;\cdot\;}{\scriptstyle\sum}_{f}
{\;\xRightarrow{\,\grave{\varphi}\;\,}\;}r{\;\cdot\;}\sigma_{1}$.
}

\subsection{Specification Morphisms.}\label{sub:sec:spec:mor}

A \texttt{FOLE} (abstract) specification morphism 
in $\mathrmbf{Spec}$,
with constant sort function
$f : X_{2} \rightarrow X_{1}$, 
is a diagram morphism 
${\langle{\mathrmbfit{R},\zeta}\rangle} : 
{\langle{\mathrmbf{R}_{2},\mathrmbfit{S}_{2}}\rangle} \rightarrow 
{\langle{\mathrmbf{R}_{1},\mathrmbfit{S}_{1}}\rangle}$
in $\mathrmbf{List}$
consisting of
a 
relation passage 
$\mathrmbfit{R} : \mathrmbf{R}_{2} \rightarrow \mathrmbf{R}_{1}$
and
a list interpretation bridge
{\footnotesize{${\mathrmbfit{S}_{2}
{\,\xRightarrow{\,\zeta\;\,}\,}
\mathrmbfit{R}{\,\circ\,}\mathrmbfit{S}_{1}}$}}
that factors
(Fig.\,\ref{fig:spec:mor:list})
\begin{equation}\label{eqn:spec:mor:def}
\zeta = (\grave{\varphi} \circ \mathrmbfit{inc}_{X_{1}}) 
\bullet (\mathrmbfit{S}_{2} \circ \grave{\iota}_{f})
\end{equation}
through the fiber adjunction 
$\mathrmbf{List}(X_{2})
\xleftarrow{\acute{\mathrmbfit{list}}_{f}\;\dashv\;\grave{\mathrmbfit{list}}_{f}}
\mathrmbf{List}(X_{1})$
%
\footnote{Fibered by signature over the adjunction
$\mathrmbf{List}(X_{2})
\xleftarrow
[{\langle{{\scriptscriptstyle\sum}_{f}{\;\dashv\;}f^{\ast}}\rangle}]
{{\langle{\acute{\mathrmbfit{list}}_{f}{\!\dashv\,}\grave{\mathrmbfit{list}}_{f}}\rangle}}
\mathrmbf{List}(X_{1})$
(Kent \cite{kent:fole:era:tbl})
representing list flow along a sort function
$f : X_{2} \rightarrow X_{1}$.}
%
in terms of 
\begin{itemize}
\item 
some bridge
$\grave{\varphi} : 
\mathrmbfit{S}_{2}\circ{\scriptstyle\sum}_{f}
\Rightarrow\mathrmbfit{R}^{\mathrm{op}}\circ\mathrmbfit{S}_{1}$
and 
\item 
the inclusion bridge
$\grave{\iota}_{f} : 
\mathrmbfit{inc}_{X_{2}}
\Rightarrow
{\scriptstyle\sum}_{f}\circ\mathrmbfit{inc}_{X_{1}}$.
\end{itemize}
%
We normally just use the bridge restriction $\grave{\varphi}$ for the specification morphism.
The original definition can be computed with the factorization in 
Disp.\,\ref{eqn:spec:mor:def}.
\comment{
Hence,
a specification morphism (with constant sort function) 
${\langle{\mathrmbfit{R},\grave{\varphi},f}\rangle} :
{\langle{\mathrmbf{R}_{2},\mathrmbfit{S}_{2},X_{2}}\rangle} \rightarrow
{\langle{\mathrmbf{R}_{1},\mathrmbfit{S}_{1},X_{1}}\rangle}$
consists of
a shape-changing relation passage 
$\mathrmbfit{R} : \mathrmbf{R}_{2} \rightarrow \mathrmbf{R}_{1}$,
a sort function
$\mathcal{X}_{2} \xrightarrow{\;f\;} X_{1}$, and
a bridge
$\grave{\varphi} : 
\mathrmbfit{S}_{2}\circ\grave{\mathrmbfit{list}}_{f}
\Rightarrow 
\mathrmbfit{R}^{\mathrm{op}}\circ\mathrmbfit{S}_{1}$.
}
%
\begin{definition}\label{def:abs:spec:mor}
An (abstract) specification morphism
(Tbl.\,\ref{tbl:fole:morph}
in \S\,\ref{sub:sec:adj:components})
%
\footnote{Visualized on the right side Fig.~\ref{fig:spec:mor:list}.}
%
%
\footnote{This is the same as a database schema morphism
of \S\,\ref{sub:sec:db:mor}.}
%
\newline\mbox{}\hfill
\rule[5pt]{0pt}{10pt}
{\footnotesize{$\mathcal{T}_{2}={\langle{\mathrmbf{R}_{2},\mathrmbfit{S}_{2},X_{2}}\rangle}
\xrightarrow{{\langle{\mathrmbfit{R},\grave{\varphi},f}\rangle}}
{\langle{\mathrmbf{R}_{1},\mathrmbfit{S}_{1},X_{1}}\rangle}=\mathcal{T}_{1}$,}}
\hfill\mbox{}\newline
along a schema morphism (Disp.\ref{struc:mor:assume} in \S\,\ref{sub:sec:struc:mor})
%
consists of 
a relation passage 
$\mathrmbf{R}_{2}\xrightarrow{\;\mathrmbfit{R}\;}\mathrmbf{R}_{1}$
extending the predicate function
$R_{2}\xrightarrow{\;r\;}R_{1}$ of the schema morphism to constraints,
the sort function $X_{2}\xrightarrow{\;f\;}X_{1}$ of the schema morphism,
and
a bridge
$\mathrmbfit{S}_{2}{\;\circ\;}{\scriptstyle\sum}_{f}\xRightarrow{\;\grave{\varphi}\;\,}
\mathrmbfit{R}{\;\circ\;}\mathrmbfit{S}_{1}$
extending the schema bridge to naturality.
%
\footnote{Naturality means that
for any source constraint
$r'_{2}\xrightarrow{\,p_{2}\,}r_{2}$
with target constraint
$\mathrmbfit{R}(r'_{2})=r(r'_{2})\xrightarrow[\mathrmbfit{R}(p_{2})]{\,p_{1}\,}r(r_{2})=\mathrmbfit{R}(r_{2})$,
if their signature morphisms are 
${\langle{I'_{2},s'_{2}}\rangle}\xrightarrow{\,h_{2}\,}{\langle{I_{2},s_{2}}\rangle}$
and
${\langle{I'_{1},s'_{1}}\rangle}\xrightarrow{\,h_{1}\,}{\langle{I_{1},s_{1}}\rangle}$,
we have the naturality diagram 
{\footnotesize{${\scriptstyle\sum}_{f}(h_{2}){\,\circ\,}\grave{\varphi}_{r_{2}}
=\grave{\varphi}_{r'_{2}}{\,\circ\,}h_{1}$.}}
}
%
%
\footnote{An example of a non-trivial bridge is projection:
assume that the arity functions 
$I'_{2}{\,\xhookrightarrow{\grave{\varphi}_{r'_{2}}}\,}I'_{1}$ and
$I_{2}{\,\xhookrightarrow{\grave{\varphi}_{r_{2}}}\,}I_{1}$
are inclusion,
thus defining signature projection, 
and assume that 
$I'_{2}\xrightarrow{\,h_{2}\,}I_{2}$
is a restriction of
$I'_{1}\xrightarrow{\,h_{1}\,}I_{1}$.}
%
\end{definition}
As noted before,
$\mathrmbf{Spec}\subseteq\mathrmbf{SPEC}$ denotes 
the mathematical context of abstract
{\ttfamily FOLE} specifications 
(with fixed sort sets).
%
\comment{When defining formalism (see later discussion),
we identify the consequence $\mathrmbf{T}^{\scriptstyle\bullet}$ of a formal specification $\mathrmbf{T}$
with the context of predicates/constraints $\mathrmbf{R}$
and we build strict specification morphisms from strict schema morphisms.}
%
This context 
is the same as 
the context of database schemas and their morphisms.
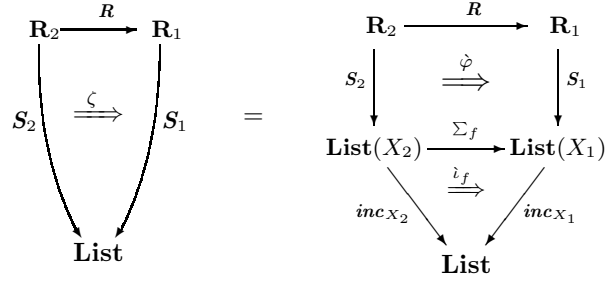
\begin{figure}
\begin{center}
{{\begin{tabular}{@{\hspace{5pt}}c@{\hspace{15pt}}c@{\hspace{5pt}}c@{\hspace{5pt}}}
{{\begin{tabular}[b]{c}
\setlength{\unitlength}{0.58pt}
\begin{picture}(80,160)(5,12)
\put(5,160){\makebox(0,0){\footnotesize{$\mathrmbf{R}_{2}$}}}
\put(85,160){\makebox(0,0){\footnotesize{$\mathrmbf{R}_{1}$}}}
\put(40,15){\makebox(0,0){\normalsize{$\mathrmbf{List}$}}}
\put(45,172){\makebox(0,0){\scriptsize{$\mathrmbfit{R}$}}}
\put(0,100){\makebox(0,0)[r]{\footnotesize{$\mathrmbfit{S}_{2}$}}}
\put(81,100){\makebox(0,0)[l]{\footnotesize{$\mathrmbfit{S}_{1}$}}}
\put(40,105){\makebox(0,0){\shortstack{\normalsize{$\xRightarrow{\,\,\zeta\;\;}$}}}}
\put(15,160){\vector(1,0){50}}
\qbezier(,150)(0,80)(25,30)\put(29,24){\vector(2,-3){0}}
\qbezier(80,150)(80,80)(55,30)\put(51,24){\vector(-2,-3){0}}
\end{picture}
\end{tabular}}}
&
{{\begin{tabular}[b]{c}
\setlength{\unitlength}{0.58pt}
\begin{picture}(80,160)(0,0)
\put(20,90){\makebox(0,0){\normalsize{$=$}}}
\end{picture}
\end{tabular}}}
&
{{\begin{tabular}[b]{c}
\setlength{\unitlength}{0.58pt}
\begin{picture}(120,160)(0,10)
\put(5,160){\makebox(0,0){\footnotesize{$\mathrmbf{R}_{2}$}}}
\put(125,160){\makebox(0,0){\footnotesize{$\mathrmbf{R}_{1}$}}}
\put(0,80){\makebox(0,0){\footnotesize{$\mathrmbf{List}(X_{2})$}}}
\put(120,80){\makebox(0,0){\footnotesize{$\mathrmbf{List}(X_{1})$}}}
\put(60,5){\makebox(0,0){\normalsize{$\mathrmbf{List}$}}}
\put(65,172){\makebox(0,0){\scriptsize{$\mathrmbfit{R}$}}}
\put(-5,125){\makebox(0,0)[r]{\scriptsize{$\mathrmbfit{S}_{2}$}}}
\put(125,125){\makebox(0,0)[l]{\scriptsize{$\mathrmbfit{S}_{1}$}}}
\put(60,92){\makebox(0,0){\scriptsize{$
{\scriptscriptstyle\sum}_{f}
$}}}
\put(24,38){\makebox(0,0)[r]{\scriptsize{$\mathrmbfit{inc}_{X_{2}}$}}}
\put(97,38){\makebox(0,0)[l]{\scriptsize{$\mathrmbfit{inc}_{X_{1}}$}}}
\put(60,130){\makebox(0,0){\shortstack{
\scriptsize{$\grave{\varphi}$}\\\large{$\Longrightarrow$}}}}
\put(60,54){\makebox(0,0){\shortstack{\footnotesize{$
\xRightarrow{\grave{\iota}_{f}}$}}}}
\put(20,160){\vector(1,0){80}}
\put(35,80){\vector(1,0){50}}
\put(0,145){\vector(0,-1){50}}
\put(120,145){\vector(0,-1){50}}
\put(9,68){\vector(3,-4){38}}
\put(111,68){\vector(-3,-4){38}}
\end{picture}
\end{tabular}}}
\end{tabular}}}
\end{center}
\caption{Specification Morphism}
\label{fig:spec:mor:list}
\end{figure}

%
\mbox{}\newline
There is an (abstract) $\mathcal{S}$-specification morphism
(LHS Fig.\,\ref{fig:spec:pass} in \S\,\ref{sub:sec:spec:sat})
\newline\mbox{}\hfill
{{$\mathcal{T} = {\langle{\mathrmbf{R},\mathrmbfit{S},X}\rangle}
\xrightarrow{\;\mathrmbfit{R}\;}
{\langle{\widehat{\mathrmbf{R}},\widehat{\mathrmbfit{S}},X}\rangle} 
= \widehat{\mathcal{T}}$}}
\hfill\mbox{}\newline
with a object-identical passage
$\mathrmbf{R}\xrightarrow{\;\mathrmbfit{R}\;}\widehat{\mathrmbf{R}}$
that preserves signature
\newline
$\mathrmbf{R}\!\xrightarrow{\;\mathrmbfit{R}\;}
\widehat{\mathrmbf{R}}
\xrightarrow{\widehat{\mathrmbfit{S}}}
\mathrmbf{List}(X) =
\mathrmbf{R}\xrightarrow{\mathrmbfit{S}}\mathrmbf{List}(X)$.


%
%

%
\newpage
\subsection{Specification Satisfaction.}\label{sub:sec:spec:sat}


Assume 
that we are given a schema 
$\mathcal{S} = {\langle{R,\sigma,X}\rangle}$
with a set of predicate symbols $R$
and a signature (header) map 
$R \xrightarrow{\sigma} \mathrmbf{List}(X)$.
\begin{definition}\label{satis:fml:spec}
(formal satisfaction)
An $\mathcal{S}$-structure $\mathcal{M} \in \mathrmbf{Struc}(\mathcal{S})$
\emph{satisfies} 
a (formal) $\mathcal{S}$-specification 
$\mathcal{T} = {\langle{\mathrmbf{R},\mathrmbfit{S},X}\rangle}$,
symbolized
$\mathcal{M}{\;\models_{\mathcal{S}}\;}\mathcal{T}$,
when it satisfies every constraint in the specification:
$\mathcal{M}{\;\models_{\mathcal{S}}\;}\mathcal{T}$
\underline{iff}
$\mathcal{M}^{\mathcal{S}}\sqsupseteq\mathrmbf{R}$.
Hence,
$\mathcal{M}^{\mathcal{S}}$ is the largest 
and most specialized (formal) 
$\mathcal{S}$-specification satisfied by $\mathcal{M}$.
\footnote{So,
$\mathcal{M}^{\mathcal{S}}$ is not just a mathematical context, but also an $\mathcal{S}$-specification.}
\end{definition}
\begin{definition}\label{satis:abs:spec}
(abstract satisfaction)
%
An $\mathcal{S}$-structure $\mathcal{M}\in\mathrmbf{Struc}(\mathcal{S})$
\emph{satisfies} 
an abstract $\mathcal{S}$-constraint $r'{\;\xrightarrow{p}\;}r$
in 
$\mathrmbf{R}$,
symbolized by
$\mathcal{M}{\;\models_{\mathcal{S}}\;}(r'{\;\xrightarrow{p}\;}r)$,
when it satisfies the associated formal constraint
$r'\xrightarrow{\,\sigma(p)\,}r$
in $\widehat{\mathrmbf{R}}$.
%
An $\mathcal{S}$-structure $\mathcal{M} \in \mathrmbf{Struc}(\mathcal{S})$
\emph{satisfies} (is a model of) 
an abstract $\mathcal{S}$-specification 
$\mathcal{T} = {\langle{\mathrmbf{R},\mathrmbfit{S},X}\rangle}$,
symbolized
$\mathcal{M}{\;\models_{\mathcal{S}}\;}\mathcal{T}$,
when it satisfies every abstract constraint in the specification.
%
Equivalently,
an $\mathcal{S}$-structure $\mathcal{M} \in \mathrmbf{Struc}(\mathcal{S})$
\emph{satisfies} 
an abstract $\mathcal{S}$-specification 
$\mathcal{T} = {\langle{\mathrmbf{R},\mathrmbfit{S},X}\rangle}$
when it satisfies 
the companion formal specification
$\widehat{\mathcal{T}} =
{\langle{\widehat{\mathrmbf{R}},\widehat{\mathrmbfit{S}},X}\rangle}$:
$\mathcal{M}{\;\models_{\mathcal{S}}\;}\mathcal{T}$
\underline{iff}
$\mathcal{M}^{\mathcal{S}}\sqsupseteq\widehat{\mathrmbf{R}}$.
\end{definition}
%
%
%

%
\begin{definition}\label{def:abs:tbl:pass}
When 
$\mathcal{M}{\;\models_{\mathcal{S}}\;}\mathcal{T}$
holds,
the abstract table passage
${\mathrmbf{R}}^{\mathrm{op}}\!\!\xrightarrow{\;\mathrmbfit{T}\;}\mathrmbf{Rel}(\mathcal{A})$
is defined to be the composition
(RHS Fig.\,\ref{fig:spec:pass})
of 
the object-identical passage
$\mathrmbf{R}\!\xrightarrow{\;\mathrmbfit{R}\;}\widehat{\mathrmbf{R}}$
with
the ``inclusion''
$\widehat{\mathrmbf{R}}\sqsubseteq\mathcal{M}^{\mathcal{S}}$
and the 
relation interpretation passage
${\mathcal{M}^{\mathcal{S}}}^{\mathrm{op}}\!
\xrightarrow{\;\mathrmbfit{R}_{\mathcal{M}}\;}
\mathrmbf{Rel}(\mathcal{A})$
of Lem.\,\ref{lem:interp:pass}
in
\S\,\ref{sub:sec:sat}.
%
\footnote{The passage
${\mathcal{M}^{\mathcal{S}}}^{\mathrm{op}}\!
\xrightarrow{\;\mathrmbfit{R}_{\mathcal{M}}\;}
\mathrmbf{Rel}(\mathcal{A})$
was defined in
Lemma.\,\ref{lem:interp:pass}
of \S\,\ref{sub:sec:sat} 
on satisfaction.
}
%
\end{definition}
Note that the
the abstract table passage
${\mathrmbf{R}}^{\mathrm{op}}\!\!\xrightarrow{\;\mathrmbfit{T}\;}\mathrmbf{Rel}(\mathcal{A})$
extends the 
table-valued function
$R \xrightarrow {T_{\mathcal{M}}} \mathrmbf{Tbl}(\mathcal{A})
\xhookrightarrow{im_{\mathcal{A}}}
\mathrmbf{Rel}(\mathcal{A})$
of the structure $\mathcal{M}$
to constraints.


%
\begin{proposition}[Key]\label{sat:tbl:interp}
Satisfaction
is equivalent to tabular interpretation.
\end{proposition}
\begin{proof}
On the one hand,
if $\mathcal{M}{\;\models_{\mathcal{S}}\;}\mathcal{T}$,
then the tabular interpretation passage
${\mathrmbf{R}}^{\mathrm{op}}\!\!\xrightarrow{\;\mathrmbfit{T}\;}
\mathrmbf{Rel}(\mathcal{A})
\xhookrightarrow{\mathrmbfit{inc}_{\mathcal{A}}} 
\mathrmbf{Tbl}(\mathcal{A})$
maps 
a constraint $r'\xrightarrow{\,p\,}r$ in $\mathrmbf{R}$
to the $\mathcal{A}$-relation morphism
%
\newline\mbox{}\hfill
$\mathrmbfit{T}(r')
={\langle{\mathrmbfit{S}(r'),{\wp}t'_{r'}(\mathrmbfit{K}(r'))}\rangle}
\xleftarrow[{\langle{h_{p},r_{p}}\rangle}]{\mathrmbfit{T}(p)}
{\langle{\mathrmbfit{S}(r),{\wp}t_{r}(\mathrmbfit{K}(r))}\rangle}=
\mathrmbfit{T}(r)$
\hfill\mbox{}\newline
as pictured in
Fig.\,\ref{fig:rel:tbl:refl}
of
\S\ref{sub:sec:sat}. 
%
\footnote{
We use relations here rather than tables,
since
(Lem.\,\ref{lem:interp:pass}
of
\S\,\ref{sub:sec:sat}): 
relational interpretation is closed under composition; but
tabular interpretation is closed under composition only up to key equivalence.}
%

On the other hand,
if there is a tabular interpretation passage
${\mathrmbf{R}}^{\mathrm{op}}\!\!\xrightarrow{\;\mathrmbfit{T}\;}\mathrmbf{Tbl}(\mathcal{A})$,
then the adjoint flow of $\mathcal{A}$-tables
%
\footnote{See the discussion about \textsl{type domain indexing}
in \S\,3.4.1 of the paper
\cite{kent:fole:era:tbl}.}
%
demonstrates that 
the $\mathcal{S}$-structure $\mathcal{M}$
satisfies 
each abstract $\mathcal{S}$-constraint $r'{\;\xrightarrow{p}\;}r$
in $\mathrmbf{R}$
(see
Disp.\ref{adj:fbr:cxt}
in
\S\,\ref{sub:sec:sat}). 
\mbox{}\hfill\rule{5pt}{5pt}
\end{proof}


%
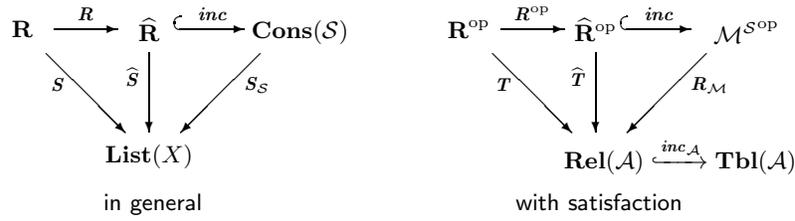
\begin{figure}
\begin{center}
{\footnotesize{
\setlength{\extrarowheight}{3.5pt}
\begin{tabular}{@{\hspace{5pt}}c@{\hspace{55pt}}c@{\hspace{15pt}}}
{{\begin{tabular}{c}
\setlength{\unitlength}{0.6pt}
\begin{picture}(180,120)(-5,-10)
\put(0,80){\makebox(0,0){\footnotesize{${\mathrmbf{R}}$}}}
\put(80,80){\makebox(0,0){\footnotesize{$\widehat{\mathrmbf{R}}$}}}
\put(175,78){\makebox(0,0){\footnotesize{$\mathrmbf{Cons}(\mathcal{S})$}}}
\put(80,0){\makebox(0,0){\footnotesize{$\mathrmbf{List}(X)$}}}
\put(40,90){\makebox(0,0){\scriptsize{$\mathrmbfit{R}$}}}
\put(120,90){\makebox(0,0){\scriptsize{$\mathrmbfit{inc}$}}}
\put(27,45){\makebox(0,0)[r]{\scriptsize{$\mathrmbfit{S}$}}}
\put(74,50){\makebox(0,0)[r]{\scriptsize{$\widehat{\mathrmbfit{S}}$}}}
\put(140,45){\makebox(0,0)[l]{\scriptsize{$\mathrmbfit{S}_{\mathcal{S}}$}}}
\put(20,80){\vector(1,0){40}}
\put(100,80){\vector(1,0){40}}\put(100,84){\oval(8,8)[l]}
\put(15,65){\vector(1,-1){50}}
\put(80,65){\vector(0,-1){50}}
\put(150,65){\vector(-1,-1){50}}
%
%
\end{picture}
\end{tabular}}}
&
{{\begin{tabular}{c}
\setlength{\unitlength}{0.6pt}
\begin{picture}(180,120)(-5,-10)
\put(0,80){\makebox(0,0){\footnotesize{${\mathrmbf{R}}^{\mathrm{op}}$}}}
\put(80,80){\makebox(0,0){\footnotesize{$\widehat{\mathrmbf{R}}^{\mathrm{op}}$}}}
\put(175,78){\makebox(0,0){\footnotesize{${\mathcal{M}^{\mathcal{S}}}^{\mathrm{op}}$}}}
\put(60,0){\makebox(0,0)[l]{\footnotesize{$
\mathrmbf{Rel}(\mathcal{A})
\xhookrightarrow{\mathrmbfit{inc}_{\mathcal{A}}} 
\mathrmbf{Tbl}(\mathcal{A})
$}}}
\put(40,90){\makebox(0,0){\scriptsize{$\mathrmbfit{R}^{\mathrm{op}}$}}}
\put(120,90){\makebox(0,0){\scriptsize{$\mathrmbfit{inc}$}}}
\put(27,45){\makebox(0,0)[r]{\scriptsize{$\mathrmbfit{T}$}}}
\put(74,50){\makebox(0,0)[r]{\scriptsize{$\widehat{\mathrmbfit{T}}$}}}
\put(140,45){\makebox(0,0)[l]{\scriptsize{$\mathrmbfit{R}_{\mathcal{M}}$}}}
\put(20,80){\vector(1,0){40}}
\put(100,80){\vector(1,0){40}}\put(100,84){\oval(8,8)[l]}
\put(15,65){\vector(1,-1){50}}
\put(80,65){\vector(0,-1){50}}
\put(150,65){\vector(-1,-1){50}}
%
%
\end{picture}
\end{tabular}}}
\\
{\textsf{in general}}
&
{\textsf{with satisfaction}}
\end{tabular}}}
\end{center}
\caption{Specification Passages}
\label{fig:spec:pass}
\end{figure}

%
%


%
\begin{table}
\begin{center}
{\footnotesize{
\setlength{\extrarowheight}{3.5pt}
\begin{tabular}{|@{\hspace{5pt}}r@{\hspace{15pt}:\hspace{15pt}}l@{\hspace{5pt}}|}
\multicolumn{2}{l}{\textsf{Specification}}
\\ \hline
(formal) $\mathcal{S}$-specification
-- 
subgraph 
& 
$\mathrmbf{R}{\;\sqsubseteq\;}\mathrmbf{Cons}(\mathcal{S})$
\\
signature passage
&
$\mathrmbf{Cons}(\mathcal{S})\rightarrow\mathrmbf{List}(X)$
\\
(abstract) $\mathcal{S}$-specification
--
passage
&
$\mathrmbf{R}\xrightarrow{\;\mathrmbfit{S}\;}\mathrmbf{List}(X)$
\\\hline
\multicolumn{2}{l}{\textsf{Sound Logic}}
\\ \hline
(formal) $\mathcal{S}$-specification
-- 
subgraph
& 
$\mathrmbf{R}{\;\sqsubseteq\;}\mathcal{M}^{\mathcal{S}}$
\\
table passage
&
${\mathcal{M}^{\mathcal{S}}}^{\mathrm{op}}
\!\!\xrightarrow{\;\mathrmbfit{R}_{\mathcal{M}}\;}
\mathrmbf{Rel}(\mathcal{A})
\xhookrightarrow{\mathrmbfit{inc}_{\mathcal{A}}} 
\mathrmbf{Tbl}(\mathcal{A})$
\\
(abstract) $\mathcal{S}$-specification
--
table passage
& 
${\mathrmbf{R}}^{\mathrm{op}}
\!\xrightarrow{\;\mathrmbfit{T}\;}
\mathrmbf{Rel}(\mathcal{A})
\xhookrightarrow{\mathrmbfit{inc}_{\mathcal{A}}} 
\mathrmbf{Tbl}(\mathcal{A})$
\\\hline
\end{tabular}}}
\end{center}
\caption{Comparisons}
\label{tbl:spec:compare}
\end{table}
%


\newpage
\section{{\ttfamily FOLE} Sound Logics.}\label{sec:log}

\subsection{Sound Logics.}\label{sub:sub:sec:log}

%
\begin{definition}\label{def:snd:log}
A (lax) sound logic
$\mathcal{L} = {\langle{\mathcal{S},\mathcal{M},\mathcal{T}}\rangle}$
with
a schema $\mathcal{S} = {\langle{R,\sigma,X}\rangle}$,
consists of
a (lax) $\mathcal{S}$-structure $\mathcal{M} = {\langle{\mathcal{E},\sigma,\tau,\mathcal{A}}\rangle}$ and
an (abstract) $\mathcal{S}$-specification
$\mathcal{T}={\langle{\mathrmbf{R},\mathrmbfit{S},X}\rangle}$,
%
where 
the structure $\mathcal{M}$
satisfies
the specification $\mathcal{T}$:\;
$\mathcal{M}{\;\models_{\mathcal{S}}\;}\mathcal{T}$.
\end{definition}
%


%
\begin{proposition}\label{prop:snd:log:2:db}
Any (lax) {\ttfamily FOLE} sound logic
$\mathcal{L} = {\langle{\mathcal{S},\mathcal{M},\mathcal{T}}\rangle}$
defines a {\ttfamily FOLE} relational database
with constant type domain
$\mathcal{R} = {\langle{\mathrmbf{R},\mathrmbfit{T},\mathcal{A}}\rangle}$.
%
\footnote{A (lax) {\ttfamily FOLE} sound logic
is described in terms of its components
in Tbl.\,\ref{tbl:snd:log}.
In \S\,\ref{sub:sec:db:obj},
a {\ttfamily FOLE} relational database
is defined in 
Def.\,\ref{def:db}
and 
is described in terms of its components
in Tbl.\,\ref{tbl:db}.}
%
\end{proposition}
\begin{proof}
By the definition of satisfaction 
(Def.\,\ref{satis:abs:spec} 
and Def.\,\ref{def:abs:tbl:pass} 
in \S~\ref{sub:sec:spec:sat}),
the tabular interpretation 
$R \xrightarrow{T} \mathrmbf{Tbl}(\mathcal{A})$
(Def.\,\ref{def:struc:lax} of 
\S\,\ref{sub:sec:struc:lax})
extends to a
table passage
$\mathrmbf{R}^{\text{op}}\!\xrightarrow{\mathrmbfit{T}}
\mathrmbf{Rel}(\mathcal{A})\subseteq
\mathrmbf{Tbl}(\mathcal{A})$,
%
\comment{
is an interpretation diagram of $\mathcal{A}$-tables,
which
extends 
to constraints,
mapping 
a constraint $r'\xrightarrow{p}r$ 
to the $\mathcal{A}$-relation morphism
(\textbf{Key} Prop.\,\ref{sat:tbl:interp} of 
\S\,\ref{sub:sec:struc:lax})
\newline\mbox{}\hfill
$\mathrmbfit{T}(r')
={\langle{\mathrmbfit{S}(r'),{\wp}t'(\mathrmbfit{K}(r'))}\rangle}
\xleftarrow[{\langle{h,\tilde{r}}\rangle}]{\mathrmbfit{T}(p)}
{\langle{\mathrmbfit{S}(r),{\wp}t(\mathrmbfit{K}(r))}\rangle}=
\mathrmbfit{T}(r)$,
\hfill\mbox{}\newline
}
%
which forms a relational database
$\mathcal{R} = {\langle{\mathrmbf{R},\mathrmbfit{T},\mathcal{A}}\rangle}$.
%
\mbox{}\hfill\rule{5pt}{5pt}
\end{proof}
\vspace{-10pt}

\begin{table}
\begin{center}
{\footnotesize{\setlength{\extrarowheight}{2pt}
\begin{tabular}{|@{\hspace{5pt}}r@{\hspace{20pt}}l@{\hspace{15pt}$$\hspace{0pt}}l
@{\hspace{5pt}}|}
\multicolumn{3}{l}{
$\mathcal{L} = {\langle{\mathcal{S},\mathcal{M},\mathcal{T}}\rangle}$ 
\hfill
\textsf{sound logic}
}
\\ \hline
schema 
& $\mathcal{S} = {\langle{R,\sigma,X}\rangle}$ 
& 
\\ \cline{1-1}
signature function
&
$R \xrightarrow
{\;\sigma\;}\mathrmbf{List}(X)$
&
\\\hline\hline
(lax) $\mathcal{S}$-structure 
& $\mathcal{M} = {\langle{\mathcal{E},\sigma,\tau,\mathcal{A}}\rangle}$ 
&
\\ \cline{1-1}
(lax) entity classification 
& $\mathcal{E}={\langle{R,\mathrmbfit{K}}\rangle}$
& 
\\ \cline{1-1}
attribute classification 
& $\mathcal{A} = {\langle{X,Y,\models_{\mathcal{A}}}\rangle}$
& 
\\ \cline{1-1}
tuple bridge
& $\mathrmbfit{K}\xRightarrow{\;\tau\;}\sigma{\,\circ\;}\mathrmbfit{tup}_{\mathcal{A}}$
& 
\\
table function
& $R \xrightarrow{\;T_{\mathcal{M}}\;} \mathrmbf{Tbl}(\mathcal{A})$
& 
\\\hline\hline 
$\mathcal{S}$-specification 
& $\mathcal{T}={\langle{\mathrmbf{R},\mathrmbfit{S},X}\rangle}$ 
& 
\\ \cline{1-1}
signature passage
&
$\mathrmbf{R}\xrightarrow{\;\mathrmbfit{S}\;}\mathrmbf{List}(X)$
&
\\\hline\hline
satisfaction 
&
$\mathcal{M}{\;\models_{\mathcal{S}}\;}\mathcal{T}$
&
\comment{
\\\hline
table morphism
&
$\mathrmbfit{T}(r')
\xleftarrow[{\langle{h_{p},r_{p}}\rangle}]{\mathrmbfit{T}(p)}
\mathrmbfit{T}(r)
\in
\mathrmbf{Rel}(\mathcal{A})$
&
\\
constraint
&
$r' \xrightarrow{p} r \in \mathrmbf{R}$
&
}
\\\hline
table passage
&
$\mathrmbf{R}^{\mathrm{op}}\!\xrightarrow{\;\mathrmbfit{T}\;}
\mathrmbf{Rel}(\mathcal{A})
\xhookrightarrow{\mathrmbfit{inc}_{\mathcal{A}}} 
\mathrmbf{Tbl}(\mathcal{A})$
&
\\\hline
\end{tabular}}}
\end{center}
\caption{Sound Logic}
\label{tbl:snd:log}
\end{table}


%
%

%
\newpage
\subsection{Sound Logic Morphisms.}\label{sub:sub:sec:log:mor}

%
\begin{definition}\label{def:snd:log:mor}
For any two sound logics 
$\mathcal{L}_{2} = {\langle{\mathcal{S}_{2},\mathcal{M}_{2},\mathcal{T}_{2}}\rangle}$
and
$\mathcal{L}_{1} = {\langle{\mathcal{S}_{1},\mathcal{M}_{1},\mathcal{T}_{1}}\rangle}$,
a (lax) sound logic morphism
(Tbl.\,\ref{tbl:fole:morph}
in \S\,\ref{sub:sec:adj:components})
\newline\mbox{}\hfill
\rule[5pt]{0pt}{10pt}
{\footnotesize{
$\mathcal{L}_{2}={\langle{\mathcal{S}_{2},\mathcal{M}_{2},\mathcal{T}_{2}}\rangle}
\xrightleftharpoons{{\langle{\mathrmbfit{R},\kappa,\grave{\varphi},f,g}\rangle}}
{\langle{\mathcal{S}_{1},\mathcal{M}_{1},\mathcal{T}_{1}}\rangle}=\mathcal{L}_{1}$
}}
\hfill\mbox{}\newline
\rule[2pt]{0pt}{10pt}
consists of
a (lax) structure morphism
\newline\mbox{}\hfill
\rule[2pt]{0pt}{10pt}
{\footnotesize{
$\mathcal{M}_{2}={\langle{\mathcal{E}_{2},{\langle{\sigma_{2},\tau_{2}}\rangle},\mathcal{A}_{2}}\rangle}
\xrightleftharpoons{{\langle{r,\kappa,\grave{\varphi},f,g}\rangle}}
{\langle{\mathcal{E}_{1},{\langle{\sigma_{1},\tau_{1}}\rangle},\mathcal{A}_{1}}\rangle}=\mathcal{M}_{1}$ and
}}\hfill\mbox{}\newline
\rule[2pt]{0pt}{10pt}
an (abstract) specification morphism
%
\newline\mbox{}\hfill
\rule[5pt]{0pt}{10pt}
{\footnotesize{
$\mathcal{T}_{2}={\langle{\mathrmbf{R}_{2},\mathrmbfit{S}_{2},X_{2}}\rangle}
\xrightarrow{{\langle{\mathrmbfit{R},\grave{\varphi},f}\rangle}}
{\langle{\mathrmbf{R}_{1},\mathrmbfit{S}_{1},X_{1}}\rangle}=\mathcal{T}_{1}$}}
\rule[5pt]{0pt}{10pt}
\hfill\mbox{}\newline
\rule[2pt]{0pt}{10pt}
along a common schema morphism
\newline\mbox{}\hfill
\rule[5pt]{0pt}{10pt}
{\footnotesize{$
\mathcal{S}_{2} = {\langle{R_{2},\sigma_{2},X_{2}}\rangle}
\xRightarrow{{\langle{r,\grave{\varphi},f}\rangle}}
{\langle{R_{1},\sigma_{1},X_{1}}\rangle} = \mathcal{S}_{1}$.}}
%
\comment{This schema morphism
(Disp.\,\ref{struc:mor:assume} in \S\,\ref{sub:sec:struc:mor})
consists of 
a function on relation symbols
$R_{2}\xrightarrow{\,\mathrmit{r}\;}R_{1}$,
a sort function
$X_{2}\xrightarrow{f}\mathcal{X}_{1}$,
and
a schema bridge
$\sigma_{2}{\;\cdot\;}{\scriptstyle\sum}_{f}
{\;\xRightarrow{\,\grave{\varphi}\;\,}\;}r{\;\cdot\;}\sigma_{1}$.}
%
\hfill\mbox{}
\rule[6pt]{0pt}{10pt}
%
\end{definition}
Let $\mathring{\mathrmbf{Snd}}$ denote the context of (lax) sound logics and their morphisms.
\begin{note}
At this point we know that a (lax) sound logic morphism
\newline\mbox{}\hfill
{\footnotesize{
$\mathcal{L}_{2}={\langle{\mathcal{S}_{2},\mathcal{M}_{2},\mathcal{T}_{2}}\rangle}
\xrightleftharpoons{{\langle{\mathrmbfit{R},\kappa,\grave{\varphi},f,g}\rangle}}
{\langle{\mathcal{S}_{1},\mathcal{M}_{1},\mathcal{T}_{1}}\rangle}=\mathcal{L}_{1}$
}}\hfill\mbox{}\newline
has 
tabular interpretation passages
$\mathrmbf{R}_{2}\xrightarrow{\;\mathrmbfit{T}_{2}\;}\mathrmbf{Rel}(\mathcal{A}_{2})\subseteq\mathrmbf{Tbl}(\mathcal{A}_{2})$
and
$\mathrmbf{R}_{1}\xrightarrow{\;\mathrmbfit{T}_{1}\;}\mathrmbf{Rel}(\mathcal{A}_{1})\subseteq\mathrmbf{Tbl}(\mathcal{A}_{1})$
at the source and target sound logics
and 
a predicate passage
$\mathrmbf{R}_{2}\xrightarrow{\;\mathrmbfit{R}\;}\mathrmbf{R}_{1}$.
We complete the picture in 
Prop.\;\ref{prop:sat:preserve}
by proving that
sound logic morphisms preserve satisfaction
in some way.
%
%
\comment{
by defining  the database morphism
\newline\mbox{}\hfill
\rule[5pt]{0pt}{10pt}
{\footnotesize{
$\mathcal{R}_{2}=
{\langle{\mathrmbf{R}_{2},\mathrmbfit{T}_{2}{\circ}\mathcal{A}_{2}}\rangle}
\xleftarrow
{{\langle{\mathrmbfit{R},\xi}\rangle}}
{\langle{\mathrmbf{R}_{1},\mathrmbfit{T}_{1}{\circ}\mathcal{A}_{1}}\rangle}
=\mathcal{R}_{1}$,
}}\hfill\mbox{}\newline
whose tabular interpretation bridge
{\footnotesize{${\mathrmbfit{T}_{2}
{\,\xLeftarrow{\;\,\xi\,}\,}
\mathrmbfit{R}{\,\circ\,}\mathrmbfit{T}_{1}}$}}
consists of the collection of table functions
$\mathrmit{R}_{2}\xrightarrow{\;\xi\;}\mathrmbf{mor}(\mathrmbf{Tbl})$
of Cor.\ref{cor:tbl:func}.
}
\end{note}
%

%
\begin{proposition}\label{prop:sat:preserve}
There is a (tabular interpretation) bridge
\begin{center}
{\footnotesize{$
\xi = (\grave{\psi} \circ \mathrmbfit{inc}_{\mathcal{A}_{1}}) 
\bullet (\mathrmbfit{T}_{2} \circ \grave{\chi}_{{\langle{f,g}\rangle}})
: \mathrmbfit{T}_{2}\Leftarrow\mathrmbfit{R}^{\mathrm{op}}\!{\,\circ\,}\mathrmbfit{T}_{1}$}}
\end{center}
that extends the collection of tabular interpretation bridge functions
\begin{center}
$\bigl\{
{\footnotesize{
\mathrmbfit{T}_{2}(r_{2})
\;\xleftarrow[{\langle{\grave{\varphi}_{r_{2}},f,g,\kappa_{r_{2}}}\rangle}]{\xi_{r_{2}}}\;
\mathrmbfit{T}_{1}(r(r_{2}))
}}
\mid r_{2} \in R_{2}
\bigr\}$
%
\end{center}
(see Cor.\;\ref{cor:tbl:func} of \S\,\ref{sub:sec:struc:mor:lax})
to constraints:
%
%
\comment{
\newline\mbox{}\hfill
{\footnotesize{$
{\scriptstyle{
\overset{\textstyle{\mathrmbfit{T}_{2}(r_{2})}}
{\overbrace{\langle{\sigma_{2}(r_{2}),\mathcal{A}_{2},\mathrmbfit{K}_{2}(r_{2}),\tau_{2,r_{2}}}\rangle}}
\;\xleftarrow[{\langle{\grave{\varphi}_{r_{2}},f,g,\kappa_{r_{2}}}\rangle}]{\xi_{r_{2}}}\;
\overset{\textstyle{\mathrmbfit{T}_{1}(r(r_{2}))}}
{\overbrace{\langle{\sigma_{1}(r(r_{2})),\mathcal{A}_{1},\mathrmbfit{K}_{1}(r(r_{2})),\tau_{1,r(r_{2})}}\rangle}}
.}}
$}}
\hfill\mbox{}\newline
(see Cor.\;\ref{cor:tbl:func} of \S\,\ref{sub:sec:struc:mor:lax})
}
for any source constraint $r'_{2}\xrightarrow{p_{2}}r_{2}$ in $\mathrmbf{R}_{2}$ 
with $\mathrmbfit{R}$-image target constraint $r'_{1}\xrightarrow{p_{1}}r_{1}$ in $\mathrmbf{R}_{1}$,
we have the following naturality diagram,
which represents ``preservation or linkage of satisfaction''.
%
\footnote{The sound logic morphism
$\mathcal{L}_{2}
\xrightleftharpoons{{\langle{\mathrmbfit{R},\kappa,\grave{\varphi},f,g}\rangle}}
\mathcal{L}_{1}$
maps 
the table morphism
$\mathrmbfit{T}_{1}(r'_{1})\xleftarrow{\mathrmbfit{T}_{1}(p_{1})}\mathrmbfit{T}_{1}(r_{1})$
representing the satisfaction
$\mathcal{M}_{1}{\;\models_{\mathcal{S}_{1}}\;}(r'_{1}{\;\xrightarrow{p_{1}}\;}r_{1})$
to the table morphism
$\mathrmbfit{T}_{2}(r'_{2})\xleftarrow{\mathrmbfit{T}_{2}(p_{2})}\mathrmbfit{T}_{2}(r_{2})$
representing the satisfaction
$\mathcal{M}_{2}{\;\models_{\mathcal{S}_{2}}\;}(r'_{2}{\;\xrightarrow{p_{2}}\;}r_{2})$.}
%
\begin{center}
{{\begin{tabular}{c}
\setlength{\unitlength}{0.54pt}
\begin{picture}(240,130)(-20,-15)
\put(0,100){\makebox(0,0){\footnotesize{${\mathrmbfit{T}_{2}(r'_{2})}$}}}
\put(180,100){\makebox(0,0){\footnotesize{${\mathrmbfit{T}_{1}(r'_{1})}$}}}
\put(-5,0){\makebox(0,0){\footnotesize{${\mathrmbfit{T}_{2}(r_{2})}$}}}
\put(180,0){\makebox(0,0){\footnotesize{${\mathrmbfit{T}_{1}(r_{1})}$}}}
%
\put(90,111){\makebox(0,0){\scriptsize{$\xi_{r'_{2}}$}}}
\put(90,88){\makebox(0,0){\scriptsize{${\langle{\grave{\varphi}_{r'_{2}},f,g,\kappa_{r'_{2}}}\rangle}$}}}
\put(90,10){\makebox(0,0){\scriptsize{$\xi_{r_{2}}$}}}
\put(90,-11){\makebox(0,0){\scriptsize{${\langle{\grave{\varphi}_{r_{2}},f,g,\kappa_{r_{2}}}\rangle}$}}}
\put(-6,50){\makebox(0,0)[r]{\scriptsize{${\mathrmbfit{T}_{2}(p_{2})}$}}}
\put(186,50){\makebox(0,0)[l]{\scriptsize{${\mathrmbfit{T}_{1}(p_{1})}$}}}
\put(150,100){\vector(-1,0){120}}
\put(150,0){\vector(-1,0){120}}
\put(0,20){\vector(0,1){60}}
\put(180,20){\vector(0,1){60}}
\put(90,55){\makebox(0,0){\scriptsize{$
\overset{\textstyle{{\mathrmbfit{T}_{2}
{\,\xLeftarrow{\;\,\xi\,}\,}
\mathrmbfit{R}{\,\circ\,}\mathrmbfit{T}_{1}}}}{\text{naturality}}$}}}
\put(-55,50){\makebox(0,0)[r]{\scriptsize{$
\stackrel{\textstyle{\mathrmbfit{source}}}{\mathrmbfit{interpretation}}\left\{\rule{0pt}{30pt}\right.$}}}
\put(235,50){\makebox(0,0)[l]{\scriptsize{$
\left.\rule{0pt}{30pt}\right\}\stackrel{\textstyle{\mathrmbfit{target}}}{\mathrmbfit{interpretation}}$}}}
\end{picture}
\end{tabular}}}
\end{center}
%
\end{proposition}

%

\newpage

\begin{proof}
The above diagram is expanded into more detail in Fig.~\ref{fig:nat:comb}.
%
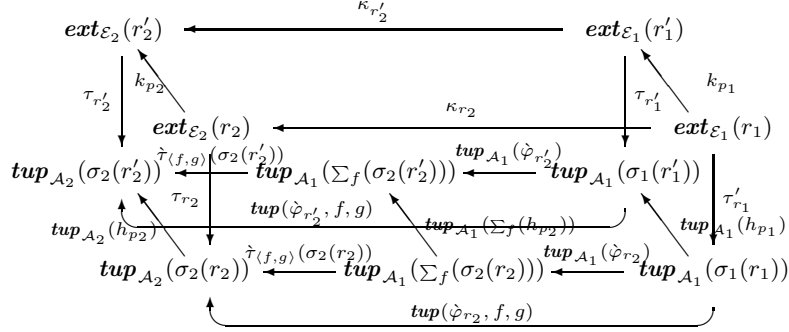
\begin{figure}
\begin{center}
{{\begin{tabular}{c}
\setlength{\unitlength}{0.68pt}
\begin{picture}(340,190)(0,-90)
\put(-5,80){\makebox(0,0){\footnotesize{$\mathrmbfit{ext}_{\mathcal{E}_{2}}(r'_{2})$}}}
\put(285,80){\makebox(0,0){\footnotesize{$\mathrmbfit{ext}_{\mathcal{E}_{1}}(r'_{1})$}}}
\put(-20,0){\makebox(0,0){\footnotesize{$\mathrmbfit{tup}_{\mathcal{A}_{2}}(\sigma_{2}(r'_{2}))$}}}
\put(130,0){\makebox(0,0){\footnotesize{$
\mathrmbfit{tup}_{\mathcal{A}_{1}}({\scriptstyle\sum}_{f}(\sigma_{2}(r'_{2})))$}}}
\put(280,0){\makebox(0,0){\footnotesize{$\mathrmbfit{tup}_{\mathcal{A}_{1}}(\sigma_{1}(r'_{1}))$}}}
\put(140,90){\makebox(0,0){\scriptsize{$\kappa_{r'_{2}}$}}}
\put(55,10){\makebox(0,0){\scriptsize{$\grave{\tau}_{{\langle{f,g}\rangle}}(\sigma_{2}(r'_{2}))$}}}
\put(185,10){\makebox(0,0)[l]{\scriptsize{$\mathrmbfit{tup}_{\mathcal{A}_{1}}(\grave{\varphi}_{r'_{2}})$}}} 
\put(140,-22){\makebox(0,0)[r]{\scriptsize{$\mathrmbfit{tup}(\grave{\varphi}_{r'_{2}},f,g)$}}}
\put(-5,40){\makebox(0,0)[r]{\scriptsize{$\tau_{r'_{2}}$}}}
\put(285,40){\makebox(0,0)[l]{\scriptsize{$\tau_{r'_{1}}$}}}
\put(245,80){\vector(-1,0){210}}
\put(70,0){\vector(-1,0){40}}
\put(230,0){\vector(-1,0){40}}
\put(0,65){\vector(0,-1){50}}
\put(280,65){\vector(0,-1){50}}
\put(10,-20){\oval(20,20)[bl]}
\put(10,-30){\line(1,0){260}}
\put(0,-14){\vector(0,1){0}}
\put(270,-20){\oval(20,20)[br]}
\put(25,50){\makebox(0,0)[r]{\scriptsize{$k_{p_{2}}$}}}
\put(325,50){\makebox(0,0)[l]{\scriptsize{$k_{p_{1}}$}}}
\put(20,-33){\makebox(0,0)[r]{\scriptsize{$\mathrmbfit{tup}_{\mathcal{A}_{2}}(h_{p_{2}})$}}}
\put(168,-30){\makebox(0,0)[l]{\scriptsize{$
\mathrmbfit{tup}_{\mathcal{A}_{1}}({\scriptstyle\sum}_{f}(h_{p_{2}}))$}}}
\put(310,-30){\makebox(0,0)[l]{\scriptsize{$\mathrmbfit{tup}_{\mathcal{A}_{1}}(h_{p_{1}})$}}}
\put(35,34){\vector(-3,4){26}}
\put(35,-46){\vector(-3,4){26}}
\put(315,34){\vector(-3,4){26}}
\put(315,-46){\vector(-3,4){26}}
\put(175,-46){\vector(-3,4){26}}
\put(50,-115){
\setlength{\unitlength}{0.68pt}
\begin{picture}(240,160)(10,-60)
\put(-5,80){\makebox(0,0){\footnotesize{$\mathrmbfit{ext}_{\mathcal{E}_{2}}(r_{2})$}}}
\put(285,80){\makebox(0,0){\footnotesize{$\mathrmbfit{ext}_{\mathcal{E}_{1}}(r_{1})$}}}
\put(-20,0){\makebox(0,0){\footnotesize{$\mathrmbfit{tup}_{\mathcal{A}_{2}}(\sigma_{2}(r_{2}))$}}}
\put(130,0){\makebox(0,0){\footnotesize{$
\mathrmbfit{tup}_{\mathcal{A}_{1}}({\scriptstyle\sum}_{f}(\sigma_{2}(r_{2})))$}}}
\put(280,0){\makebox(0,0){\footnotesize{$\mathrmbfit{tup}_{\mathcal{A}_{1}}(\sigma_{1}(r_{1}))$}}}
\put(140,90){\makebox(0,0){\scriptsize{$\kappa_{r_{2}}$}}}
\put(55,10){\makebox(0,0){\scriptsize{$\grave{\tau}_{{\langle{f,g}\rangle}}(\sigma_{2}(r_{2}))$}}}
\put(185,10){\makebox(0,0)[l]{\scriptsize{$\mathrmbfit{tup}_{\mathcal{A}_{1}}(\grave{\varphi}_{r_{2}})$}}} 
\put(180,-23){\makebox(0,0)[r]{\scriptsize{$\mathrmbfit{tup}(\grave{\varphi}_{r_{2}},f,g)$}}}
\put(-5,40){\makebox(0,0)[r]{\scriptsize{$\tau_{r_{2}}$}}}
\put(285,40){\makebox(0,0)[l]{\scriptsize{$\tau_{r_{1}}'$}}}
\put(245,80){\vector(-1,0){210}}
\put(70,0){\vector(-1,0){40}}
\put(230,0){\vector(-1,0){40}}
\put(0,65){\vector(0,-1){50}}
\put(280,65){\vector(0,-1){50}}
\put(10,-20){\oval(20,20)[bl]}
\put(10,-30){\line(1,0){260}}
\put(0,-14){\vector(0,1){0}}
\put(270,-20){\oval(20,20)[br]}
\end{picture}}
\end{picture}
\end{tabular}}}
\end{center}
\caption{Naturality Combined}
\label{fig:nat:comb}
\end{figure}
To understand this, we discuss each part (facet) separately.
In short:
%
the front/back 
is due to
Cor.\;\ref{cor:tbl:func} in \S\,\ref{sub:sec:struc:mor:lax}
for (lax) structure morphisms;
%
the left/right hold by 
satisfaction of source/target sound logics;
%
the bottom-right is due to the (abstract) specification morphism;
%
the bottom-left is due to naturality of
$\mathrmbfit{tup}_{\mathcal{A}_{2}}
\xLeftarrow{\;\grave{\tau}_{{\langle{f,g}\rangle}}\;}
{\scriptstyle\sum}_{f}^{\mathrm{op}}
{\;\circ\;}\mathrmbfit{tup}_{\mathcal{A}_{1}}$;
and
%
the top holds for relations.
%
We now discuss each facet in more detail.
\begin{itemize}
%
\item [\textbf{front/back:}]
For each source predicate $r_{2} \in \mathrmbf{R}_{2}$,
the (lax) structure morphism
$\mathcal{M}_{2}
\xrightleftharpoons{{\langle{r,\kappa,\grave{\varphi},f,g}\rangle}}\mathcal{M}_{1}$
defines 
(Cor.\;\ref{cor:tbl:func} of \S\,\ref{sub:sec:struc:mor:lax})
the table morphism
\newline\mbox{}\hfill
\rule[6pt]{0pt}{10pt}
{\footnotesize{$\underset{\mathrmbfit{T}_{2}(r_{2})}
{\langle{\overset{\mathrmbfit{S}}{\sigma}(r_{2}),\mathcal{A}_{2},\overset{\mathrmbfit{K}}{\mathrmbfit{ext}}_{\mathcal{E}_{2}}(r_{2}),\tau_{r_{2}}}\rangle}
\xleftarrow[\xi_{r_{2}}]
{\;{\langle{\grave{\varphi}_{r_{2}},f,g,\kappa_{r_{2}}}\rangle}\;}
\underset{\mathrmbfit{T}_{1}(\mathrmbfit{R}(r_{2}))}
{\langle{\overset{\mathrmbfit{S}}{\sigma}(r_{1}),\mathcal{A}_{1},\overset{\mathrmbfit{K}}{\mathrmbfit{ext}}_{\mathcal{E}_{1}}(r_{1}),\tau_{r_{1}}}\rangle}$}}
\hfill\mbox{}\newline
(same for $r'_{2} \in \mathrmbf{R}_{2}$).
%
\item [\textbf{left/right (sat):}]
Satisfaction for source/target sound logics 
(\S~\ref{sub:sec:spec:sat})
define table morphisms
\newline\mbox{}\hfill
\rule[6pt]{0pt}{10pt}
{\footnotesize{$
{\langle{\sigma(r'_{i}),\mathrmbfit{ext}_{\mathcal{E}}(r'_{i}),\tau_{r'_{i}}}\rangle}
\xleftarrow[\mathrmbfit{T}_{i}(p_{i})]{\;{\langle{h_{p_{i}},k_{p_{i}}}\rangle}\;}
{\langle{\sigma(r_{i}),\mathrmbfit{ext}_{\mathcal{E}}(r_{i}),\tau_{r_{i}}}\rangle}
$,}}
\hfill\mbox{}\newline
which are equivalent to 
\underline{naturality} of the bridges 
$\mathrmbfit{K}_{i}\xRightarrow[\mathrmbfit{T}_{i}{\,\circ\,}\tau_{\mathcal{A}_{i}}]{\,\tau_{i}\;} 
\mathrmbfit{S}_{i}^{\mathrm{op}}\!{\circ\;}\mathrmbfit{tup}_{\mathcal{A}_{i}}$.
\newline
\item [\textbf{bottom right:}]
The structure and specification morphisms have the same underlying schema morphism
$\mathcal{S}_{2} = {\langle{R_{2},\sigma_{2},X_{2}}\rangle}
\xRightarrow{{\langle{r,\grave{\varphi},f}\rangle}}
{\langle{R_{1},\sigma_{1},X_{1}}\rangle} = \mathcal{S}_{1}$.
Apply the tuple function
$\mathrmbfit{tup}_{\mathcal{A}_{1}}$
to the \underline{naturality} of 
the bridge
$\mathrmbfit{S}_{2}{\;\circ\;}{\scriptstyle\sum}_{f}\xRightarrow{\;\grave{\varphi}\;\,}
\mathrmbfit{R}{\;\circ\;}\mathrmbfit{S}_{1}$
for any abstract $\mathcal{S}$-constraint
$r'_{2}\xrightarrow{\,p_{2}\,}r_{2}$
in $\mathrmbf{R}_{2}$: 
for any constraint
$r'_{2}\xrightarrow{p_{2}}r_{2}$ in $\mathrmbf{R}_{2}$,
the following diagram commutes.
%
\begin{center}
{{\begin{tabular}{c}
\setlength{\unitlength}{0.5pt}
\begin{picture}(240,150)(-30,-25)
\put(-5,100){\makebox(0,0){\footnotesize{${\scriptstyle\sum}_{f}(\sigma_{2}(r'_{2}))$}}}
\put(180,100){\makebox(0,0){\footnotesize{$\sigma_{1}(r'_{1})$}}}
\put(-5,0){\makebox(0,0){\footnotesize{${\scriptstyle\sum}_{f}(\sigma_{2}(r_{2}))$}}}
\put(180,0){\makebox(0,0){\footnotesize{$\sigma_{1}(r'_{1})$}}}
\put(90,110){\makebox(0,0){\scriptsize{$\grave{\varphi}_{r'_{2}}$}}}
\put(90,87){\makebox(0,0){\scriptsize{$\grave{\varphi}_{r'_{2}}$}}}
\put(90,10){\makebox(0,0){\scriptsize{$\grave{\varphi}_{r_{2}}$}}}
\put(90,-17){\makebox(0,0){\scriptsize{$\grave{\varphi}_{r_{2}}$}}}
\put(-8,50){\makebox(0,0)[r]{\scriptsize{${\scriptstyle\sum}_{f}(h_{p_{2}})=h_{p_{2}}$}}}
\put(188,50){\makebox(0,0)[l]{\scriptsize{$h_{p_{1}}$}}}
\put(45,100){\vector(1,0){95}}
\put(45,0){\vector(1,0){95}}
\put(0,20){\vector(0,1){60}}
\put(180,20){\vector(0,1){60}}
\put(90,45){\makebox(0,0){\footnotesize{$\grave{\varphi}$ naturality}}}
\end{picture}
\end{tabular}}}
\end{center}
\item [\textbf{bottom left:}]
The tuple bridge
$\mathrmbfit{tup}_{\mathcal{A}_{2}}
\xLeftarrow{\;\grave{\tau}_{{\langle{f,g}\rangle}}\;}
{\scriptstyle\sum}_{f}^{\mathrm{op}}
{\;\circ\;}\mathrmbfit{tup}_{\mathcal{A}_{1}}$
of the type domain morphism 
$\mathcal{A}_{2}
\xrightleftharpoons{{\langle{f,g}\rangle}} 
\mathcal{A}_{1}$
(see footnote \ref{tup:bridge} in \S\;\ref{sub:sec:struc:mor})
satisfies \underline{naturality}: 
for any constraint
$r'_{2}\xrightarrow{p_{2}}r_{2}$ in $\mathrmbf{R}_{2}$,
the following diagram commutes.
%
\begin{center}
{{\begin{tabular}{c}
\setlength{\unitlength}{0.5pt}
\begin{picture}(240,150)(0,-25)
\put(0,100){\makebox(0,0){\footnotesize{$\mathrmbfit{tup}_{\mathcal{A}_{2}}(
\sigma_{2}(r_{2}))$}}}
\put(240,100){\makebox(0,0){\footnotesize{$\mathrmbfit{tup}_{\mathcal{A}_{1}}({\scriptstyle\sum}_{f}(
\sigma_{2}(r_{2})))$}}}
\put(-5,0){\makebox(0,0){\footnotesize{$\mathrmbfit{tup}_{\mathcal{A}_{2}}(
\sigma_{2}(r_{2}))$}}}
\put(245,0){\makebox(0,0){\footnotesize{$\mathrmbfit{tup}_{\mathcal{A}_{1}}({\scriptstyle\sum}_{f}(
\sigma_{2}(r_{2})))$}}}
\put(110,110){\makebox(0,0){\scriptsize{$\grave{\tau}_{{\langle{f,g}\rangle}}(
\sigma_{2}(r_{2}))$}}}
\put(110,90){\makebox(0,0){\scriptsize{${(\mbox{-})}{\,\cdot\,}g$}}}
\put(110,10){\makebox(0,0){\scriptsize{${(\mbox{-})}{\,\cdot\,}g$}}}
\put(110,-14){\makebox(0,0){\scriptsize{$\grave{\tau}_{{\langle{f,g}\rangle}}(
\sigma_{2}(r_{2}))$}}}
\put(-8,50){\makebox(0,0)[r]{\scriptsize{$k_{p_{2}}$}}}
\put(248,50){\makebox(0,0)[l]{\scriptsize{$k_{p_{1}}$}}}
\put(160,100){\vector(-1,0){95}}
\put(160,0){\vector(-1,0){95}}
\put(0,20){\vector(0,1){60}}
\put(240,20){\vector(0,1){60}}
\put(120,45){\makebox(0,0){\footnotesize{$\grave{\tau}_{{\langle{f,g}\rangle}}$ naturality}}}
\end{picture}
\end{tabular}}}
\end{center}
%
\item [\textbf{top:}]
By the other naturality conditions just proven,
for any constraint
$r'_{2}\xrightarrow{p_{2}}r_{2}$ in $\mathrmbf{R}_{2}$
with $\mathrmbfit{R}$-image
$r'_{1}\xrightarrow{p_{1}}r_{1}$ in $\mathrmbf{R}_{1}$,
we know that
$k_{p_{1}}{\,\cdot\,}\kappa_{r'_{2}}{\,\cdot\,}\tau_{r'_{2}}
= \kappa_{r_{2}}{\,\cdot\,}k_{p_{2}}{\,\cdot\,}\tau_{r'_{2}}$.
If the tuple map
$\mathrmbfit{ext}_{\mathcal{E}_{2}}(r'_{2})
\xrightarrow{\tau_{r'_{2}}}
\mathrmbfit{tup}_{\mathcal{A}_{2}}(\sigma_{2}(r'_{2}))$
were injective,
the bridge 
$\overset{\mathrmbfit{K}_{2}}{\mathrmbfit{ext}_{\mathcal{E}_{2}}}
\xLeftarrow{\;\kappa\;}
\overset{\mathrmbfit{R}}{r}{\;\circ\;}
\overset{\mathrmbfit{K}_{1}}{\mathrmbfit{ext}_{\mathcal{E}_{1}}}$
would satisfy ``extent naturality'':
%
\begin{center}
{{\begin{tabular}{c}
\setlength{\unitlength}{0.5pt}
\begin{picture}(240,150)(-25,-25)
\put(0,100){\makebox(0,0){\footnotesize{$\mathrmbfit{ext}_{\mathcal{E}_{2}}(r'_{2})$}}}
\put(180,100){\makebox(0,0){\footnotesize{$\mathrmbfit{ext}_{\mathcal{E}_{1}}(r'_{1})$}}}
\put(-5,0){\makebox(0,0){\footnotesize{$\mathrmbfit{ext}_{\mathcal{E}_{2}}(r_{2})$}}}
\put(185,0){\makebox(0,0){\footnotesize{$\mathrmbfit{ext}_{\mathcal{E}_{1}}(r_{1})$}}}
\put(90,110){\makebox(0,0){\scriptsize{$\kappa_{r'_{2}}$}}}
\put(90,90){\makebox(0,0){\scriptsize{$k$}}}
\put(90,10){\makebox(0,0){\scriptsize{$k$}}}
\put(90,-14){\makebox(0,0){\scriptsize{$\kappa_{r_{2}}$}}}
\put(-8,50){\makebox(0,0)[r]{\scriptsize{$k_{p_{2}}$}}}
\put(188,50){\makebox(0,0)[l]{\scriptsize{$k_{p_{1}}$}}}
\put(135,100){\vector(-1,0){95}}
\put(135,0){\vector(-1,0){95}}
\put(0,20){\vector(0,1){60}}
\put(180,20){\vector(0,1){60}}
\put(90,45){\makebox(0,0){\footnotesize{$\kappa$ naturality}}}
\end{picture}
\end{tabular}}}
\end{center}
%
Here,
we use 
the image part of 
the table-relation reflection
${\langle{\mathrmbfit{im}{\;\dashv\;}\mathrmbfit{inc}}\rangle}
:\mathrmbf{Tbl}{\;\rightleftarrows\;}\mathrmbf{Rel}$.
%
\footnote{The reflection 
${\langle{\mathrmbfit{im}{\;\dashv\;}\mathrmbfit{inc}}\rangle}
:\mathrmbf{Tbl}{\;\rightleftarrows\;}\mathrmbf{Rel}$
(``The {\ttfamily FOLE} Table''\cite{kent:fole:era:tbl})
of the context of relations into the context of tables
embodies the notion of ``informational equivalence''.}
%
Then, 
we use diagonal fill-in Fig.\,\ref{fig:nat:comb}.
Hence, extent naturality holds for the relational interpretation
into $\mathrmbf{Rel}$.
%
\mbox{}\hfill\rule{5pt}{5pt}
\end{itemize}
\end{proof}
%
\comment{
\mbox{}\hfill
{\footnotesize{$
{\scriptstyle{
\overset{\textstyle{\mathrmbfit{T}_{2}(r_{2})}}
{\overbrace{\langle{\sigma_{2}(r_{2}),\mathcal{A}_{2},\mathrmbfit{K}_{2}(r_{2}),\tau_{2,r_{2}}}\rangle}}
\;\xleftarrow[{\langle{\grave{\varphi}_{r_{2}},f,g,\kappa_{r_{2}}}\rangle}]{\xi_{r_{2}}}\;
\overset{\textstyle{\mathrmbfit{T}_{1}(r(r_{2}))}}
{\overbrace{\langle{\sigma_{1}(r(r_{2})),\mathcal{A}_{1},\mathrmbfit{K}_{1}(r(r_{2})),\tau_{1,r(r_{2})}}\rangle}}
.}}
$}}
\hfill\mbox{}\newline
%
}
\newpage
\begin{proposition}\label{prop:snd:log:mor:2:db:mor}
A (lax) sound logic morphism
\newline\mbox{}\hfill
\rule[5pt]{0pt}{10pt}
{\footnotesize{
$\mathcal{L}_{2}={\langle{\mathcal{S}_{2},\mathcal{M}_{2},\mathcal{R}_{2}}\rangle}
\xrightleftharpoons{{\langle{\mathrmbfit{R},\kappa,\grave{\varphi},f,g}\rangle}}
{\langle{\mathcal{S}_{1},\mathcal{M}_{1},\mathcal{R}_{1}}\rangle}=\mathcal{L}_{1}$
\comment{See Tbl.\,\ref{tbl:snd:log:mor}.}
}}\hfill\mbox{}\newline
defines 
a database morphism
%
\[\mbox{\footnotesize$
\mathcal{R}_{2}=
{\langle{\mathrmbf{R}_{2},\mathrmbfit{S}_{2},\mathcal{A}_{2},\mathrmbfit{K}_{2},\tau_{2}}\rangle}
\xrightarrow[{\langle{\mathrmbfit{R},\xi}\rangle}]
{{\langle{\mathrmbfit{R},\kappa,\grave{\varphi},f,g}\rangle}}
{\langle{\mathrmbf{R}_{1},\mathrmbfit{S}_{1},\mathcal{A}_{1},\mathrmbfit{K}_{1},\tau_{1}}\rangle}=\mathcal{R}_{1}$\normalsize}
\footnote{See 
Def.\,\ref{def:db:mor:proj}
and
Tbl.\,\ref{tbl:db:mor} 
in \S\;\ref{sub:sec:db:mor}.}
\]
whose tabular interpretation bridge
{\footnotesize{${\mathrmbfit{T}_{2}
{\,\xLeftarrow{\;\,\xi\,}\,}
\mathrmbfit{R}{\,\circ\,}\mathrmbfit{T}_{1}}$}}
factors
\newline\mbox{}\hfill
\rule[-5pt]{0pt}{10pt}
$\xi = (\grave{\psi} \circ \mathrmbfit{inc}_{\mathcal{A}_{1}}) 
\bullet (\mathrmbfit{T}_{2} \circ \grave{\chi}_{{\langle{f,g}\rangle}})$
\hfill\mbox{}\newline
%
%
through the fiber adjunction 
$\mathrmbf{Tbl}(\mathcal{A}_{2})
\xleftarrow{\acute{\mathrmbfit{tbl}}_{{\langle{f,g}\rangle}}\;\dashv\;\grave{\mathrmbfit{tbl}}_{{\langle{f,g}\rangle}}}
\mathrmbf{Tbl}(\mathcal{A}_{1})$.

%
\end{proposition}
\begin{proof}
By Prop.~\ref{prop:snd:log:2:db}
the sound logics 
$\mathcal{L}_{2}$ 
and
$\mathcal{L}_{1}$
define two databases
$\mathcal{R}_{2}$
and
$\mathcal{R}_{1}$.
By Prop.\,\ref{prop:sat:preserve}
the (lax) sound logic morphism
{{
$\mathcal{L}_{2}
\xrightleftharpoons{{\langle{\mathrmbfit{R},\kappa,\grave{\varphi},f,g}\rangle}}
\mathcal{L}_{1}$
}}
defines a
database morphism
$ 
{\langle{\mathrmbf{R}_{2},\mathrmbfit{T}_{2}}\rangle} 
\xleftarrow{{\langle{\mathrmbfit{R},\xi}\rangle}} 
{\langle{\mathrmbf{R}_{1},\mathrmbfit{T}_{1}}\rangle}$,
whose tabular interpretation bridge
{\footnotesize{${\mathrmbfit{T}_{2}
{\,\xLeftarrow{\;\,\xi\,}\,}
\mathrmbfit{R}{\,\circ\,}\mathrmbfit{T}_{1}}$}}
factors as above.
%
\mbox{}\hfill\rule{5pt}{5pt}
\end{proof}
%


%
\begin{theorem}\label{thm:snd:log:2:db}
There is a passage
$\mathring{\mathrmbf{Snd}}^{\mathrm{op}}\!\xrightarrow{\;\mathring{\mathrmbfit{db}}\;}\mathrmbf{Db}$.
%
\end{theorem}
\begin{proof}
A sound logic 
$\mathcal{L}={\langle{\mathcal{S},\mathcal{M},\mathcal{T}}\rangle}$
is mapped to its associated database 
$\mathcal{R}={\langle{\mathrmbf{R},\mathrmbfit{K},\mathrmbfit{S},\tau,\mathcal{A}}\rangle}$
by Prop.~\ref{prop:snd:log:2:db}.
A sound logic morphism 
$\mathcal{L}_{2}
\xrightleftharpoons{{\langle{\mathrmbfit{R},k,\grave{\varphi},f,g}\rangle}}
\mathcal{L}_{1}$
is mapped to its associated database morphism
$\mathcal{R}_{2}
\xrightarrow{{\langle{\mathrmbfit{R},\kappa,\grave{\varphi},f,g}\rangle}}
\mathcal{R}_{1}$
by Prop.~\ref{prop:snd:log:mor:2:db:mor}.
\mbox{}\hfill\rule{5pt}{5pt}
\end{proof}
%
\comment{
The table-relation reflection
${\langle{\mathrmbfit{im}{\;\dashv\;}\mathrmbfit{inc}}\rangle}
:\mathrmbf{Tbl}{\;\rightleftarrows\;}\mathrmbf{Rel}$
embodies the notion of informational equivalence.}

%
\begin{table}
\begin{flushleft}
{\fbox{\footnotesize{\begin{minipage}{360pt}
\begin{description}
\item[sound logic morphism] 
\mbox{}
\comment{$\equiv$ database morphism
\hfill
{\scriptsize{$\mathcal{R}_{2} = 
{\langle{\mathrmbf{R}_{2},\mathcal{S}_{2},\mathcal{A}_{2},
\mathrmbfit{K}_{2},\tau_{2}}\rangle}
\xleftarrow{{\langle{\mathrmbfit{R},\grave{\varphi},f,g,\kappa}\rangle}}
{\langle{\mathrmbf{R}_{1},\mathcal{S}_{1},\mathcal{A}_{1},\mathrmbfit{K}_{1},\tau_{1}}\rangle}
 = \mathcal{R}_{1}$}\normalsize}}
\hfill
{\footnotesize{$\mathcal{L}_{2}={\langle{\mathcal{S}_{2},\mathcal{M}_{2},\mathcal{T}_{2}}\rangle}
\xrightleftharpoons{{\langle{\mathrmbfit{R},\kappa,\grave{\varphi},f,g}\rangle}}
{\langle{\mathcal{S}_{1},\mathcal{M}_{1},\mathcal{T}_{1}}\rangle}=\mathcal{L}_{1}$}}
\begin{itemize}
%
\item schema morphism
\hfill
{\scriptsize{$\mathcal{S}_{2}={\langle{R_{2},\sigma_{2},X_{2}}\rangle}
\xRightarrow{\;{\langle{r,\grave{\varphi},f}\rangle}\;}
{\langle{R_{1},\sigma_{1},X_{1}}\rangle}=\mathcal{S}_{1}$}}
\item (lax) struc morph
\comment{$\equiv$ constr-free database morphism
\hfill
{\scriptsize{${\langle{R_{2},\sigma_{2},\mathcal{A}_{2},
\mathrmbfit{K}_{2},\tau_{2}}\rangle}
\xleftarrow{{\langle{r,\grave{\varphi},f,g,\kappa}\rangle}}
{\langle{R_{1},\sigma_{1},\mathcal{A}_{1},\mathrmbfit{K}_{1},\tau_{1}}\rangle}
$}\normalsize}}
\hfill
{\scriptsize{$\mathcal{M}_{2} 
= {\langle{\overset{\textstyle{\mathcal{E}_{2}}}{\langle{R_{2},\mathrmbfit{K}_{2}}\rangle},\sigma_{2},\tau_{2},\mathcal{A}_{2}}\rangle}
\xrightleftharpoons{{\langle{r,\kappa,\grave{\varphi},f,g}\rangle}}
{\langle{\overset{\textstyle{\mathcal{E}_{1}}}{\langle{R_{1},\mathrmbfit{K}_{1}}\rangle},\sigma_{1},\tau_{1},\mathcal{A}_{1}}\rangle} = 
\mathcal{M}_{1}$}}
\begin{itemize}
%
%
\item type domain morphism
\hfill
{\scriptsize{$\mathcal{A}_{2} = {\langle{X_{2},Y_{2},\models_{\mathcal{A}_{2}}}\rangle} 
\xrightleftharpoons{{\langle{f,g}\rangle}} 
{\langle{X_{1},Y_{1},\models_{\mathcal{A}_{1}}}\rangle} = \mathcal{A}_{1}$}}
%
%
\item (lax) entity infomorphism
\hfill
{\scriptsize{$\mathcal{K}_{2}={\langle{\mathrmbf{R}_{2},\mathrmbfit{K}_{2}}\rangle} 
\xrightarrow{\;{\langle{\mathrmbfit{R},\kappa}\rangle}\;}
{\langle{\mathrmbf{R}_{1},\mathrmbfit{K}_{1}}\rangle}=\mathcal{K}_{1}$}}
%
%
\end{itemize}
\item[]
which satisfy the constraint 
\newline
{{\mbox{}\hspace{20pt} 
$\kappa\bullet\tau_{2}
= (\mathrmbfit{R}^{\mathrm{op}}\!{\circ\;}\tau_{1})
{\;\bullet\;}(\grave{\varphi}^{\mathrm{op}}\!{\circ\;}\mathrmbfit{tup}_{\mathcal{A}_{1}})
{\;\bullet\;}(\mathrmbfit{S}_{2}^{\mathrm{op}}{\;\circ\;}\grave{\tau}_{{\langle{f,g}\rangle}})$}}
%
%
%
%
\item specification morphism
\hfill
{\scriptsize{$\mathcal{T}_{2}={\langle{\mathrmbf{R}_{2},\mathrmbfit{S}_{2},X_{2}}\rangle}
\xrightarrow{{\langle{\mathrmbfit{R},\grave{\varphi},f}\rangle}}
{\langle{\mathrmbf{R}_{1},\mathrmbfit{S}_{1},X_{1}}\rangle}=\mathcal{T}_{1}$}}
\end{itemize}
\end{description}
\end{minipage}}}}
\end{flushleft}
\caption{Sound Logic Morphism}
\label{tbl:snd:log:mor}
\end{table}
\newpage
\section{{\ttfamily FOLE} Databases.}\label{sec:db}


%
A {\ttfamily FOLE}
relational database
$\mathcal{R} = {\langle{\mathrmbf{R},\mathrmbfit{T}}\rangle}$
consists of
a context $\mathrmbf{R}$ of predicates linked by constraints and  
an interpretation diagram 
$\mathrmbfit{T} : \mathrmbf{R}^{\mathrm{op}} 
\rightarrow \mathrmbf{Tbl}$
of tables.
%
A relational database morphism 
${\langle{\mathrmbfit{R},\xi}\rangle} : 
{\langle{\mathrmbf{R}_{2},\mathrmbfit{T}_{2}}\rangle} \leftarrow 
{\langle{\mathrmbf{R}_{1},\mathrmbfit{T}_{1}}\rangle}$
(Tbl.\,\ref{tbl:fole:morph}
in \S\,\ref{sub:sec:adj:components})
is 
a diagram morphism
(LHS Fig.\,\ref{fig:db:mor}
in \S\,\ref{sub:sec:db:mor})
consisting of a shape-changing passage 
$\mathrmbf{R}_{2} \xrightarrow{\:\mathrmbfit{R}\:} \mathrmbf{R}_{1}$ 
and
a bridge
$\xi : \mathrmbfit{T}_{2} \Leftarrow \mathrmbfit{R}^{\mathrm{op}} \circ \mathrmbfit{T}_{1}$. 
%
Composition is component-wise.
The mathematical context of 
{\ttfamily FOLE} relational databases 
is denoted by 
$\mathrmbf{DB}
= \mathrmbf{Tbl}^{\scriptscriptstyle{\Downarrow}}
= \bigl({(\mbox{-})}^{\mathrm{op}}{\Downarrow\,}\mathrmbf{Tbl}\bigr)$.
%
\footnote{See \S\,4.2.1 in the paper ``The \texttt{FOLE} Database'' \cite{kent:fole:era:db}.}

\subsection{Databases.}\label{sub:sec:db:obj}

%
\begin{definition}\label{def:db}
A {\ttfamily FOLE} database
$\mathcal{R} = {\langle{\mathrmbf{R},\mathrmbfit{T},\mathcal{A}}\rangle}$,
with constant type domain $\mathcal{A}$,
is a database 
with an interpretation diagram 
$\mathrmbfit{T} : \mathrmbf{R}^{\mathrm{op}} 
\rightarrow \mathrmbf{Tbl}(\mathcal{A})
\hookrightarrow \mathrmbf{Tbl}$
that factors through the context of $\mathcal{A}$-tables.
\end{definition}
\begin{figure}
\begin{center}
{{\begin{tabular}{c}
\setlength{\unitlength}{0.5pt}
\begin{picture}(120,180)(-60,0)
\put(10,180){\makebox(0,0){\footnotesize{$\mathrmbf{R}^{\mathrm{op}}$}}}
\put(0,118){\makebox(0,0){\footnotesize{$\mathrmbf{Tbl}(\mathcal{A})$}}}
\put(75,60){\makebox(0,0){\footnotesize{$\mathrmbf{List}(X)^{\mathrm{op}}$}}}
\put(0,0){\makebox(0,0){\footnotesize{$\mathrmbf{Set}$}}}
\put(6,148){\makebox(0,0)[l]{\scriptsize{$\mathrmbfit{T}$}}}
\put(-75,95){\makebox(0,0)[r]{\scriptsize{$\mathrmbfit{K}$}}}
\put(-37,72){\makebox(0,0)[r]{\scriptsize{$\mathrmbfit{key}_{\mathcal{A}}$}}}
\put(36,93){\makebox(0,0)[l]{\scriptsize{$\mathrmbfit{sign}_{\mathcal{A}}^{\mathrm{op}}$}}}
\put(65,130){\makebox(0,0)[l]{\scriptsize{$\mathrmbfit{S}^{\mathrm{op}}$}}}
\put(36,26){\makebox(0,0)[l]{\scriptsize{$\mathrmbfit{tup}_{\mathcal{A}}$}}}
\put(0,60){\makebox(0,0){\shortstack{\scriptsize{$\;\tau_{\mathcal{A}}$}\\\large{$\Longrightarrow$}}}}
\put(0,165){\vector(0,-1){34}}
\put(15,105){\vector(1,-1){30}}
\put(45,45){\vector(-1,-1){30}}
\qbezier(-18,167)(-120,90)(-20,13)\put(-20,13){\vector(1,-1){0}}
\qbezier(-12,105)(-60,60)(-12,15)\put(-12,15){\vector(1,-1){0}}
\qbezier(18,167)(70,140)(66,76)\put(66,76){\vector(0,-1){0}}
\end{picture}
\end{tabular}}}
\end{center}
\caption{Database: Type Domain}
\label{fig:db}
\end{figure}
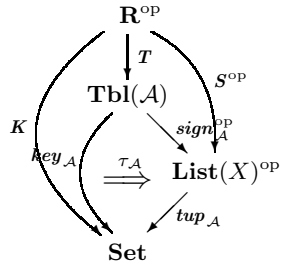
%


%
\begin{definition}\label{def:db:proj}
Using $\mathcal{A}$-table projection passages
\footnote{See \S\,3.4.1 of the paper
``The {\ttfamily FOLE} Table''
\cite{kent:fole:era:tbl}.}
(Fig.\,\ref{fig:db}),
a database 
$\mathcal{R} =
{\langle{\mathrmbf{R},\mathrmbfit{K},\mathrmbfit{S},\tau,\mathcal{A}}\rangle}$,
with constant type domain $\mathcal{A}$,
consists of
\begin{itemize}
\item 
a context $\mathrmbf{R}$ of predicates,
%
\item 
%
a relational database schema
$\mathcal{S} = {\langle{\mathrmbf{R},\mathrmbfit{S},X}\rangle}$
with
signature diagram
$\mathrmbfit{S} = \mathrmbfit{T}^{\mathrm{op}}{\circ\;}\mathrmbfit{sign}_{\mathcal{A}} :
\mathrmbf{R}\rightarrow\mathrmbf{List}(X)$,
\item 
%
a key diagram
$\mathrmbfit{K} = 
\mathrmbfit{T}{\;\circ\;}\mathrmbfit{key}_{\mathcal{A}} :
\mathrmbf{R}^{\mathrm{op}}\rightarrow\mathrmbf{Set}$,
%
and
\item 
a tuple bridge
$\tau = \mathrmbfit{T}{\,\circ\,}\tau_{\mathcal{A}} :
\mathrmbfit{K}\Rightarrow\mathrmbfit{S}^{\mathrm{op}}\!{\,\circ\,}\mathrmbfit{tup}_{\mathcal{A}}$.
%
\end{itemize}
\end{definition}

%

%
\begin{table}
\begin{center}
{\footnotesize{\setlength{\extrarowheight}{2pt}
\begin{tabular}
{|@{\hspace{5pt}}r@{\hspace{20pt}}l@{\hspace{10pt}$
$\hspace{0pt}}l@{\hspace{0pt}}|}
\multicolumn{3}{l}{
$\mathcal{R} 
= {\langle{\mathrmbf{R},\mathrmbfit{T},\mathcal{A}}\rangle}$
\hfill
\textsf{relational database}
}
\\ \hline
predicate context
& $\mathrmbf{R}$
&
\\
type domain (attribute classification) 
& $\mathcal{A} = {\langle{X,Y,\models_{\mathcal{A}}}\rangle}$
& 
\\
table passage
&
$\mathrmbf{R}^\text{op}\xrightarrow{\;\mathrmbfit{T}\;}\mathrmbf{Tbl}(\mathcal{A})$
&
\\
\hline
\multicolumn{3}{l}{
$\mathcal{R} 
= {\langle{\mathrmbf{R},\mathrmbfit{S},\mathrmbfit{K},\tau,\mathcal{A}}\rangle}$}
\\ \hline
database schema (specification) 
& 
$\mathcal{T} = {\langle{\mathrmbf{R},\mathrmbfit{S},X}\rangle}$
& 
\\
signature passage 
&
$\mathrmbf{R}
\xrightarrow
[\mathrmbfit{T}^{\mathrm{op}}{\circ\;}\mathrmbfit{sign}_{\mathcal{A}}]
{\mathrmbfit{S}}
\mathrmbf{List}(X)$
&
\\ \cline{1-1}
key passage 
& $\mathrmbf{R}^{\mathrm{op}}
\xrightarrow[\mathrmbfit{T}{\;\circ\;}\mathrmbfit{key}_{\mathcal{A}}]{\mathrmbfit{K}}
\mathrmbf{Set}$
&
\\ \cline{1-1}
tuple bridge 
& 
$\mathrmbfit{K}
\xRightarrow[\mathrmbfit{T}{\,\circ\,}\tau_{\mathcal{A}}]{\tau}
\mathrmbfit{S}^{\mathrm{op}}\!{\,\circ\,}\mathrmbfit{tup}_{\mathcal{A}}$
& 
\\
\hline
\end{tabular}}}
\end{center}
\caption{Relational Database}
\label{tbl:db}
\end{table}

\newpage

\begin{proposition}
\label{prop:db:to:struc}
The constraint-free aspect of a {\ttfamily FOLE} database 
$\mathcal{R} =
{\langle{\mathrmbf{R},\mathrmbfit{K},\mathrmbfit{S},\tau,\mathcal{A}}\rangle}$,
with constant type domain $\mathcal{A}$,
is the same as
%
%
a (lax) {\ttfamily FOLE} structure
$\mathcal{M} = {\langle{\mathcal{E},\sigma,\tau,\mathcal{A}}\rangle}$
in $\mathring{\mathrmbf{Struc}}$.
%
%
\end{proposition}
\begin{proof}
We define the various components
of the (lax) structure 
$\mathcal{M} = {\langle{\mathcal{E},\sigma,\tau,\mathcal{A}}\rangle}$
listed in 
Def.\,\ref{def:struc:lax} of \S\,\ref{sub:sec:struc:lax}
and pictured in 
Fig.\,\ref{fig:fole:struc} of \S\;\ref{sub:sec:struc}.
\begin{itemize}
\item[$\mathcal{A}$\textbf{:}]
The attribute classification (typed domain)
$\mathcal{A} = {\langle{X,Y,\models_{\mathcal{A}}}\rangle}$ is given.
%
\item[$\sigma$\textbf{:}]
%
The schema (type hypergraph)
$\mathcal{S}={\langle{R,\sigma,X}\rangle}$
consists of 
the set $R$ of relation symbols (predicates)
and
the signature map 
$\sigma : R\rightarrow\mathrmbf{List}(X) : r \mapsto \sigma(r)=\mathrmbfit{S}(r)$.
This is the constraint-free aspect of the database schema
$\mathcal{S} = {\langle{\mathrmbf{R},\mathrmbfit{S},X}\rangle}$
with signature diagram
$\mathrmbfit{S} 
: \mathrmbf{R}\rightarrow\mathrmbf{List}(X)$.
%
\item[$\mathcal{E}$\textbf{:}] 
%
The (lax) entity classification
$\mathcal{E}={\langle{R,\mathrmbfit{K}}\rangle}$
consists
of 
the set $R$ of relation symbols (predicates)
and
the key function $R \xrightarrow{\mathrmbfit{K}}\mathrmbf{Set}$.
This is the constraint-free aspect of the key diagram
$\mathrmbf{R}^{\mathrm{op}}\xrightarrow{\mathrmbfit{K}}\mathrmbf{Set}$.
\item[$\tau$\textbf{:}]
The tuple bridge
$\tau :
\mathrmbfit{K}\Rightarrow\sigma{\;\circ\;}\mathrmbfit{tup}_{\mathcal{A}}$
is the constraint-free aspect of the tuple bridge 
$\tau :
\mathrmbfit{K}
\xRightarrow{\,\tau}
\mathrmbfit{S}^{\mathrm{op}}\!{\,\circ\,}\mathrmbfit{tup}_{\mathcal{A}}$.
%
\hfill\rule{5pt}{5pt}
\end{itemize}
\end{proof}
\begin{proposition}
\label{prop:db:2:snd:log}
Any {\ttfamily FOLE} database 
$\mathcal{R} 
= {\langle{\mathrmbf{R},\mathrmbfit{T},\mathcal{A}}\rangle}
= {\langle{\mathrmbf{R},\mathrmbfit{K},\mathrmbfit{S},\tau,\mathcal{A}}\rangle}$
in $\mathrmbf{Db}(\mathcal{A})$
defines  
a (lax)
{\ttfamily FOLE} sound logic
$\mathcal{L}={\langle{\mathcal{S},\mathcal{M},\mathcal{T}}\rangle}$
in $\mathring{\mathrmbf{Snd}}$.
\end{proposition}
\begin{proof}
The schema 
$\mathcal{S}={\langle{R,\sigma,X}\rangle}$ 
is the structure schema (mentioned above), 
the constraint-free aspect of the database schema
${\langle{\mathrmbf{R},\mathrmbfit{S},X}\rangle}$.
The  $\mathcal{S}$-structure
$\mathcal{M} = {\langle{\mathcal{E},\sigma,\tau,\mathcal{A}}\rangle}$
is defined in Prop.~\ref{prop:db:to:struc} above.
The (abstract) 
$\mathcal{S}$-specification
$\mathcal{T} = {\langle{\mathrmbf{R},\mathrmbfit{S},X}\rangle}$ 
is the same as
the database schema,
as discussed in \S\,\ref{sub:sec:spec}.
By Prop\;\ref{sat:tbl:interp} in \S\,\ref{sub:sec:spec:sat},
the 
tabular interpretation passage
$\mathrmbf{R}^{\mathrm{op}}\!\xrightarrow{\mathrmbfit{T}\;}\mathrmbf{Tbl}(\mathcal{A})$
demonstrates satisfaction
$\mathcal{M}{\;\models_{\mathcal{S}}\;}\mathcal{T}$.
\hfill\rule{5pt}{5pt}
\end{proof}
%

%
\newpage
\subsection{Database Morphisms.}\label{sub:sec:db:mor}

A \texttt{FOLE} database morphism, 
with constant type domain morphism
${\langle{f,g}\rangle} : \mathcal{A}_{2} \rightleftarrows \mathcal{A}_{1}$, 
is a 
\texttt{FOLE} database morphism
${\langle{\mathrmbfit{R},\xi}\rangle} : 
{\langle{\mathrmbf{R}_{2},\mathrmbfit{T}_{2}}\rangle} \leftarrow 
{\langle{\mathrmbf{R}_{1},\mathrmbfit{T}_{1}}\rangle}$,
whose tabular interpretation bridge
{\footnotesize{${\mathrmbfit{T}_{2}
{\,\xLeftarrow{\;\,\xi\,}\,}
\mathrmbfit{R}{\,\circ\,}\mathrmbfit{T}_{1}}$}}
factors
(Fig.\,\ref{fig:db:mor}) 
%
\begin{equation}\label{eqn:db:mor:def}
\xi = (\grave{\psi} \circ \mathrmbfit{inc}_{\mathcal{A}_{1}}) 
\bullet (\mathrmbfit{T}_{2} \circ \grave{\chi}_{{\langle{f,g}\rangle}})
\end{equation}
through the fiber adjunction 
$\mathrmbf{Tbl}(\mathcal{A}_{2})
\xleftarrow{\acute{\mathrmbfit{tbl}}_{{\langle{f,g}\rangle}}\;\dashv\;\grave{\mathrmbfit{tbl}}_{{\langle{f,g}\rangle}}}
\mathrmbf{Tbl}(\mathcal{A}_{1})$
%
\footnote{
\label{tbl:fbr:adj}
There is a table fiber adjunction
$\mathrmbf{Tbl}(\mathcal{A}_{2})
\xleftarrow{{\langle{
\acute{\mathrmbfit{tbl}}_{{\langle{f,g}\rangle}}{\!\dashv\,}\grave{\mathrmbfit{tbl}}_{{\langle{f,g}\rangle}}
}\rangle}}
\mathrmbf{Tbl}(\mathcal{A}_{1})$
representing tabular flow along a type domain morphism
$\mathcal{A}_{2}\xrightleftharpoons{{\langle{f,g}\rangle}}\mathcal{A}_{1}$.
See \S\,3.4.2 of Kent \cite{kent:fole:era:tbl}.
The fiber passage
$\mathrmbf{Tbl}(\mathcal{A}_{2})
\xleftarrow{\acute{\mathrmbfit{tbl}}_{{\langle{f,g}\rangle}}}
\mathrmbf{Tbl}(\mathcal{A}_{1})$
is define in terms of 
the tuple bridge
{\footnotesize{${f^{\ast}}^{\mathrm{op}}{\circ\;}\mathrmbfit{tup}_{\mathcal{A}_{2}}
\xLeftarrow{\;\acute{\tau}_{{\langle{f,g}\rangle}}\;}
\mathrmbfit{tup}_{\mathcal{A}_{1}}$}}
and the substitution function
$\mathrmbf{List}(X_{2})\xleftarrow{f^{\ast}}\mathrmbf{List}(X_{1})$.
The adjoint fiber passage
$\mathrmbf{Tbl}(\mathcal{A}_{2})
\xrightarrow{\grave{\mathrmbfit{tbl}}_{{\langle{f,g}\rangle}}}
\mathrmbf{Tbl}(\mathcal{A}_{1})$
is define in terms the adjoints,
the tuple bridge
{\footnotesize{$\mathrmbfit{tup}_{\mathcal{A}_{2}}
\xLeftarrow{\;\grave{\tau}_{{\langle{f,g}\rangle}}\;}
{\scriptstyle\sum}_{f}^{\mathrm{op}}{\circ\;}\mathrmbfit{tup}_{\mathcal{A}_{1}}$}}
and the existential quantifier function
$\mathrmbf{List}(X_{2})\xrightarrow{{\scriptscriptstyle\sum}_{f}}\mathrmbf{List}(X_{1})$.}
\comment{Fibered by signature over the adjunction
$\mathrmbf{List}(X_{2})
\xrightarrow{{\langle{{{\Sigma}_{f}}\;\dashv\;{f^{\ast}}}\rangle}}
\mathrmbf{List}(X_{2})$.}
in terms of 
\begin{itemize}
\item 
some bridge
$\grave{\psi} : 
\mathrmbfit{T}_{2}\circ\grave{\mathrmbfit{tbl}}_{{\langle{f,g}\rangle}}
\Leftarrow 
\mathrmbfit{R}^{\mathrm{op}}\circ\mathrmbfit{T}_{1}$
and 
\item 
the inclusion bridge
$
\grave{\chi}_{{\langle{f,g}\rangle}} : 
\mathrmbfit{inc}_{\mathcal{A}_{2}}
\Leftarrow
\grave{\mathrmbfit{tbl}}_{{\langle{f,g}\rangle}}\circ\mathrmbfit{inc}_{\mathcal{A}_{1}}$
%
\footnote{Equivalently, 
in terms of their \textbf{levo} bridge adjoints
in Tbl.\,\ref{tbl:adjoints} of \S\,\ref{sub:sec:adj:components}.}
.
\end{itemize}
%
\comment{
For any 
source predicate symbol $r_{2} \in R_{2}$
with 
target predicate symbol $r_{1} = \mathrmbfit{R}(r_{2}) \in R_{1}$, 
the $r_{2}^{\text{th}}$ component of 
the tabular interpretation bridge
{\footnotesize{${\mathrmbfit{T}_{2}
{\,\xLeftarrow{\;\,\xi\,}\,}
\mathrmbfit{R}{\,\circ\,}\mathrmbfit{T}_{1}}$}}
is the table morphism
{\footnotesize{${\mathrmbfit{T}_{2}(r_{2})
{\,\xleftarrow[{\langle{\grave{\varphi}_{r_{2}},f,g,\kappa_{r_{2}}}\rangle}]
{\;\,\xi_{r_{2}}\,}\,}\mathrmbfit{T}_{1}(r_{1})}$}}
which factors as
\begin{equation}\label{eqn:db:mor:factor}
\mathrmbfit{T}_{2}(r_{2})
\xleftarrow{
{\text{$\grave{\chi}_{{\langle{f,g}\rangle}}$}}_{r_{2}}}
{\grave{\mathrmbfit{tbl}}_{{\langle{f,g}\rangle}}}(\mathrmbfit{T}_{2}(r_{2}))
\xleftarrow{\grave{\psi}_{r_{2}}} 
\mathrmbfit{T}_{1}(r_{1}).
\end{equation}
%
%
}
We normally just use the bridge restriction $\grave{\psi}$ for the database morphism.
The original definition can be computed with the factorization in 
Disp.\,\ref{eqn:db:mor:def}.

\begin{figure}
\begin{center}
{{\begin{tabular}{@{\hspace{5pt}}c@{\hspace{15pt}}c@{\hspace{5pt}}c@{\hspace{5pt}}}
{{\begin{tabular}[b]{c}
\setlength{\unitlength}{0.58pt}
\begin{picture}(80,160)(5,12)
\put(5,160){\makebox(0,0){\footnotesize{$\mathrmbf{R}_{2}^{\mathrm{op}}$}}}
\put(85,160){\makebox(0,0){\footnotesize{$\mathrmbf{R}_{1}^{\mathrm{op}}$}}}
\put(40,15){\makebox(0,0){\normalsize{$\mathrmbf{Tbl}$}}}
\put(45,172){\makebox(0,0){\scriptsize{$\mathrmbfit{R}^{\mathrm{op}}$}}}
\put(0,100){\makebox(0,0)[r]{\footnotesize{$\mathrmbfit{T}_{2}$}}}
\put(81,100){\makebox(0,0)[l]{\footnotesize{$\mathrmbfit{T}_{1}$}}}
\put(40,105){\makebox(0,0){\shortstack{\normalsize{$\xLeftarrow{\;\;\;\xi\;\;}$}}}}
\put(15,160){\vector(1,0){50}}
\qbezier(,150)(0,80)(25,30)\put(29,24){\vector(2,-3){0}}
\qbezier(80,150)(80,80)(55,30)\put(51,24){\vector(-2,-3){0}}
\end{picture}
\end{tabular}}}
&
{{\begin{tabular}[b]{c}
\setlength{\unitlength}{0.58pt}
\begin{picture}(80,160)(0,0)
\put(20,90){\makebox(0,0){\normalsize{$=$}}}
\end{picture}
\end{tabular}}}
&
{{\begin{tabular}[b]{c}
\setlength{\unitlength}{0.58pt}
\begin{picture}(120,160)(0,10)
\put(5,160){\makebox(0,0){\footnotesize{$\mathrmbf{R}_{2}^{\mathrm{op}}$}}}
\put(125,160){\makebox(0,0){\footnotesize{$\mathrmbf{R}_{1}^{\mathrm{op}}$}}}
\put(0,80){\makebox(0,0){\footnotesize{$\mathrmbf{Tbl}(\mathcal{A}_{2})$}}}
\put(120,80){\makebox(0,0){\footnotesize{$\mathrmbf{Tbl}(\mathcal{A}_{1})$}}}
\put(60,5){\makebox(0,0){\normalsize{$\mathrmbf{Tbl}$}}}
\put(65,172){\makebox(0,0){\scriptsize{$\mathrmbfit{R}^{\mathrm{op}}$}}}
\put(-5,125){\makebox(0,0)[r]{\scriptsize{$\mathrmbfit{T}_{2}$}}}
\put(125,125){\makebox(0,0)[l]{\scriptsize{$\mathrmbfit{T}_{1}$}}}
\put(60,92){\makebox(0,0){\scriptsize{$\grave{\mathrmbfit{tbl}}_{{\langle{f,g}\rangle}}$}}}
\put(24,38){\makebox(0,0)[r]{\scriptsize{$\mathrmbfit{inc}_{\mathcal{A}_{2}}$}}}
\put(97,38){\makebox(0,0)[l]{\scriptsize{$\mathrmbfit{inc}_{\mathcal{A}_{1}}$}}}
\put(60,130){\makebox(0,0){\shortstack{
\scriptsize{$\grave{\psi}$}\\\large{$\Longleftarrow$}}}}
\put(60,54){\makebox(0,0){\shortstack{\footnotesize{$\xLeftarrow{\grave{\chi}_{{\langle{f,g}\rangle}}}$}}}}
\put(20,160){\vector(1,0){80}}
\put(35,80){\vector(1,0){50}}
\put(0,145){\vector(0,-1){50}}
\put(120,145){\vector(0,-1){50}}
\put(9,68){\vector(3,-4){38}}
\put(111,68){\vector(-3,-4){38}}
\end{picture}
\end{tabular}}}
\end{tabular}}}
\end{center}
\caption{Database Morphism}
\label{fig:db:mor}
\end{figure}
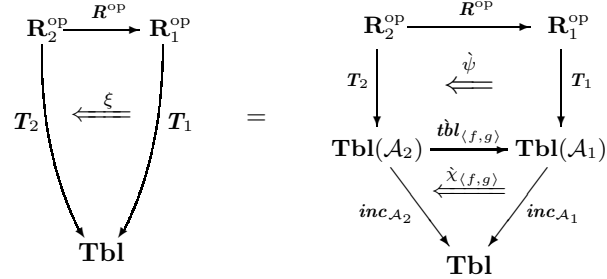
\begin{figure}
\begin{center}
\begin{tabular}{c}
\begin{tabular}{c@{\hspace{50pt}}c}
\\
\textbf{levo} & \textbf{dextro}
\\&\\&\\
{{\begin{tabular}[b]{c}
\setlength{\unitlength}{0.55pt}
\begin{picture}(120,80)(0,0)
\put(0,80){\makebox(0,0){\footnotesize{$\mathrmbf{Tbl}(\mathcal{A}_{2})$}}}
\put(125,80){\makebox(0,0){\footnotesize{$\mathrmbf{Tbl}(\mathcal{A}_{1})$}}}
\put(60,5){\makebox(0,0){\footnotesize{$\mathrmbf{Tbl}$}}}
\put(60,92){\makebox(0,0){\scriptsize{$\acute{\mathrmbfit{tbl}}_{{\langle{f,g}\rangle}}$}}}
\put(24,38){\makebox(0,0)[r]{\scriptsize{$\mathrmbfit{inc}_{\mathcal{A}_{2}}$}}}
\put(97,38){\makebox(0,0)[l]{\scriptsize{$\mathrmbfit{inc}_{\mathcal{A}_{1}}$}}}
\put(60,54){\makebox(0,0){\shortstack{\scriptsize{$\acute{\chi}_{{\langle{f,g}\rangle}}$}\\\large{$\Longleftarrow$}}}}
\put(85,80){\vector(-1,0){50}}
\put(10,67){\vector(3,-4){38}}
\put(111,68){\vector(-3,-4){38}}
\end{picture}
\end{tabular}}}
&
{{\begin{tabular}[b]{c}
\setlength{\unitlength}{0.55pt}
\begin{picture}(120,80)(0,0)
\put(0,80){\makebox(0,0){\footnotesize{$\mathrmbf{Tbl}(\mathcal{A}_{2})$}}}
\put(125,80){\makebox(0,0){\footnotesize{$\mathrmbf{Tbl}(\mathcal{A}_{1})$}}}
\put(60,5){\makebox(0,0){\footnotesize{$\mathrmbf{Tbl}$}}}
\put(60,92){\makebox(0,0){\scriptsize{$\grave{\mathrmbfit{tbl}}_{{\langle{f,g}\rangle}}$}}}
\put(24,38){\makebox(0,0)[r]{\scriptsize{$\mathrmbfit{inc}_{\mathcal{A}_{2}}$}}}
\put(97,38){\makebox(0,0)[l]{\scriptsize{$\mathrmbfit{inc}_{\mathcal{A}_{1}}$}}}
\put(60,54){\makebox(0,0){\shortstack{\scriptsize{$\grave{\chi}_{{\langle{f,g}\rangle}}$}\\\large{$\Longleftarrow$}}}}
\put(35,80){\vector(1,0){50}}
\put(9,68){\vector(3,-4){38}}
\put(111,68){\vector(-3,-4){38}}
\end{picture}
\end{tabular}}}
\\
\multicolumn{2}{c}{{\scriptsize{$\mathcal{A}_{2}\xrightleftharpoons{\;{\langle{f,g}\rangle}\;}\mathcal{A}_{1}$}}}
\end{tabular}
\\\\
{\scriptsize\setlength{\extrarowheight}{4pt}$\begin{array}{|@{\hspace{5pt}}l@{\hspace{15pt}}l@{\hspace{5pt}}|}
\multicolumn{1}{l}{\text{\bfseries levo}}
& 
\multicolumn{1}{l}{\text{\bfseries dextro}} 
\\ \hline
\acute{\chi}_{{\langle{f,g}\rangle}} : 
\acute{\mathrmbfit{tbl}}_{{\langle{f,g}\rangle}}\circ\mathrmbfit{inc}_{\mathcal{A}_{2}}\Leftarrow\mathrmbfit{inc}_{\mathcal{A}_{1}}
&
\grave{\chi}_{{\langle{f,g}\rangle}} : 
\mathrmbfit{inc}_{\mathcal{A}_{2}}\Leftarrow\grave{\mathrmbfit{tbl}}_{{\langle{f,g}\rangle}}\circ\mathrmbfit{inc}_{\mathcal{A}_{1}}
\\
\acute{\chi}_{{\langle{f,g}\rangle}} =
(\eta_{{\langle{f,g}\rangle}}\circ\mathrmbfit{inc}_{\mathcal{A}_{1}})
\bullet
(\acute{\mathrmbfit{tbl}}_{{\langle{f,g}\rangle}}\circ\grave{\chi}_{{\langle{f,g}\rangle}})
&
\grave{\chi}_{{\langle{f,g}\rangle}} =
(\grave{\mathrmbfit{tbl}}_{{\langle{f,g}\rangle}}\circ\acute{\chi}_{{\langle{f,g}\rangle}})
\bullet
(\varepsilon_{{\langle{f,g}\rangle}}\circ\mathrmbfit{inc}_{\mathcal{A}_{2}})
\rule[-7pt]{0pt}{10pt}
\\\hline
\end{array}$}
\end{tabular}
\end{center}
\caption{Inclusion Bridge: Tables}
\label{fig:incl:bridge:tbl}
\end{figure}
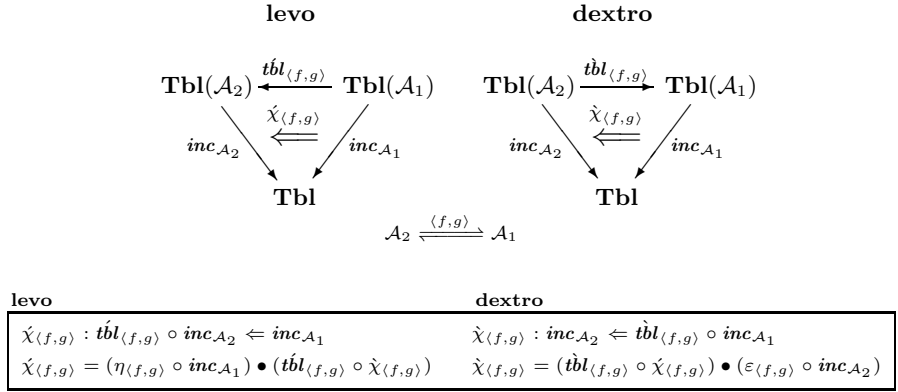
The inclusion bridges for the table fiber adjunction
%
\footnote{Notation from \S~3.4.2.}
%
\begin{center}
{{\begin{tabular}[b]{c}
\setlength{\unitlength}{0.6pt}
\begin{picture}(160,40)(0,0)
\put(0,20){\makebox(0,0){\footnotesize{$\mathrmbf{Tbl}(\mathcal{A}_{2})$}}}
\put(160,20){\makebox(0,0){\footnotesize{$\mathrmbf{Tbl}(\mathcal{A}_{1})$}}}
\put(80,40){\makebox(0,0){\scriptsize{$\acute{\mathrmbfit{tbl}}_{{\langle{f,g}\rangle}}$}}}
\put(80,0){\makebox(0,0){\scriptsize{$\grave{\mathrmbfit{tbl}}_{{\langle{f,g}\rangle}}$}}}
\put(55,20){\makebox(0,0){\tiny{$\eta_{{\langle{f,g}\rangle}}$}}}
\put(80,20){\makebox(0,0){\scriptsize{${\;\dashv\;}$}}}
\put(105,20){\makebox(0,0){\tiny{$\varepsilon_{{\langle{f,g}\rangle}}$}}}
\put(120,30){\vector(-1,0){80}}
\put(40,10){\vector(1,0){80}}
\end{picture}
\end{tabular}}}
\end{center}
are defined by
\[\mbox{\footnotesize{
$\overset{\textstyle{\acute{\mathrmbfit{tbl}}_{{\langle{f,g}\rangle}}(\mathcal{T}_{1})}}
{\overbrace{\langle{f^{\ast}(I_{1},s_{1}),{\scriptscriptstyle\sum}_{\acute{\tau}_{{\langle{f,g}\rangle}}(I_{1},s_{1})}(K_{1},t_{1})}\rangle}}
\xleftarrow
[{{\langle{\varepsilon_{f}(I_{1},s_{1}),f,g,1_{K_{1}}}\rangle}}]
{\acute{\chi}_{{\langle{f,g}\rangle}}(\mathcal{T}_{1})}
\overset{\textstyle{\mathcal{T}_{1}}}
{\overbrace{\langle{I_{1},s_{1},K_{1},t_{1}}\rangle}}$,
}}\]
\[\mbox{\footnotesize{$
\overset{\textstyle{\mathcal{T}_{2}}}
{\overbrace{\langle{I_{2},s_{2},K_{2},t_{2}}\rangle}}
\xrightarrow
[{{\langle{1_{I_{2}},f,g,\hat{k}}\rangle}}]
{\grave{\chi}_{{\langle{f,g}\rangle}}(\mathcal{T}_{2})}
\overset{\textstyle{\grave{\mathrmbfit{tbl}}_{{\langle{f,g}\rangle}}(\mathcal{T}_{2})}}
{\overbrace{\langle{{\scriptstyle\sum}_{f}(I_{2},s_{2}),(\grave{\tau}_{{\langle{f,g}\rangle}}(I_{2},s_{2}))^{\ast}(K_{2},t_{2})}\rangle}}
$.}}\]

\newpage

\begin{definition}\label{def:db:mor}
For any two databases
$\mathcal{R}_{2} = {\langle{\mathrmbf{R}_{2},\mathrmbfit{T}_{2},\mathcal{A}_{2}}\rangle}$
and
$\mathcal{R}_{1} = {\langle{\mathrmbf{R}_{1},\mathrmbfit{T}_{1},\mathcal{A}_{1}}\rangle}$,
a database morphism, 
with constant type domain morphism
$\mathcal{A}_{2}\xrightleftharpoons{{\langle{f,g}\rangle}}\mathcal{A}_{1}$, 
(Fig.\,\ref{fig:db:mor}) 
\begin{center}
{\normalsize{
$\mathcal{R}_{2} = {\langle{\mathrmbf{R}_{2},\mathrmbfit{T}_{2},\mathcal{A}_{2}}\rangle} 
\xleftarrow{{\langle{\mathrmbfit{R},\grave{\psi},f,g}\rangle}}
{\langle{\mathrmbf{R}_{1},\mathrmbfit{T}_{1},\mathcal{A}_{1}}\rangle} = \mathcal{R}_{1}$
}}
\end{center}
(Tbl.\,\ref{tbl:fole:morph}
in \S\,\ref{sub:sec:adj:components})
consists of
a shape-changing relation passage 
$\mathrmbfit{R} : \mathrmbf{R}_{2} \rightarrow \mathrmbf{R}_{1}$
and
a bridge
$\grave{\psi} : 
\mathrmbfit{T}_{2}\circ\grave{\mathrmbfit{tbl}}_{{\langle{f,g}\rangle}}
\Leftarrow 
\mathrmbfit{R}^{\mathrm{op}}\circ\mathrmbfit{T}_{1}$.
\end{definition}
The 
subcontext of 
{\ttfamily FOLE} relational databases 
(with constant type domains
and constant type domain morphisms) 
is denoted by $\mathrmbf{Db}\subseteq\mathrmbf{DB}$.
%
%
\footnote{See \S\,4.2.2 in the paper ``The {\ttfamily FOLE} Database'' \cite{kent:fole:era:db}.}
%

\newpage

\begin{definition}\label{def:db:mor:proj}
Using projections
(Fig.\,\ref{fig:db:mor:typ:dom}),
a \texttt{FOLE} database morphism, 
with constant type domain morphism
$\mathcal{A}_{2}\xrightleftharpoons{{\langle{f,g}\rangle}}\mathcal{A}_{1}$, 
%
\[\mbox{\footnotesize$
\mathcal{R}_{2}={\langle{\mathrmbf{R}_{2},\mathrmbfit{S}_{2},\mathcal{A}_{2},\mathrmbfit{K}_{2},\tau_{2}}\rangle}
\xleftarrow{{\langle{\mathrmbfit{R},\kappa,\grave{\varphi},f,g}\rangle}}
{\langle{\mathrmbf{R}_{1},\mathrmbfit{S}_{1},\mathcal{A}_{1},\mathrmbfit{K}_{1},\tau_{1}}\rangle}=\mathcal{R}_{1}
$\normalsize}\]
(Tbl.\,\ref{tbl:fole:morph}
in \S\,\ref{sub:sec:adj:components})
consists of 
a relation passage
$\mathrmbf{R}_{2}\xrightarrow{\,\mathrmbfit{R}\;}\mathrmbf{R}_{1}$ 
as above,
and
%
\begin{itemize}
\item 
a schema bridge
$
\grave{\varphi}
=
\grave{\psi}^{\mathrm{op}}{\circ\;}\mathrmbfit{sign}
:
\mathrmbfit{S}_{2}{\;\circ\;}{\scriptstyle\sum}_{f}
\Rightarrow
\mathrmbfit{R}{\;\circ\;}\mathrmbfit{S}_{1}$
\footnote{This defines a database schema morphism
\newline\mbox{}\hfill
$\mathcal{T}_{2} = {\langle{\mathrmbf{R}_{2},\mathrmbfit{S}_{2},X_{2}}\rangle}
\xrightleftharpoons{{\langle{\mathrmbfit{R},\grave{\varphi},f}\rangle}}
{\langle{\mathrmbf{R}_{1},\mathrmbfit{S}_{1},X_{1}}\rangle} = \mathcal{T}_{1}$,
\hfill\mbox{}\newline
which is the same as the (abstract) specification morphism in
Def.\,\ref{def:abs:spec:mor}
of \S\,\ref{sub:sec:spec:mor}.}
,
and
\item 
a key bridge
$\kappa
= \grave{\psi}{\;\circ\;}\mathrmbfit{key}
:
\mathrmbfit{K}_{2}
\Leftarrow
\mathrmbfit{R}^{\mathrm{op}}{\circ\;}\mathrmbfit{K}_{1}$
consisting of an $R_{2}$-indexed collection
$\{
\mathrmbfit{K}_{2}(r_{2})\xleftarrow{\kappa_{r_{2}}}\mathrmbfit{K}_{1}(r(r_{2})) 
\mid r_{2} \in R_{2}
\}$
of key functions.
\footnote{This defines a lax entity infomorphism
$\mathrmbfit{K}_{2}\xLeftarrow{\;\,\kappa\;}{r}{\;\circ\;}\mathrmbfit{K}_{1}$
consisting of
an $R_{2}$-indexed collection of key functions 
$\bigl\{
\mathrmbfit{K}_{2}(r_{2})\xleftarrow{\kappa_{r_{2}}}\mathrmbfit{K}_{1}(r(r_{2})) 
\mid r_{2} \in R_{2}
\bigr\}$}
%
\end{itemize}
These components satisfy the condition
%
\begin{equation}\label{eqn:db:mor:typ:dom:factor}
\kappa{\;\bullet\;}\tau_{2}
= 
(\mathrmbfit{R}^{\mathrm{op}}\!{\circ\;}\tau_{1})
{\;\bullet\;}(\grave{\varphi}^{\mathrm{op}}\!{\circ\;}\mathrmbfit{tup}_{\mathcal{A}_{1}})
{\;\bullet\;}(\mathrmbfit{S}_{2}^{\mathrm{op}}{\;\circ\;}\grave{\tau}_{{\langle{f,g}\rangle}}),
\end{equation}
%
as pictured in the following bridge diagram.
\begin{center}
{{\begin{tabular}{c}
\setlength{\unitlength}{0.66pt}
\begin{picture}(240,140)(0,-40)
\put(0,80){\makebox(0,0){\scriptsize{$\mathrmbfit{R}^{\mathrm{op}}{\circ\;}\mathrmbfit{K}_{1}$}}}
\put(0,0){\makebox(0,0){\scriptsize{$\mathrmbfit{K}_{2}$}}}
\put(125,80){\makebox(0,0){\scriptsize{$
\mathrmbfit{R}^{\mathrm{op}}{\circ\;}\mathrmbfit{S}_{1}^{\mathrm{op}}
{\circ\;}\mathrmbfit{tup}_{\mathcal{A}_{1}}$}}}
\put(125,0){\makebox(0,0){\scriptsize{$
\mathrmbfit{S}_{2}^{\mathrm{op}}
{\circ\;}\mathrmbfit{tup}_{\mathcal{A}_{2}}$}}}
\put(125,45){\makebox(0,0){\scriptsize{$({\grave{\varphi}}^{\mathrm{op}}{\circ\;}\mathrmbfit{tup}_{\mathcal{A}_{1}})
{\;\bullet\;}(\mathrmbfit{S}_{2}^{\mathrm{op}}{\circ\;}\grave{\tau}_{{\langle{f,g}\rangle}})$}}}
\put(50,96){\makebox(0,0){\scriptsize{$\mathrmbfit{R}^{\mathrm{op}}{\circ\;}\tau_{1}$}}}
\put(50,-10){\makebox(0,0){\scriptsize{$\tau_{2}$}}}
\put(-10,40){\makebox(0,0)[r]{\scriptsize{$\kappa$}}}
\put(50,80){\makebox(0,0){\large{$\xRightarrow{\;\;\;\;\;\;\;\;\;\;\;\;\;}$}}}
\put(50,0){\makebox(0,0){\large{$\xRightarrow{\;\;\;\;\;\;\;\;\;\;\;\;\;}$}}}
\put(0,40){\makebox(0,0){\large{$\bigg\Downarrow$}}}
\put(120,40){\makebox(0,0){\large{$\bigg\Downarrow$}}}
\put(220,80){\makebox(0,0){\scriptsize{$
{\scriptstyle\sum}_{f}^{\,\mathrm{op}}{\circ\;}\mathrmbfit{tup}_{\mathcal{A}_{1}}$}}}
\put(220,0){\makebox(0,0){\scriptsize{$\mathrmbfit{tup}_{\mathcal{A}_{2}}$}}}
\put(230,45){\makebox(0,0)[l]{\scriptsize{$\grave{\tau}_{{\langle{f,g}\rangle}}$}}}
\put(220,40){\makebox(0,0){\large{$\bigg\Downarrow$}}}
\put(0,-20){\makebox(0,0){\scriptsize{\textit{lax entity}}}}
\put(0,-30){\makebox(0,0){\scriptsize{\textit{infomorphism}}}}
\put(220,-20){\makebox(0,0){\scriptsize{\textit{tuple attribute}}}}
\put(220,-30){\makebox(0,0){\scriptsize{\textit{infomorphism}}}}
\end{picture}
\end{tabular}}}
\end{center}
\end{definition}
\begin{table}
\begin{center}
{\fbox{\footnotesize{\begin{minipage}{280pt}
\begin{description}
\item[ ] 
{\footnotesize{$
\mathcal{R}_{2} = 
{\langle{\mathrmbf{R}_{2},\mathrmbfit{K}_{2},\mathrmbfit{S}_{2},\tau_{2},\mathcal{A}_{2}}\rangle}
\xleftarrow{{\langle{\mathrmbfit{R},\kappa,\grave{\varphi},f,g}\rangle}}
{\langle{\mathrmbf{R}_{1},\mathrmbfit{K}_{1},\mathrmbfit{S}_{1},\tau_{1},\mathcal{A}_{1}}\rangle}
 = \mathcal{R}_{1}
$}}
\begin{itemize}
\comment{
\item schema morphism
\hfill
{\scriptsize{$\mathcal{S}_{2}={\langle{R_{2},\sigma_{2},X_{2}}\rangle}
\xRightarrow{\;{\langle{r,\grave{\varphi},f}\rangle}\;}
{\langle{R_{1},\sigma_{1},X_{1}}\rangle}=\mathcal{S}_{1}$}}
\item (lax) entity infomorphism
\hfill
{\scriptsize{$\mathcal{K}_{2}={\langle{\mathrmbf{R}_{2},\mathrmbfit{K}_{2}}\rangle} 
\xrightarrow{\;{\langle{\mathrmbfit{R},\kappa}\rangle}\;}
{\langle{\mathrmbf{R}_{1},\mathrmbfit{K}_{1}}\rangle}=\mathcal{K}_{1}$}}
}
\item 
schema bridge
\hfill
{\footnotesize{$\grave{\varphi} 
:
\mathrmbfit{S}_{2}{\;\circ\;}{\scriptstyle\sum}_{f}
\Rightarrow
\mathrmbfit{R}{\;\circ\;}\mathrmbfit{S}_{1}$}}
\item 
a key bridge 
\hfill
{\footnotesize{$\kappa 
:
\mathrmbfit{K}_{2}\Leftarrow\mathrmbfit{R}^{\mathrm{op}}{\circ\;}\mathrmbfit{K}_{1}$}}
\item
constraint 
\hfill
{{$\kappa\bullet\tau_{2}
= (\mathrmbfit{R}^{\mathrm{op}}\!{\circ\;}\tau_{1})
{\;\bullet\;}(\grave{\varphi}^{\mathrm{op}}\!{\circ\;}\mathrmbfit{tup}_{\mathcal{A}_{1}})
{\;\bullet\;}(\mathrmbfit{S}_{2}^{\mathrm{op}}{\;\circ\;}\grave{\tau}_{{\langle{f,g}\rangle}})$}}
\end{itemize}
\end{description}
\end{minipage}}}}
\end{center}
\caption{Database Morphism}
\label{tbl:db:mor}
\end{table}
\begin{center}
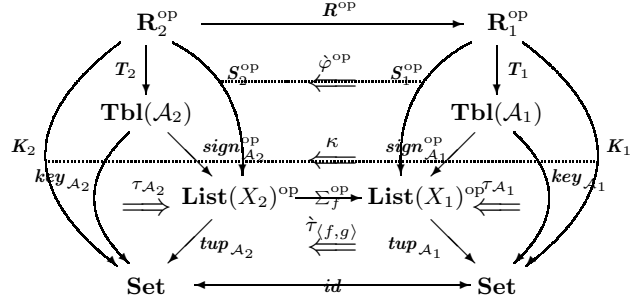
\begin{figure}
{\begin{tabular}{@{\hspace{106pt}}c}
\setlength{\unitlength}{0.55pt}
\begin{picture}(240,200)(0,-10)
\put(125,190){\makebox(0,0){\scriptsize{$\mathrmbfit{R}^{\mathrm{op}}$}}}
\put(120,60){\makebox(0,0){\scriptsize{${\scriptstyle\sum}_{f}^{\mathrm{op}}
$}}}
\put(120,0){\makebox(0,0){\scriptsize{$\mathrmbfit{id}$}}}
\put(30,180){\vector(1,0){180}}
\put(95,60){\vector(1,0){45}}
\put(25,0){\vector(1,0){190}}
\put(25,0){\vector(-1,0){0}}
\put(120,146){\makebox(0,0){\shortstack{\scriptsize{$\;\grave{\varphi}^{\mathrm{op}}$}
\\\large{$\Longleftarrow$}}}}
\put(120,92){\makebox(0,0){\large{$\overset{\rule[-2pt]{0pt}{5pt}\kappa}
{\Longleftarrow}$}}}
\put(120,35){\makebox(0,0){\large{$
\overset{\grave{\tau}_{{\langle{f,g}\rangle}}}
{\Longleftarrow}
$}}}
\qbezier[60](42,140)(110,140)(180,140)
\qbezier[150](-75,85)(120,85)(300,85)
\put(-8,0){\begin{picture}(0,0)(0,0)
\put(8,180){\makebox(0,0){\footnotesize{$\mathrmbf{R}_{2}^{\mathrm{op}}$}}}
\put(0,118){\makebox(0,0){\footnotesize{$\mathrmbf{Tbl}(\mathcal{A}_{2})$}}}
\put(65,60){\makebox(0,0){\footnotesize{$\mathrmbf{List}(X_{2})^{\mathrm{op}}$}}}
\put(0,0){\makebox(0,0){\footnotesize{$\mathrmbf{Set}$}}}
\put(-6,148){\makebox(0,0)[r]{\scriptsize{$\mathrmbfit{T}_{2}$}}}
\put(-75,95){\makebox(0,0)[r]{\scriptsize{$\mathrmbfit{K}_{2}$}}}
\put(55,145){\makebox(0,0)[l]{\scriptsize{$\mathrmbfit{S}_{2}^{\mathrm{op}}$}}}
\put(-37,72){\makebox(0,0)[r]{\scriptsize{$\mathrmbfit{key}_{\mathcal{A}_{2}}$}}}
\put(38,93){\makebox(0,0)[l]{\scriptsize{$\mathrmbfit{sign}_{\mathcal{A}_{2}}^{\mathrm{op}}$}}}
\put(36,26){\makebox(0,0)[l]{\scriptsize{$\mathrmbfit{tup}_{\mathcal{A}_{2}}$}}}
\put(0,60){\makebox(0,0){\shortstack{\scriptsize{$
\;\tau_{\mathcal{A}_{2}}$}\\\large{$\Longrightarrow$}}}}
\put(0,165){\vector(0,-1){34}}
\put(15,105){\vector(1,-1){30}}
\put(45,45){\vector(-1,-1){30}}
\qbezier(-18,167)(-120,90)(-20,13)\put(-20,13){\vector(1,-1){0}}
\qbezier(-12,105)(-60,60)(-12,15)\put(-12,15){\vector(1,-1){0}}
\qbezier(18,167)(70,140)(66,76)\put(66,76){\vector(0,-1){0}}
\end{picture}}
\put(233,0){\begin{picture}(0,0)(0,0)
\put(8,180){\makebox(0,0){\footnotesize{$\mathrmbf{R}_{1}^{\mathrm{op}}$}}}
\put(0,118){\makebox(0,0){\footnotesize{$\mathrmbf{Tbl}(\mathcal{A}_{1})$}}}
\put(-48,60){\makebox(0,0){\footnotesize{$\mathrmbf{List}(X_{1})^{\mathrm{op}}$}}}
\put(0,0){\makebox(0,0){\footnotesize{$\mathrmbf{Set}$}}}
\put(6,148){\makebox(0,0)[l]{\scriptsize{$\mathrmbfit{T}_{1}$}}}
\put(75,95){\makebox(0,0)[l]{\scriptsize{$\mathrmbfit{K}_{1}$}}}
\put(-50,145){\makebox(0,0)[r]{\scriptsize{$\mathrmbfit{S}_{1}^{\mathrm{op}}$}}}
\put(37,72){\makebox(0,0)[l]{\scriptsize{$\mathrmbfit{key}_{\mathcal{A}_{1}}$}}}
\put(-32,93){\makebox(0,0)[r]{\scriptsize{$\mathrmbfit{sign}_{\mathcal{A}_{1}}^{\mathrm{op}}$}}}
\put(-36,26){\makebox(0,0)[r]{\scriptsize{$\mathrmbfit{tup}_{\mathcal{A}_{1}}$}}}
\put(0,60){\makebox(0,0){\shortstack{\scriptsize{$
\;\tau_{\mathcal{A}_{1}}$}\\\large{$\Longleftarrow$}}}}
\put(0,165){\vector(0,-1){34}}
\put(-15,105){\vector(-1,-1){30}}
\put(-45,45){\vector(1,-1){30}}
\qbezier(18,167)(120,90)(20,13)\put(20,13){\vector(-1,-1){0}}
\qbezier(12,105)(60,60)(12,15)\put(12,15){\vector(-1,-1){0}}
\qbezier(-18,167)(-70,140)(-66,76)\put(-66,76){\vector(0,-1){0}}
\end{picture}}
%
%
\end{picture}
\end{tabular}}
\caption{Database Morphism: Type Domain}
\label{fig:db:mor:typ:dom}
\end{figure}
\end{center}

%

\newpage

\begin{proposition}\label{db:mor:2:struc:mor}
\label{prop:db:mor:to:struc:mor}
%
\comment{
Any 
{\ttfamily FOLE} database morphism in $\mathrmbf{Db}$
\newline\mbox{}\hfill
$\mathcal{R}_{2}={\langle{\mathrmbf{R}_{2},\mathrmbfit{S}_{2},\mathcal{A}_{2},\mathrmbfit{K}_{2},\tau_{2}}\rangle}
\xleftarrow{{\langle{\mathrmbfit{R},\kappa,\grave{\varphi},f,g}\rangle}}
{\langle{\mathrmbf{R}_{1},\mathrmbfit{S}_{1},\mathcal{A}_{1},\mathrmbfit{K}_{1},\tau_{1}}\rangle}=\mathcal{R}_{1}$,
\hfill\mbox{}\newline
has a companion (lax) {\ttfamily FOLE} structure morphism
$\mathcal{M}_{2}\xrightleftharpoons{{\langle{r,\kappa,\grave{\varphi},f,g}\rangle}}\mathcal{M}_{1}$
in $\mathrmbf{Struc}$.
}
%
The constraint-free aspect of a {\ttfamily FOLE} database morphism 
\newline\mbox{}\hfill
$\mathcal{R}_{2}={\langle{\mathrmbf{R}_{2},\mathrmbfit{S}_{2},\mathcal{A}_{2},\mathrmbfit{K}_{2},\tau_{2}}\rangle}
\xleftarrow{{\langle{\mathrmbfit{R},\kappa,\grave{\varphi},f,g}\rangle}}
{\langle{\mathrmbf{R}_{1},\mathrmbfit{S}_{1},\mathcal{A}_{1},\mathrmbfit{K}_{1},\tau_{1}}\rangle}=\mathcal{R}_{1}$,
\hfill\mbox{}\newline
in $\mathrmbf{Db}$
defines,
is the same as,
a (lax) {\ttfamily FOLE} structure morphism 
\newline\mbox{}\hfill
$\mathcal{M}_{2}=
{\langle{\mathcal{E}_{2},\sigma_{2},\tau_{2},\mathcal{A}_{2}}\rangle}
\xrightleftharpoons{{\langle{r,\kappa,\grave{\varphi},f,g}\rangle}}
{\langle{\mathcal{E}_{1},\sigma_{1},\tau_{1},\mathcal{A}_{1}}\rangle}=\mathcal{M}_{1}$
\hfill\mbox{}\newline
in $\mathring{\mathrmbf{Struc}}$.
\end{proposition}
\begin{proof}
We define the various components of the (lax) structure morphism
above, 
as defined in Def.\ref{def:lax:struc:mor} 
and pictured in Fig.\,\ref{fig:lax:struc:mor} 
of \S\,\ref{sub:sec:struc:mor:lax}.
\begin{itemize}
\item[$r$\textbf{:}]
The predicate function $R_{2}\xrightarrow{\;r\,}R_{1}$
is the constraint-free aspect of the relation passage 
$\mathrmbf{R}_{2} \xrightarrow{\mathrmbfit{R}} \mathrmbf{R}_{1}$.
\item[${\langle{f,g}\rangle}$\textbf{:}]
The type domain morphism
$\mathcal{A}_{2} 
\xrightleftharpoons{{\langle{f,g}\rangle}} 
\mathcal{A}_{1}$
is given.
\item[${\langle{r,\grave{\varphi},f}\rangle}$\textbf{:}]
The 
schema morphism
{\footnotesize{
$\mathcal{S}_{2}={\langle{R_{2},{\sigma_{2}},X_{2}}\rangle}
\xRightarrow{{\langle{r,\grave{\varphi},f}\rangle}}
{\langle{R_{1},{\sigma_{1}},X_{1}}\rangle}=\mathcal{S}_{1}$,}}
consisting of the
$R_{2}$-indexed 
collection 
of signature morphisms
$\{
{\scriptstyle\sum}_{f}({\sigma_{2}}(r_{2}))
\xrightarrow[h]{\;\grave{\varphi}_{r_{2}}\;}{\sigma_{1}}(r(r_{2}))
\mid r_{2} \in R_{2}
\}$,
is the constraint-free aspect of the database schema morphism
in Def.\,\ref{def:db:mor:proj}.
\item[${\langle{r,\kappa}\rangle}$\textbf{:}]
The lax entity infomorphism
$\mathcal{E}_{2}
 = {\langle{R_{2},\mathrmbfit{K}_{2}}\rangle} 
\xleftharpoondown{{\langle{r,\kappa}\rangle}}
{\langle{R_{1},\mathrmbfit{K}_{1}}\rangle} = 
\mathcal{E}_{1}$,
consisting of
the 
$R_{2}$-indexed 
collection of key functions
$\bigl\{{\mathrmbfit{K}_{2}}(r_{2})\xleftarrow{\kappa_{r_{2}}}{\mathrmbfit{K}_{1}}(r(r_{2}))
\mid r_{2} \in R_{2}
\bigr\}$,
is the constraint-free aspect of the key bridge
$\kappa :
\mathrmbfit{K}_{2}
\Leftarrow
\mathrmbfit{R}^{\mathrm{op}}{\circ\;}\mathrmbfit{K}_{1}$.
These components 
satisfy the condition
%
\footnote{This is the constraint-free aspect of the database morphism condition
\newline\mbox{}\hfill
{\footnotesize$
\kappa{\;\bullet\;}\tau_{2}
= 
(\mathrmbfit{R}^{\mathrm{op}}\!{\circ\;}\tau_{1})
{\;\bullet\;}(\grave{\varphi}^{\mathrm{op}}\!{\circ\;}\mathrmbfit{tup}_{\mathcal{A}_{1}})
{\;\bullet\;}(\mathrmbfit{S}_{2}^{\mathrm{op}}{\;\circ\;}\grave{\tau}_{{\langle{f,g}\rangle}})
$.\normalsize}
\hfill\mbox{}\newline
See Def. \ref{def:db:mor:proj} above.}
%
\footnote{See Disp.\ref{lax:struc:mor:cond}
in Def.\ref{def:lax:struc:mor} 
of 
\S\,\ref{sub:sec:struc:mor:lax}.}
%
\[\mbox
{\footnotesize{$
{{
\Big\{
\kappa_{r_{2}}{\;\cdot\;}\tau_{2,r_{2}}
= 
\tau_{1,r(r_{2})}{\;\cdot\;}
\underset{\mathrmbfit{tup}(\grave{\varphi}_{r_{2}},f,g)}
{\underbrace{\mathrmbfit{tup}_{\mathcal{A}_{1}}(\grave{\varphi}_{r_{2}})
{\;\cdot\;}
\grave{\tau}_{{\langle{f,g}\rangle}}({\sigma_{2}}(r_{2}))}}
\mid r_{2} \in R_{2}
\Bigr\}
.}}
$}.
\normalsize}
\]
\end{itemize}
\hfill\rule{5pt}{5pt}
\end{proof}
%


%
\begin{proposition}\label{prop:db:mor:2:snd:log:mor}
Any 
{\ttfamily FOLE} database morphism 
\newline\mbox{}\hfill
$\mathcal{R}_{2}={\langle{\mathrmbf{R}_{2},\mathrmbfit{S}_{2},\mathcal{A}_{2},\mathrmbfit{K}_{2},\tau_{2}}\rangle}
\xleftarrow{{\langle{\mathrmbfit{R},\kappa,\grave{\varphi},f,g}\rangle}}
{\langle{\mathrmbf{R}_{1},\mathrmbfit{S}_{1},\mathcal{A}_{1},\mathrmbfit{K}_{1},\tau_{1}}\rangle}=\mathcal{R}_{1}$
\hfill\mbox{}\newline
in $\mathrmbf{Db}$
defines a (lax) {\ttfamily FOLE} sound logic morphism
\newline\mbox{}\hfill
$\mathcal{L}_{2}={\langle{\mathcal{S}_{2},\mathcal{M}_{2},\mathcal{T}_{2}}\rangle}
\xrightleftharpoons{{\langle{\mathrmbfit{R},\kappa,\grave{\varphi},f,g}\rangle}}
{\langle{\mathcal{S}_{1},\mathcal{M}_{1},\mathcal{T}_{1}}\rangle}=\mathcal{L}_{1}$
\hfill\mbox{}\newline
in $\mathring{\mathrmbf{Snd}}$.
\end{proposition}
\begin{proof}
The source and target sound logics are defined by 
Prop.~\ref{prop:db:2:snd:log} above.
The structure morphism is given by Prop.~\ref{db:mor:2:struc:mor} above.
%
The (abstract) specification morphism
is the same as
the database schema morphism
\newline\mbox{}\hfill
$\mathcal{T}_{2} = {\langle{\mathrmbf{R}_{2},\mathrmbfit{S}_{2},X_{2}}\rangle}
\xrightleftharpoons{{\langle{\mathrmbfit{R},\grave{\varphi},f}\rangle}}
{\langle{\mathrmbf{R}_{1},\mathrmbfit{S}_{1},X_{1}}\rangle} = \mathcal{T}_{1}$,
\hfill\mbox{}\newline
as discussed in \S\,\ref{sub:sec:spec}.
\mbox{}\hfill\rule{5pt}{5pt}
\end{proof}
\begin{theorem}\label{thm:db:2:snd:log}
There is a passage
$\mathrmbf{Db}^{\mathrm{op}}\xrightarrow{\;\mathring{\mathrmbfit{snd}}\;}
\mathring{\mathrmbf{Snd}}$.
%
\end{theorem}
\begin{proof}
A 
database 
$\mathcal{R} = {\langle{\mathrmbf{R},\mathrmbfit{T},\mathcal{A}}\rangle}$
in $\underline{\mathrmbf{Db}}$
is mapped to its associated 
sound logic 
$\mathcal{L}={\langle{\mathcal{S},\mathcal{M},\mathrmbf{T}}\rangle}$
by Prop.~\ref{prop:db:2:snd:log}.
A
database morphism
$\mathcal{R}_{2} 
\xleftarrow{{\langle{\mathrmbfit{R},\psi,f,g}\rangle}}
\mathcal{R}_{1}$
in $\underline{\mathrmbf{Db}}$
is mapped to its associated 
sound logic morphism 
$\mathcal{L}_{2}
\xrightleftharpoons{{\langle{\mathrmbfit{R},\kappa,\alpha,f,g}\rangle}}
\mathcal{L}_{1}$
by Prop.~\ref{prop:db:mor:2:snd:log:mor}.
\mbox{}\hfill\rule{5pt}{5pt}
%
\end{proof}

\begin{theorem}\label{thm:db:iso:snd:log}
The contexts of 
databases
and
(lax) sound logics form a reflection
\begin{center}
{{\footnotesize{\setlength{\extrarowheight}{2pt}{$\begin{array}{c}
{\setlength{\unitlength}{0.6pt}\begin{picture}(120,30)(0,-12)
\put(10,0){\makebox(0,0){\footnotesize{$\mathrmbf{Db}^{\mathrm{op}}$}}}
\put(110,2){\makebox(0,0){\footnotesize{$\mathring{\mathrmbf{Snd}}$}}}
\put(60,22){\makebox(0,0){\scriptsize{$\mathrmbfit{im}$}}}
\put(62,0){\makebox(0,0){\scriptsize{$\dashv$}}}
\put(66,-20){\makebox(0,0){\scriptsize{$\mathrmbfit{inc}$}}}
\put(35,10){\vector(1,0){50}}
\put(85,-10){\vector(-1,0){50}}
\put(85,-6){\oval(8,8)[r]}
\end{picture}}
\end{array}$,}}}}
\end{center}
so that these two representation of \texttt{FOLE}  
are ``informationally equivalent''. 
%
\footnote{The database-logic reflection 
\newline\mbox{}\hfill
{{$\mathrmbf{Db}(\mathcal{A})
\;\xrightarrow{{\langle{\mathrmbfit{im}_{\mathcal{A}}{\;\dashv\;}\mathrmbfit{inc}_{\mathcal{A}}}\rangle}}\;
\mathrmbf{Log}(\mathcal{A})$}}
\hfill\mbox{}\newline
generalizes 
the table-relation reflection
$\mathrmbf{Tbl}(\mathcal{A})
\;\xrightarrow{{\langle{\mathrmbfit{im}_{\mathcal{A}}{\;\dashv\;}\mathrmbfit{inc}_{\mathcal{A}}}\rangle}}\;
\mathrmbf{Rel}(\mathcal{A})$.
These reflections embody the notion of informational equivalence.}
%
%
\footnote{See Prop.\,12 in \S\,A.1 of the paper
\cite{kent:fole:era:tbl}
``The {\ttfamily FOLE} Table''.}
%
%
\end{theorem}
\begin{proof}
The passages in 
Thm.\,\ref{thm:snd:log:2:db}
and
Thm.\,\ref{thm:db:2:snd:log}
form a reflection.
Here,
we use 
the image part 
$\mathrmbf{Tbl}{\;\xhookrightarrow{\;\mathrmbfit{im}\;}\;}\mathrmbf{Rel}$
of 
the table-relation reflection.
%
\footnote{\label{tbl:rel:refl}The reflection 
${\langle{\mathrmbfit{im}{\;\dashv\;}\mathrmbfit{inc}}\rangle}
:\mathrmbf{Tbl}{\;\rightleftarrows\;}\mathrmbf{Rel}$
(``The {\ttfamily FOLE} Table''\cite{kent:fole:era:tbl})
of the context of relations into the context of tables
embodies the notion of ``informational equivalence''.}
%
For an alternate proof, 
first compose with image part of the table-relation reflection
$^{\mathrm{\ref{tbl:rel:refl}}}$, 
and then restrict to fixed type domains.
\mbox{}\hfill\rule{5pt}{5pt}
\end{proof}
\newpage
\appendix
\section{Appendix}\label{sec:append}


\subsection{\texttt{FOLE} Components}
\label{sub:sec:adj:components}


\paragraph{Bridges.}

Although adjointly equivalent,
no levo bridges are used throughout this paper,
except for the levo tuple bridge
{\footnotesize{${f^{\ast}}^{\mathrm{op}}{\circ\;}\mathrmbfit{tup}_{\mathcal{A}_{2}}
\xLeftarrow{\;\acute{\tau}_{{\langle{f,g}\rangle}}\;}
\mathrmbfit{tup}_{\mathcal{A}_{1}}$}}
discussed in footnote\,\ref{tbl:fbr:adj}
in \S\,\ref{sub:sec:db:mor}.

%
\begin{table}
\begin{center}
{\footnotesize\setlength{\extrarowheight}{4pt}$\begin{array}{|@{\hspace{5pt}}l@{\hspace{15pt}}l@{\hspace{5pt}}|}
\multicolumn{1}{l}{\text{\bfseries levo}} & \multicolumn{1}{l}{\text{\bfseries dextro}}
\\ \hline
\acute{\varphi}:\mathrmbfit{S}_{2}\Rightarrow\mathrmbfit{R}\circ\mathrmbfit{S}_{1}\circ{f^{\ast}}
& 
\grave{\varphi}:\mathrmbfit{S}_{2}\circ{{\Sigma}_{f}}\Rightarrow\mathrmbfit{R}\circ\mathrmbfit{S}_{1}
\\
\acute{\varphi}=(\mathrmbfit{S}_{2}\circ\eta_{f})\bullet(\grave{\varphi}\circ{f^{\ast}})
& 
\grave{\varphi}=(\acute{\varphi}\circ{{\Sigma}_{f}})\bullet(\mathrmbfit{R}\circ\mathrmbfit{S}_{1}\circ\varepsilon_{f})
\\
\acute{\varphi}=(\acute{\psi}\circ\mathrmbfit{sign}_{\mathcal{A}_{2}}^{\mathrm{op}})^{\mathrm{op}}
&
\grave{\varphi}=(\grave{\psi}\circ\mathrmbfit{sign}_{\mathcal{A}_{1}}^{\mathrm{op}})^{\mathrm{op}}
\\ \hline
\acute{\psi}:\mathrmbfit{T}_{2}\Leftarrow\mathrmbfit{R}^{\mathrm{op}}\circ\mathrmbfit{T}_{1}\circ\acute{\mathrmbfit{tbl}}_{{\langle{f,g}\rangle}}
&
\grave{\psi}:\mathrmbfit{T}_{2}\circ\grave{\mathrmbfit{tbl}}_{{\langle{f,g}\rangle}}\Leftarrow\mathrmbfit{R}^{\mathrm{op}}\circ\mathrmbfit{T}_{1}
\\
\acute{\psi}=(\grave{\psi}\circ\acute{\mathrmbfit{tbl}}_{{\langle{f,g}\rangle}})\bullet(\mathrmbfit{T}_{2}\circ\varepsilon_{{\langle{f,g}\rangle}})
&
\grave{\psi}=(\mathrmbfit{R}^{\mathrm{op}}\circ\mathrmbfit{T}_{1})\eta_{{\langle{f,g}\rangle}}\bullet(\acute{\psi}\circ\grave{\mathrmbfit{tbl}}_{{\langle{f,g}\rangle}})
\\ \hline
\acute{\chi}_{{\langle{f,g}\rangle}}:\acute{\mathrmbfit{tbl}}_{{\langle{f,g}\rangle}}\circ\mathrmbfit{inc}_{\mathcal{A}_{2}}\Leftarrow\mathrmbfit{inc}_{\mathcal{A}_{1}}
&
\grave{\chi}_{{\langle{f,g}\rangle}}:\mathrmbfit{inc}_{\mathcal{A}_{2}}\Leftarrow\grave{\mathrmbfit{tbl}}_{{\langle{f,g}\rangle}}\circ\mathrmbfit{inc}_{\mathcal{A}_{1}}
\\
\acute{\chi}_{{\langle{f,g}\rangle}}=(\eta_{{\langle{f,g}\rangle}}\circ\mathrmbfit{inc}_{\mathcal{A}_{1}})\bullet(\acute{\mathrmbfit{tbl}}_{{\langle{f,g}\rangle}}\circ\grave{\chi}_{{\langle{f,g}\rangle}})
&
\grave{\chi}_{{\langle{f,g}\rangle}}=(\grave{\mathrmbfit{tbl}}_{{\langle{f,g}\rangle}}\circ{\chi}_{{\langle{f,g}\rangle}})\bullet(\varepsilon_{{\langle{f,g}\rangle}}\circ\mathrmbfit{inc}_{\mathcal{A}_{2}})
\\ \hline
\acute{\tau}_{{\langle{f,g}\rangle}}:(f^{\ast})^{\mathrm{op}} \circ \mathrmbfit{tup}_{\mathcal{A}_{2}}\Leftarrow\mathrmbfit{tup}_{\mathcal{A}_{1}} 
&
\grave{\tau}_{{\langle{f,g}\rangle}}:\mathrmbfit{tup}_{\mathcal{A}_{2}}\Leftarrow{\scriptstyle\sum}_{f}^{\mathrm{op}} \circ \mathrmbfit{tup}_{\mathcal{A}_{1}} 
\\
\acute{\tau}_{{\langle{f,g}\rangle}}=(\varepsilon_{f}^{\mathrm{op}} \circ \mathrmbfit{tup}_{\mathcal{A}_{1}})\bullet((f^{\ast})^{\mathrm{op}} \circ \grave{\tau}_{{\langle{f,g}\rangle}})
&
\grave{\tau}_{{\langle{f,g}\rangle}}=({\scriptstyle\sum}_{f}^{\mathrm{op}} \circ \acute{\tau}_{{\langle{f,g}\rangle}})\bullet(\eta_{f}^{\mathrm{op}} \circ \mathrmbfit{tup}_{\mathcal{A}_{2}})
\\ \hline
\multicolumn{2}{|c|}{\xi:\mathrmbfit{T}_{2}\circ\mathrmbfit{inc}_{\mathcal{A}_{2}}\Leftarrow\mathrmbfit{R}^{\mathrm{op}}\circ\mathrmbfit{T}_{1}\circ\mathrmbfit{inc}_{\mathcal{A}_{1}}}
\\
\multicolumn{2}{|c|}
{
(\mathrmbfit{R}^{\mathrm{op}}{\!\circ}\mathrmbfit{T}_{1}\circ\acute{\chi}_{{\langle{f,g}\rangle}})
\bullet
(\acute{\psi}\circ\mathrmbfit{inc}_{\mathcal{A}_{2}})
\;=\;\xi\;=\;
(\grave{\psi}\circ\mathrmbfit{inc}_{\mathcal{A}_{1}})
\bullet
(\mathrmbfit{T}_{2}{\;\circ\;}\grave{\chi}_{{\langle{f,g}\rangle}})
}
\\ 
\hline
\end{array}$}
\end{center}
\caption{\texttt{FOLE} Adjoint Bridges}
\label{tbl:adjoints}
\end{table}
%

%
%

\newpage
\paragraph{Morphisms.}

The \texttt{FOLE} equivalence is explained and understood 
principally in terms of its various morphisms (Tbl.\,\ref{tbl:fole:morph}).
\begin{table}
\begin{center}
{\footnotesize{\setlength{\extrarowheight}{2pt}
\begin{tabular}
{|@{\hspace{5pt}}r@{\hspace{20pt}}c@{\hspace{10pt}}|}
\hline
schema morphism 
&
$\mathcal{S}_{2}
= {\langle{R_{2},\sigma_{2},X_{2}}\rangle}
\xRightarrow{\;{\langle{r,\grave{\varphi},f}\rangle}\;}
{\langle{R_{1},\sigma_{1},X_{1}}\rangle} = 
\mathcal{S}_{1}$
\\
schemed domain morphism
&
{\footnotesize{$
\mathcal{D}_{2}
= 
{\langle{R_{2},\sigma_{2},\mathcal{A}_{2}}\rangle} 
\xrightarrow{\;{\langle{r,\grave{\varphi},f,g}\rangle}\;}
{\langle{R_{1},\sigma_{1},\mathcal{A}_{1}}\rangle}
= \mathcal{D}_{1}
$}}
\\
(lax) structure morphism
&
{\footnotesize{
$\mathcal{M}_{2} = {\langle{\mathcal{E}_{2},\sigma_{2},\tau_{2},\mathcal{A}_{2}}\rangle}
\xrightleftharpoons{{\langle{r,\kappa,\grave{\varphi},f,g}\rangle}}
{\langle{\mathcal{E}_{1},\sigma_{1},\tau_{1},\mathcal{A}_{1}}\rangle} = \mathcal{M}_{1}$
}}
\\
specification morphism
&
{\footnotesize{$\mathcal{T}_{2}={\langle{\mathrmbf{R}_{2},\mathrmbfit{S}_{2},X_{2}}\rangle}
\xrightarrow{{\langle{\mathrmbfit{R},\grave{\varphi},f}\rangle}}
{\langle{\mathrmbf{R}_{1},\mathrmbfit{S}_{1},X_{1}}\rangle}=\mathcal{T}_{1}$}}
\\
sound logic morphism
&
{\footnotesize{$\mathcal{L}_{2}={\langle{\mathcal{S}_{2},\mathcal{M}_{2},\mathcal{T}_{2}}\rangle}
\xrightleftharpoons{{\langle{\mathrmbfit{R},\kappa,\grave{\varphi},f,g}\rangle}}
{\langle{\mathcal{S}_{1},\mathcal{M}_{1},\mathcal{T}_{1}}\rangle}=\mathcal{L}_{1}$}}
\\\hline
\multicolumn{2}{c}{}
\\\hline
$\text{database morphism}_{1}$
&
$\mathcal{R}_{2} = {\langle{\mathrmbf{R}_{2},\mathrmbfit{T}_{2}}\rangle} 
\xleftarrow{{\langle{\mathrmbfit{R},\xi}\rangle}}
{\langle{\mathrmbf{R}_{1},\mathrmbfit{T}_{1}}\rangle} = \mathcal{R}_{1}
$
\\
$\text{database morphism}_{2}$
&
$\mathcal{R}_{2} = {\langle{\mathrmbf{R}_{2},\mathrmbfit{T}_{2},\mathcal{A}_{2}}\rangle} 
\xleftarrow{{\langle{\mathrmbfit{R},\grave{\psi},f,g}\rangle}}
{\langle{\mathrmbf{R}_{1},\mathrmbfit{T}_{1},\mathcal{A}_{1}}\rangle} = \mathcal{R}_{1}$
\\
$\text{database morphism}_{3}$
&
$\mathcal{R}_{2} = {\langle{\mathrmbf{R}_{2},\mathrmbfit{S}_{2},\mathcal{A}_{2},\mathrmbfit{K}_{2},\tau_{2}}\rangle}
\xleftarrow{{\langle{\mathrmbfit{R},\kappa,\grave{\varphi},f,g}\rangle}}
{\langle{\mathrmbf{R}_{1},\mathrmbfit{S}_{1},\mathcal{A}_{1},\mathrmbfit{K}_{1},\tau_{1}}\rangle} = \mathcal{R}_{1}$
\\\cline{1-1}
\multicolumn{2}{|l|}{\hspace{30pt}
{$\scriptsize{\text{
1 full, 2 fixed type domain, 3 with projections
}}$}}
\\\hline
\end{tabular}}}
\end{center}
\caption{\texttt{FOLE} Morphisms}
\label{tbl:fole:morph}
\end{table}
%

%
\newpage
\subsection{Classifications and Infomorphisms}
\label{sub:sec:class:info}

The concept of a ``classification'' comes from the theory of Information Flow:
see the book 
{\itshape Information Flow: The Logic of Distributed Systems}
by Barwise and Seligman
\cite{barwise:seligman:97}.
A classification is also important in the theory of Formal Concept Analysis,
where it is called a ``formal context'':
see the book 
{\itshape Formal Concept Analysis: Mathematical Foundations}
by Ganter and Wille
\cite{ganter:wille:99}.


\paragraph{Classification.}

A classification
$\mathcal{A}={\langle{X,Y,\models_{\mathcal{A}}}\rangle}$ consists of:
a set
$Y = \mathrmbfit{inst}(\mathcal{A})$
of objects to be classified, called tokens or instances;
a set
$X = \mathrmbfit{typ}(\mathcal{A})$
of objects that classify the instances, called types; and
a binary relation $\models_{\mathcal{A}}{\;\subseteq\;}Y{\times}X$ between instances and types.
%
The \textit{extent} of any type $x{\,\in\,}X$
is the subset of instances classified by $x$:
$\mathrmbfit{ext}_{\mathcal{A}}(x) 
= \{ y \in Y \mid y\;\models_{\mathcal{A}}x \}$.
Dually,
the \textit{intent} of any instance $y{\,\in\,}Y$
is the subset of types that classify $y$:
$\mathrmbfit{int}_{\mathcal{A}}(y)
= \{ x \in X \mid y\;\models_{\mathcal{A}}x \}$.
Two types $x,x' \in X$ are \textit{coextensive} in $\mathcal{A}$ when
$\mathrmbfit{ext}_{\mathcal{A}}(x)
=\mathrmbfit{ext}_{\mathcal{A}}(x')$.
Two instances $y,y' \in Y$ are \textit{cointensive} or \textit{indistinguishable} in $\mathcal{A}$ when
$\mathrmbfit{int}_{\mathcal{A}}(y)
=\mathrmbfit{int}_{\mathcal{A}}(y')$.
A classification
$\mathcal{A}$ is \textit{separated} or \textit{intensional} 
when there are no two indistinguishable instances:
$y \neq y'$ implies $\mathrmbfit{int}_{\mathcal{A}}(y)\neq\mathrmbfit{int}_{\mathcal{A}}(y')$.
A classification
$\mathcal{A}$ is \textit{extensional} when all coextensive types are identical:
$\mathrmbfit{ext}_{\mathcal{A}}(x)
=\mathrmbfit{ext}_{\mathcal{A}}(x')$
implies $x=x'$;
equivalently,
$x \neq x'$
implies
$\mathrmbfit{ext}_{\mathcal{A}}(x)
\neq\mathrmbfit{ext}_{\mathcal{A}}(x')$.
More strongly,
a classification
$\mathcal{A}$ is \textit{disjoint}
when
distinct types have disjoint extents:
$x \neq x'$
implies
$\mathrmbfit{ext}_{\mathcal{A}}(x)\cap\mathrmbfit{ext}_{\mathcal{A}}(x')=\emptyset$.
Even more strongly,
a classification
$\mathcal{A}$ is \textit{partitioned}
when
it is disjoint
and
all instances are classified:
$\mathrmbfit{int}_{\mathcal{A}}(y)\neq\emptyset$
equivalently
$y\in\mathrmbfit{ext}_{\mathcal{A}}(x)$
some $x \in X$
for all 
$y \in Y$.
A classification
$\mathcal{A}$ is \textit{pseudo-partitioned}
when
it is disjoint
and
all instances are classified except for one special instance ${\cdot} \in Y$:
$\mathrmbfit{int}_{\mathcal{A}}({\cdot})=\emptyset$.

\paragraph{Infomorphism.}

An infomorphism 
$\mathcal{A}_{2}={\langle{X_{2},Y_{2},\models_{\mathcal{A}_{2}}}\rangle}
\xrightleftharpoons{{\langle{f,g}\rangle}}
{\langle{X_{1},Y_{1},\models_{\mathcal{A}_{1}}}\rangle}=\mathcal{A}_{1}$
consists of a type map 
$X_{2}\xrightarrow{\;f\;}X_{1}$ 
and an instance map
$Y_{2}\xleftarrow{\;g\;}Y_{1}$, 
which satisfy the following:
\begin{center}
{\footnotesize{$\setlength{\extrarowheight}{4pt}{
\begin{array}{c}
g(y_{1}){\;\models_{\mathcal{A}_{2}}\;}x_{2}
{\;\;\;\text{\underline{iff}}\;\;\;}
y_{1}{\;\models_{\mathcal{A}_{1}}\;}x_{1}=f(x_{2})
\\
g(y_{1}){\,\in\,}\underset{\mathrmbfit{Y}_{2}(x_{2})}{\mathrmbfit{ext}_{\mathcal{A}_{2}}(x_{2})}
{\;\;\;\text{\underline{iff}}\;\;\;}
y_{1}{\,\in\,}\underset{\mathrmbfit{Y}_{1}(f(x_{2}))}{\mathrmbfit{ext}_{\mathcal{A}_{1}}(f(x_{2}))}
\\
g^{-1}(\mathrmbfit{ext}_{\mathcal{A}_{2}}(x_{2})) = \mathrmbfit{ext}_{\mathcal{A}_{1}}(f(x_{2}))
\end{array}}$}}
\end{center}

\begin{figure}
\begin{center}
{{\begin{tabular}{c
@{\hspace{33pt}{\textit{or}}\hspace{20pt}}
c}
{{\begin{tabular}{c}
\setlength{\unitlength}{0.6pt}
\begin{picture}(120,90)(0,0)
\put(2,80){\makebox(0,0){\footnotesize{$X_{2}$}}}
\put(120,80){\makebox(0,0){\footnotesize{$X_{1}$}}}
\put(2,0){\makebox(0,0){\footnotesize{$Y_{2}$}}}
\put(120,0){\makebox(0,0){\footnotesize{$Y_{1}$}}}
\put(8,40){\makebox(0,0)[l]{\scriptsize{$\models_{\mathcal{A}_{2}}$}}}
\put(130,40){\makebox(0,0)[l]{\scriptsize{$\models_{\mathcal{A}_{1}}$}}}
\put(60,94){\makebox(0,0){\scriptsize{$f$}}}
\put(62,-14){\makebox(0,0){\scriptsize{$g$}}}
\put(20,80){\vector(1,0){80}}
\put(100,0){\vector(-1,0){80}}
\put(0,65){\line(0,-1){50}}
\put(120,65){\line(0,-1){50}}
\end{picture}
\end{tabular}}}
&
{{\begin{tabular}{c}
\setlength{\unitlength}{0.6pt}
\begin{picture}(120,90)(0,0)
\put(0,80){\makebox(0,0){\footnotesize{$X_{2}$}}}
\put(120,80){\makebox(0,0){\footnotesize{$X_{1}$}}}
\put(0,0){\makebox(0,0){\footnotesize{${\wp}Y_{2}$}}}
\put(120,0){\makebox(0,0){\footnotesize{${\wp}Y_{1}$}}}
\put(8,40){\makebox(0,0)[l]{\scriptsize{$
\mathrmbfit{ext}_{\mathcal{A}_{2}}$}}}
\put(128,40){\makebox(0,0)[l]{\scriptsize{$
\mathrmbfit{ext}_{\mathcal{A}_{1}}$}}}
\put(60,94){\makebox(0,0){\scriptsize{$f$}}}
\put(62,-14){\makebox(0,0){\scriptsize{$g^{-1}$}}}
\put(20,80){\vector(1,0){80}}
\put(20,0){\vector(1,0){80}}
\put(0,65){\vector(0,-1){50}}
\put(120,65){\vector(0,-1){50}}
\end{picture}
\end{tabular}}}
\end{tabular}}}
\end{center}
\caption{Infomorphism}
\label{fig:info:mor}
\end{figure}

If two distinct source types $x_{2},x'_{2} \in X_{2}$, 
$x_{2} \not= x'_{2}$,
are mapped by the type map 
to the same target type 
\footnote{Are you producting the two types by doing this?
Remember how the sum and product of classifications is defined.}
$f(x_{2}) = f(x'_{2}) = x_{1} \in X_{1}$,
then any target instance in the extent 
$y_{1} \in \mathrmbfit{ext}_{\mathcal{A}_{1}}(x_{1})$
is mapped by the instance map 
to the extent intersection
$g(y_{1}) \in 
\mathrmbfit{ext}_{\mathcal{A}_{2}}(x_{2}){\;\cap\;}\mathrmbfit{ext}_{\mathcal{A}_{2}}(x'_{2})$.
Hence,
if extent sets are disjoint in the source classification,
then the type map must be injective.
Note however
that there is restricted map
$\mathrmbfit{ext}_{\mathcal{A}_{2}}(x_{2}){\;\times\;}\mathrmbfit{ext}_{\mathcal{A}_{2}}(x'_{2})
\xleftarrow{\;g\;}
\mathrmbfit{ext}_{\mathcal{A}_{1}}(x_{1})$.
Let $\mathrmbf{Cls}$
denote the context of classifications and infomorphisms.



%


\paragraph{Type Domains.}

%
In the \texttt{FOLE} theory of data-types \cite{kent:fole:era:found},
a classification 
$\mathcal{A} = {\langle{X,Y,\models_{\mathcal{A}}}\rangle}$
is known as a type domain.
A type domain
\cite{kent:fole:era:tbl}
is an sort-indexed collection of data types 
from which a table's tuples are chosen.
It consists of 
a set of sorts (data types) $X$,
a set of data values (instances) $Y$,
and a binary (classification) relation
{\footnotesize{$\models_{\mathcal{A}}$}}
between data values and sorts.
The extent
of any sort (data type) $x \in X$
is the subset
$\mathrmbfit{ext}_{\mathcal{A}}(x) = A_{x} 
= \{ y \in Y \mid y{\;\models_{\mathcal{A}}\;}x \}$.
Hence,
a type domain is equivalent to be a sort-indexed collection of subsets of data values
$X \xrightarrow{\;\mathcal{A}\;} {\wp}Y: x \mapsto \mathrmbfit{ext}_{\mathcal{A}}(x) = A_{x}$.
%
\footnote{Some examples of data-types 
useful in databases are as follows.
The real numbers might use 
sort symbol $\Re$ 
with extent $\{-\infty, \cdots, 0, \cdots, \infty\}$.
The alphabet might use
sort symbol $\aleph$ 
with extent $\{a,b,c, \cdots, x,y,z\}$.
Words, as a data-type, would be lists of alphabetic symbols
with
sort symbol $\aleph^{\ast}$ 
with extent $\{a,b,c, \cdots, x,y,z\}^{\ast}$
being all strings of alphabetic symbols.
The periodic table of elements might use
sort symbol $\textbf{E}$
with extent $\{\text{H, He, Li, He,} \cdots\text{, Hs,Mt}\}$.
Of course,
chemical elements can also be regarded as 
\underline{entity} types in a database 
with various properties
such as
name, 
symbol, 
atomic number, 
atomic mass, 
density, 
melting point, 
boiling point, etc.}
%
The list classification 
$\mathrmbf{List}(\mathcal{A}) = {\langle{\mathrmbf{List}(X),\mathrmbf{List}(Y),\models_{\mathrmbf{List}(\mathcal{A})}}\rangle}$
has $X$-signatures as types and
$Y$-tuples as instances,
with classification by common arity and universal $\mathcal{A}$-classification:
a $Y$-tuple ${\langle{J,t}\rangle}$ 
is classified by 
an $X$-signature ${\langle{I,s}\rangle}$ 
when
$J = I$ and
$t_{k} \models_{\mathcal{A}} s_{k}$
for all $k \in J = I$.
%
\footnote{In particular,
when $I = 1$ is a singleton,
an $X$-signature ${\langle{1,s}\rangle}$ 
is the same as a sort  
$s({\cdot)} = x \in X$,
a $Y$-tuple ${\langle{1,t}\rangle}$ 
is the same as a data value  
$t({\cdot}) = y \in Y$,
and
$\mathrmbfit{tup}(1,s,\mathcal{A})
=\mathrmbfit{tup}_{\mathcal{A}}(1,s)
=\mathrmbfit{ext}_{\mathrmbf{List}(\mathcal{A})}(1,s) 
=\mathcal{A}_{x}$.}
%
In the \texttt{FOLE} theory of data-types,
an infomorphism 
$\mathcal{A}'\xrightleftharpoons{{\langle{f,g}\rangle}}\mathcal{A}$
is known as a type domain morphism,
and consists of
a sort function $X'\xrightarrow{\;f\;}X$ and
a data value function $Y'\xleftarrow{\;g\;}Y$
that satisfy the condition
{\footnotesize{
$g(y){\;\models_{\mathcal{A}'}\;}x'$
\underline{iff}
$y{\;\models_{\mathcal{A}}\;}f(x')$
}\normalsize}
%
for any source sort $x'{\,\in\,}X'$ 
and target data value $y{\,\in\,}Y$.
%


\paragraph{Lax Classification.}

From an extensional point-of-view,
a classification 
$\mathcal{A} = {\langle{X,Y,\models_{\mathcal{A}}}\rangle}$ 
is a set-valued function
$X \xrightarrow[\mathrmbfit{Y}]{\mathrmbfit{ext}_{\mathcal{A}}} {\wp}Y$;
that is,
an indexed collection of (sub)sets
$\bigl\{ \mathrmbfit{ext}_{\mathcal{A}}(x)=\mathrmbfit{Y}(x) \mid x \in X \bigr\}$.
In a lax classification we omit the global set of data values $Y$.
\footnote{We can always define a global set of instances; 
for example,
$Y = \bigcup_{x \in X} \mathrmbfit{Y}(x)$,
$Y = \coprod_{x \in X} \mathrmbfit{Y}(x)$,
etc.}
Hence, 
a (lax) classification is just the pair
$\mathcal{A} = {\langle{X,\mathrmbfit{Y}}\rangle}$
consisting of a set
$X = \mathrmbfit{typ}(\mathcal{A})$
of objects called types that classify the instances, and
a collection of indexed subsets
$\bigl\{ \mathrmbfit{Y}(x) \mid x \in X \bigr\}$,
where
$\mathrmbfit{Y}(x) = \mathrmbfit{ext}_{\mathcal{A}}(x)$ is the set 
of instances (tokens) classified by the type $x \in X$.
More briefly,
a (lax) classification 
is 
a $\mathrmbf{Set}$-valued diagram
$X\xrightarrow{\,\mathrmbfit{Y}\;}\mathrmbf{Set}$.

\comment{
\paragraph{Lax Infomorphism.}
%
A lax infomorphism 
$\mathcal{A}_{2}={\langle{X_{2},\mathrmbfit{Y}_{2}}\rangle}
\xrightleftharpoons{{\langle{f,\gamma}\rangle}}
{\langle{X_{1},\mathrmbfit{Y}_{1}}\rangle}=\mathcal{A}_{1}$
between two lax classifications 
$\mathcal{A}_{2}$ and $\mathcal{A}_{1}$
consists of a type map 
$X_{2}\xrightarrow{\;f\;}X_{1}$ 
and an instance bridge
$\mathrmbfit{Y}_{2}\xLeftarrow{\;\,\gamma\;}r{\;\circ\;}\mathrmbfit{Y}_{1}$
consisting of an indexed collection of instance functions
$\bigl\{
{\mathrmbfit{Y}_{2}}(x_{2})
\xleftarrow{\gamma_{x_{2}}}{\mathrmbfit{Y}_{1}}(f(x_{2}))
\mid x_{2} \in X_{2}
\bigr\}$.
}

\paragraph{Lax Infomorphism.}

%
Tbl.\;\ref{tbl:lax:info} motivates the definition of a lax infomorphism.
In a lax infomorphism
$\mathcal{A}_{2} = {\langle{X_{2},\mathrmbfit{Y}_{2}}\rangle} 
\xleftharpoondown{{\langle{f,\gamma}\rangle}}
{\langle{X_{1},\mathrmbfit{Y}_{1}}\rangle} = 
\mathcal{A}_{1}$
the 
instance function
$Y_{2}\xleftarrow{\;g\,}Y_{1}$
is replaced by the collection of instance functions
$\bigl\{
\mathrmbfit{Y}_{2}(x_{2})
\xleftarrow{\gamma_{r_{2}}}
\mathrmbfit{Y}_{1}
(f(x_{2}))
\mid x_{2} \in X_{2}
\bigr\}$.
Hence,
an infomorphism 
is 
a $\mathrmbf{Set}$-valued instance bridge 
$\mathrmbfit{Y}_{2}\xLeftarrow{\;\,\gamma\;}f{\,\circ\,}\mathrmbfit{Y}_{1}$.
Let 
$\mathring{\mathrmbf{Cls}}$
denote the context of lax classifications and lax infomorphisms.
There is 
a passage
$\mathrmbf{Cls} \xrightarrow{\mathrmbfit{lax}} \dot{\mathrmbf{Cls}}$.

%
\begin{table}
\begin{center}
{\fbox{\setlength{\extrarowheight}{2pt}
$\begin{array}{r@{\hspace{5pt}}c@{\hspace{5pt}}l}
\multicolumn{3}{c}{\underline{\text{infomorphism}}}\rule[-10pt]{0pt}{10pt}
\\
\mathcal{A}_{2} = {\langle{X_{2},Y_{2},\models_{\mathcal{A}_{2}}}\rangle} 
&
\xrightleftharpoons{{\langle{f,g}\rangle}}
&
{\langle{X_{1},Y_{1},\models_{\mathcal{A}_{1}}}\rangle} = \mathcal{A}_{1}
\\
g(y_{1}){\;\models_{\mathcal{A}_{2}}\;}x_{2}
&
\;\Longleftrightarrow\;
&
y_{1}{\;\models_{\mathcal{A}_{1}}\;}f(x_{2})
\\
g(y_{1}){\;\in\;}\mathrmbfit{ext}_{\mathcal{A}_{2}}(x_{2})
&
\;\Longleftrightarrow\;
&
y_{1}{\;\in\;}\mathrmbfit{ext}_{\mathcal{A}_{1}}(f(x_{2}))
\\
g^{-1}(\mathrmbfit{ext}_{\mathcal{A}_{2}}(x_{2}))
&
\;=\;
&
\mathrmbfit{ext}_{\mathcal{A}_{1}}(f(x_{2}))
\\
...............................................&\text{\underline{implies}}&...............................................
\\
g^{-1}(\overset{\mathrmbfit{Y}_{2}}{\mathrmbfit{ext}_{\mathcal{A}_{2}}}(x_{2}))
&
\;\supseteq\;
&
\overset{\mathrmbfit{Y}_{1}}{\mathrmbfit{ext}_{\mathcal{A}_{1}}}(f(x_{2}))
\\
g(x_{1}){\;\in\;}\overset{\mathrmbfit{Y}_{2}}{\mathrmbfit{ext}_{\mathcal{A}_{2}}}(x_{2})
&
\;\Longleftarrow\;
&
x_{1}{\;\in\;}\overset{\mathrmbfit{Y}_{1}}{\mathrmbfit{ext}_{\mathcal{A}_{1}}}(f(x_{2}))
\\
g(y_{1}){\;\models_{\mathcal{A}_{2}}\;}x_{2}
&
\;\Longleftarrow\;
&
y_{1}{\;\models_{\mathcal{A}_{1}}\;}f(x_{2})
\\
\mathcal{A}_{2} = {\langle{X_{2},\mathrmbfit{Y}_{2}}\rangle} 
&
\xleftharpoondown{{\langle{f,\gamma}\rangle}}
&
{\langle{X_{1},\mathrmbfit{Y}_{1}}\rangle} = \mathcal{A}_{1}
\\
\multicolumn{3}{c}{\text{\underline{lax infomorphism}}}\rule[6pt]{0pt}{10pt}
\\&&
\end{array}$}}
\end{center}
\caption{Lax Infomorphisms}
\label{tbl:lax:info}
\end{table}
%


%
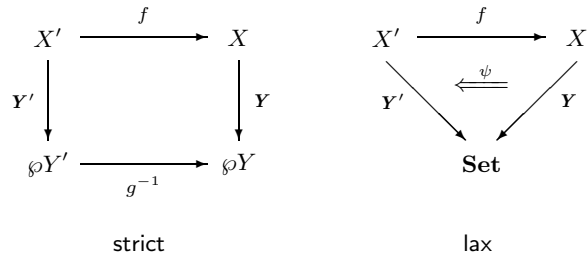
\begin{figure}
\begin{center}
{{\begin{tabular}{c@{\hspace{50pt}}c}
{{\begin{tabular}{c}
\setlength{\unitlength}{0.6pt}
\begin{picture}(120,90)(0,0)
\put(0,80){\makebox(0,0){\footnotesize{$X'$}}}
\put(120,80){\makebox(0,0){\footnotesize{$X$}}}
\put(0,0){\makebox(0,0){\footnotesize{${\wp}Y'$}}}
\put(120,0){\makebox(0,0){\footnotesize{${\wp}Y$}}}
\put(-8,40){\makebox(0,0)[r]{\scriptsize{$
\mathrmbfit{Y}'$}}}
\put(128,40){\makebox(0,0)[l]{\scriptsize{$
\mathrmbfit{Y}$}}}
\put(60,94){\makebox(0,0){\scriptsize{$f$}}}
\put(60,-14){\makebox(0,0){\scriptsize{$g^{-1}$}}}
\put(20,80){\vector(1,0){80}}
\put(20,0){\vector(1,0){80}}
\put(0,65){\vector(0,-1){50}}
\put(120,65){\vector(0,-1){50}}
\end{picture}
\end{tabular}}}
&
{{\begin{tabular}{c}
\setlength{\unitlength}{0.6pt}
\begin{picture}(120,90)(0,0)
\put(0,80){\makebox(0,0){\footnotesize{$X'$}}}
\put(120,80){\makebox(0,0){\footnotesize{$X$}}}
\put(60,0){\makebox(0,0){\footnotesize{$\mathrmbf{Set}$}}}
\put(12,40){\makebox(0,0)[r]{\scriptsize{$\mathrmbfit{Y}'$}}}
\put(108,40){\makebox(0,0)[l]{\scriptsize{$\mathrmbfit{Y}$}}}
\put(60,94){\makebox(0,0){\scriptsize{$f$}}}
%
\put(20,80){\vector(1,0){80}}
\put(60,50){\makebox(0,0){{$\xLeftarrow{\;\;\psi\;\;}$}}}
\put(0,65){\vector(1,-1){50}}
\put(120,65){\vector(-1,-1){50}}
\end{picture}
\end{tabular}}}
\\&\\&\\
\textsf{strict}
&
\textsf{lax}
\end{tabular}}}
\end{center}
\caption{Lax Infomorphism}
\label{fig:lax:info:mor}
\end{figure}







\end{document}